\documentclass[11pt]{article}

\PassOptionsToPackage{nosort}{cite}

\usepackage{graphicx}
\usepackage{float}
\usepackage{mathtools}
\usepackage{scalerel}
\usepackage{amssymb}
\usepackage{amsfonts}
\usepackage{amsthm}
\usepackage{dsfont}
\usepackage{mathrsfs}
\usepackage{enumitem}
\usepackage[T1]{fontenc}
\usepackage{chet}
\usepackage{multirow}
\usepackage[font=small,format=hang,labelfont={sf,bf}]{caption}
\usepackage{placeins}
\usepackage[capitalise]{cleveref}
\usepackage{longtable}
\usepackage{slashed}
\usepackage{pgfplots}
\usepackage{makecell}

\DeclareMathOperator{\Rea}{Re}
\DeclareMathOperator{\Ima}{Im}
\DeclareMathOperator{\Tr}{Tr}

\usepgfplotslibrary{colorbrewer}
\pgfplotsset{width=10cm,compat=1.9}

\newtheorem*{theorem}{Theorem}
\newcommand{\overbar}[1]{\mkern 1.5mu\overline{\mkern-1.5mu#1\mkern-1.5mu}\mkern 1.5mu}

\newcommand{\lsp}{\hspace{1pt}}
\newcommand{\llsp}{\hspace{0.5pt}}
\newcommand{\lnsp}{\hspace{-1pt}}

\newcommand{\veps}{\varepsilon}

\renewcommand{\geq}{\geqslant}
\renewcommand{\leq}{\leqslant}

\allowdisplaybreaks

\definecolor{darkblue}{rgb}{0.1,0.1,0.7}
\hypersetup{colorlinks,
           linkcolor={darkblue},
           citecolor={darkblue},
           urlcolor={darkblue},
           pdftitle={Scalar-Fermion Fixed Points in the Epsilon Expansion},
           pdfdisplaydoctitle
}

\title{Scalar-Fermion Fixed Points in the $\varepsilon$ Expansion}

\author{William H.\ Pannell and Andreas Stergiou\emails{(\href{mailto:william.pannell@kcl.ac.uk}{william.pannell}, \href{mailto:andreas.stergiou@kcl.ac.uk}{andreas.stergiou})@kcl.ac.uk}}

\affiliation{Department of Mathematics, King's College London, Strand, London WC2R 2LS, United Kingdom}

\abstract{The one-loop beta functions for systems of $N_s$ scalars and $N_f$ fermions interacting via a general potential are analysed as tensorial equations in $4-\varepsilon$ dimensions. Two distinct bounds on combinations of invariants constructed from the couplings are derived and, subject to an assumption, are used to prove that at one-loop order the anomalous dimensions of the elementary fields are universally restricted by $\gamma_\phi\leq\frac{1}{2}N_s\lsp\varepsilon$ and $\gamma_\psi\leq N_s\lsp \varepsilon$. For each root of the Yukawa beta function there is a number of roots of the quartic beta function, giving rise to the concept of `levels' of fixed points in scalar-fermion theories. It is proven that if a stable fixed point exists within a certain level, then it is the only such fixed point at that level. Solving the beta function equations, both analytically and numerically, for low numbers of scalars and fermions, well-known and novel fixed points are found and their stability properties are examined. While a number of fixed points saturate one out of the two bounds, only one fixed point is found which saturates both of them.}

\date{May 2023}

\begin{document}

\maketitle

\toc

\section{Introduction}
There exist a number of well-known models involving scalars and fermions coupled together by a Yukawa-type interaction, which have a variety of applications both in modelling condensed matter systems and in providing toy models for behaviour such as chiral symmetry breaking in quantum field theory (QFT). Two of the more famous are the Gross--Neveu--Yukawa (GNY) model~\cite{Gross:1974jv, Zinn-Justin:1991ksq}, and the Nambu--Jona-Lasinio--Yukawa (NJLY) model~\cite{Nambu:1961tp, Nambu:1961fr, Zinn-Justin:1991ksq}, ultraviolet (UV) completions of two-dimensional purely fermionic theories first introduced to study continuous and discrete chiral symmetry breaking. Recently, it has been suggested that these models may have a rich structure in 3d, where for certain numbers of fermions, including non-integer values, they flow to an infrared (IR) conformal field theory (CFT) with emergent supersymmetry~\cite{Fei:2016sgs, Liendo:2021wpo}.

In the Wilsonian paradigm, one can consider CFTs to arise as theories which do not flow under renormalisation. These fixed points will then act as the UV and IR limits of QFT flows through theory space. The $\varepsilon$ expansion, first introduced by Wilson and Fisher in the 1970s~\cite{Wilson:1971dc}, has been a remarkably powerful tool for calculating and understanding the properties of these fixed points. Picking out the $\varepsilon$ poles arising from Feynman diagrams via counterterms, the beta functions become polynomials in the interaction couplings, so that the problem of finding fixed points becomes a purely algebraic one. Though this expansion  relies on taking $\varepsilon$ to be a small parameter, it is still able to provide meaningful information in the limit $\varepsilon\rightarrow 1$. Thus, one can derive information about CFTs in three dimensions by considering theories close to four dimensions. While it may be difficult to determine rigorously how data from the four dimensional theory carries over to three dimensions, it is believed that, except for special bifurcation points, the existence of a fixed point of the one-loop beta function is not affected by higher order terms when taking $\varepsilon\rightarrow1$.

To derive fixed points one typically considers a specific system, such as the GNY or NJLY model, where the interactions preserve a definite symmetry group $G$. One then examines how this system behaves under symmetry-preserving deformations. The group $G$ will restrict the number of couplings and the allowed form of the beta functions, generally simplifying the algebraic equations. The restriction provided by the symmetry group $G$ is not a general one, and greatly limits the number and type of fixed points which will be found. Recently, there has been interest in extending beyond this line of thinking in order to systematically examine the space of all fixed points. Such a search requires a selection of the field content of the theory, and starting with a potential that contains all possible classically marginal terms that can be written down from the basic fields. Previously, this search has been undertaken in the context of purely (massless) scalar theories\cite{Osborn:2017ucf,Osborn:2020cnf}, where the potential is the completely unconstrained
\begin{equation}
    V(\phi)=\tfrac{1}{4!}\lambda_{ijkl}\phi_i\phi_j\phi_k\phi_l\,,
    \label{generalscalarpotential}
\end{equation}
for $N_s$ scalar fields $\phi_i$ with $i=1,\ldots,N_s$ and $\lambda_{ijkl}$ a symmetric tensor. Beyond the well-known models with simple symmetry groups and only a few contributions to the scalar potential, there exist a plethora of more complicated fixed points with reduced symmetry.

Considering the importance of theories containing both scalars and fermions, it is natural to generalise (\ref{generalscalarpotential}) to include $N_f$ Weyl fermions as
\begin{equation}
    V(\phi,\psi,\bar{\psi})=\tfrac{1}{4!}\lambda_{ijkl}\phi_i\phi_j\phi_k\phi_l+\big(\tfrac{1}{2}y_{iab}\phi_i\psi_a\psi_b+\,\text{h.c.}\big)\,,
\end{equation} 
where $y_{iab}$ is symmetric in $a$ and $b$. For specific choices of $N_s$, $\lambda_{ijkl}$ and $y_{iab}$, one is able to recover the GNY or NJLY model. However, this framework is presumably general enough to, as in the scalar case, produce novel fixed points. Recently, the beta functions for $\lambda_{ijkl}$ and $y_{iab}$ have been written down to three-loop order and examined in a supersymmetric context~\cite{Jack:2023zjt}. As we are more concerned with the simple existence of fixed points in a simple brute-force search, the one-loop beta function will be sufficient for our purposes.

In this context, we undertake here a general analysis of fixed points in scalar-fermion theories. After discussing the beta functions in such theories, we analyse the stability matrix which determines repulsive and attractive directions close to a fixed point in a flow towards the IR. We show that there is always a positive eigenvalue equal to $\varepsilon$ and that whenever the global symmetry of the free action is broken there are zero eigenvalues at the fixed point. These zero eigenvalues do not correspond to exactly marginal operators but are rather zero due to the associated deformations being total derivatives~\cite{Rychkov:2018vya}. Such deformations arise from broken currents of the free theory.

In scalar models, it is well-known that if with a given set of deformations there exists a renormalisation group (RG) stable fixed point, in the sense that all directions around it are attractive in a flow towards the IR, then it is the only RG stable fixed point \cite{Michel:1983in}. In the case of scalar-fermion models, the one-loop Yukawa beta function does not depend on the quartic couplings and so it can be solved independently of the quartic beta function. For every root of the Yukawa beta function there will be multiple roots of the quartic beta function and thus there are `levels' of fixed points (with each level corresponding to a different root of the Yukawa beta function). We prove, under fairly general assumptions, that if a stable fixed point exists within a level, then it is the only fixed point with that property. Note that different levels can have different stable fixed points.

We also derive upper bounds on linear combinations of group invariants that can be used to characterise scalar-fermion fixed points. This extends results found for purely scalar theories~\cite{Brezin:1973jt, Rychkov:2018vya, Osborn:2020cnf, Hogervorst:2020gtc}. The Yukawa beta function leads to a bound on a combination of invariants constructed from the Yukawa couplings alone, while the quartic beta function leads to a bound on a combination of mixed invariants involving both Yukawa and quartic couplings. There is a way to combine these bounds so that a linear combination of the tensorial norms $||\lambda||^2=\lambda_{ijkl}\lambda_{ijkl}$ and $||y_iy^*{\hspace{-5pt}}_i\llsp||^2=\Tr(y_{i}y^*{\hspace{-5pt}}_{i}y_{j}y^*{\!\!\lnsp}_{j})$ is bounded above:
\begin{equation}\label{bintro}
    ||\lambda||^2-6\,||y_iy^*{\hspace{-5pt}}_i\llsp||^2\leq\tfrac18 N_s\lsp\varepsilon^2\,.
\end{equation}
The minus sign in \eqref{bintro} is ultimately due to the minus sign in front of the contribution of purely Yukawa terms to the beta function of $\lambda$. The bound \eqref{bintro} indicates that the allowed region in which unitary fixed points may be found is bounded by a hyperbola in the $||\lambda||$-$||y_iy^*{\hspace{-5pt}}_i\llsp||$ plane. 

These bounds are shown to lead to upper bounds on the one-loop anomalous dimensions of the scalar and fermion fields if a certain combination quartic in Yukawa couplings is non-negative. This follows from the fact that anomalous dimensions are given simply by combinations of couplings evaluated at fixed points. In fact, at leading order only the Yukawa couplings at the fixed point are enough to determine the anomalous dimensions. Subject to an unproven assumption that we have numerically checked in a variety of cases, we find that 
\begin{equation}
    \gamma_\phi\leq \tfrac12\lsp N_s\lsp\varepsilon\,,\qquad \gamma_\psi\leq N_s\lsp\varepsilon\,.
\end{equation}

As we have already mentioned, fixed points obtained in the $\varepsilon$ expansion are usually limited by overarching assumptions of symmetry. Here we relax these assumptions and seek all fixed points in a brute force search for small $N_s, N_f$. We perform analytic searches but also numerical ones when we can no longer find all roots of the beta functions analytically. These searches are not systematic enough to claim that we have found all possible fixed points, but they are broad enough to indicate the size of the space of fixed points and to locate well-known theories. While the potentials of fixed points found in purely scalar theories are necessarily bounded below \cite{Brezin:1973jt, Rychkov:2018vya}, this is not the case for scalar-fermion fixed points. In this work we will not filter fixed points based on the criterion of stability of the scalar potential.

The outline of the paper is as follows. In section \ref{generalsec} we describe the scalar-fermion systems, their beta functions and stability matrices, and then derive general bounds that the beta function equations place on certain combinations of invariants constructed from the couplings. In this section we also examine how to transfer between Weyl fermions in four dimensions and real two-component Majorana fermions in three dimensions, and prove that for each Yukawa-coupling solution there exists at most one stable fixed point. In section \ref{sec:wellknown} we survey some well-known fixed points involving scalars and fermions. In section \ref{analyticsec} we begin our search of the space of fixed points by directly solving the beta function equations for small numbers of scalars and fermions, and verify the consistency with bounds from the previous section. Section \ref{numericsec} then extends the search beyond the scope of simple analytic techniques by applying numerical methods to the problem. We first verify that the numerical results are consistent with the analytical results of the previous section, and then set out towards the unknown. Finally, we conclude in section \ref{conc}.

\section{General results}\label{generalsec}
In $d=4$ we may start with the general Lagrangian
\eqn{\mathscr{L}=\tfrac12\lsp\partial^\mu\phi_i\lsp\partial_\mu\phi_i+i\bar{\psi}^a\bar{\sigma}^\mu\partial_\mu\psi_a+\tfrac{1}{4!}\lambda_{ijkl}\lsp\phi_i\phi_j\phi_k\phi_l+(\tfrac12\lsp  y_{iab}\lsp\phi_i\psi_a\psi_b+\text{h.c.})\,,}[basicLag]
where $\phi_i$, $i=1,\ldots,N_s$, are real scalar fields, $\psi_a$, $a=1,\ldots,N_f$, are two-component Weyl fermions, $\bar{\psi}=\psi^\dagger$ and $\bar{\sigma}^\mu=(\mathds{1}_{2\times 2},-\vec{\sigma})$ with $\vec{\sigma}$ containing the three Pauli matrices, $\vec{\sigma}=(\sigma^1,\sigma^2,\sigma^3)$. Repeated indices are always summed over the values they take. The tensor $\lambda_{ijkl}$ is fully symmetric and thus has $\frac{1}{4!}N_s(N_s+1)(N_s+2)(N_s+3)$ independent real components, while $y_{iab}$ is a complex tensor symmetric in the two fermionic flavour indices and thus has $\frac12 N_sN_f(N_f+1)$ independent complex components.

Two-component spinors may be related to four-component spinors in the standard way; see e.g.\ \cite{Dreiner:2008tw}. First of all, the Lagrangian \basicLag is equivalent to a Lagrangian built out of four-component Majorana spinors 
\eqn{\Psi^M_a=\begin{pmatrix}\psi_{a\lsp\alpha} \\ \bar{\psi}^{a\lsp\dot{\alpha}}\end{pmatrix}\,,\qquad
\overbar{\Psi}^M_a=\begin{pmatrix}\psi_{a}^{\lsp\alpha} & \bar{\psi}^{a}_{\lsp\dot{\alpha}}\end{pmatrix}\,.}[]
This Lagrangian takes the form
\eqn{\mathscr{L}=\tfrac12\lsp\partial^\mu\phi_i\lsp\partial_\mu\phi_i+\tfrac12 i\lsp\overbar{\Psi}^M_a\gamma^\mu\partial_\mu\Psi^M_a+\tfrac{1}{4!}\lambda_{ijkl}\lsp\phi_i\phi_j\phi_k\phi_l+\tfrac12\lsp  y^M_{iab}\lsp\phi_i\overbar{\Psi}^M_a\Psi^M_b+\tfrac12 i\lsp \hat{y}^M_{iab}\lsp\phi_i\overbar{\Psi}^M_a\gamma^5\Psi^M_b\,,}[basicLagM]
where now $y^M_{iab}$ and $\hat{y}^M_{iab}$ are real tensors symmetric in the fermionic flavour indices. Comparing with \basicLag we find
\eqn{y_{iab}=y^M_{iab}-i\lsp \hat{y}^M_{iab}\,.}[]
For the gamma matrices we use the chiral representation and $\overbar{\Psi}=\Psi^\dagger\gamma^0$. 

Dirac spinors can also be used. To define a Dirac spinor one needs two different two-component spinors. Let us choose
\eqn{\Psi_a=\begin{pmatrix}\eta_{a\lsp\alpha} \\ \bar{\chi}^{a\lsp\dot{\alpha}}\end{pmatrix}\,,\qquad
\overbar{\Psi}_a=\begin{pmatrix}\chi_{a}^{\lsp\alpha} & \bar{\eta}^{a}_{\lsp\dot{\alpha}}\end{pmatrix}\,,}[diracfers]
with $a=1,\ldots,N_f$ counting Dirac spinors so that we need an even number of two-component spinors. Without using explicit projectors to two-component spinors, we may write the Lagrangian
\eqn{\mathscr{L}_D=\tfrac12\lsp\partial^\mu\phi_i\lsp\partial_\mu\phi_i+i\lsp\overbar{\Psi}_a\gamma^\mu\partial_\mu\Psi_a+\tfrac{1}{4!}\lambda_{ijkl}\lsp\phi_i\phi_j\phi_k\phi_l+y^D_{iab}\lsp\phi_i\overbar{\Psi}_a\Psi_b+i\lsp \hat{y}^D_{iab}\lsp\phi_i\overbar{\Psi}_a\gamma^5\Psi_b\,,}[basicLagD]
where $y^D_{iab}$ and $\hat{y}^D_{iab}$ are complex tensors Hermitian in the fermionic flavour indices, i.e.\ $y^{D\lsp *}_{iba}=y^D_{iab}$ and $\hat{y}^{D\lsp *}_{iba}=\hat{y}^D_{iab}$. Beginning with the Lagrangian (\ref{basicLagD}), we may write each Dirac fermion as in \eqref{diracfers} for Weyl fermions $\eta_a$ and $\chi_a$, and re-express this Lagrangian in the form of (\ref{basicLag}) for $2N_f$ Weyl fermions. If we define the Weyl fermions $\psi_A, A=1,\ldots,2N_f$, such that
\begin{equation}
    \eta_a=\psi_{a}\,,\qquad \chi_a=\psi_{N_f+a}\,,
\end{equation}
we find that the Yukawa coupling of the Weyl fermions can be written in block-matrix form like
\begin{equation}
    y_{iAB}=\begin{pmatrix}\textbf{0} & \frac{1}{2}(y_i^{D})^*-\frac{i}{2}(\hat{y}_i^D)^* \\ 
    \frac{1}{2}y_i^D-\frac{i}{2}\hat{y}_i^D & \textbf{0}\end{pmatrix}\,.
    \label{eq:diractoweyl}
\end{equation}
As $y^D_i, \hat{y}^D_i$ are Hermitian matrices, this coupling will be, as expected, symmetric but not real.

The Lagrangian \basicLagD is not as general as \basicLag, but widely analysed models like the GNY and NJLY models are most easily described as special cases of \basicLagD. Let us remark here that parity is broken in \basicLagD. To restore it one may consider scalars $\phi_i$ and pseudoscalars $\phi'_p$ and allow only the couplings $y^D_{iab}\lsp\phi_i\overbar{\Psi}_a\Psi_b+i\lsp \hat{y}^D_{pab}\lsp\phi'_p\overbar{\Psi}_a\gamma^5\Psi^M_b$. Time reversal is also broken in \basicLagD.

In our treatment we will use \basicLag in $d=4-\veps$ dimensions, with the eventual goal of taking the limit $\veps\to1$. A proper treatment of fermions in this excursion across dimensions from 4 to 3 is, to our knowledge, unknown. In the literature, it is common to start with Dirac fermions in $d=4$, which have 8 real degrees of freedom, and assert that proper Majorana fermions in $d=3$ can be obtained, with two real degrees of freedom each, essentially emerging by some decomposition of the original Dirac spinor. This gives rise to 4 Majorana fermions in $d=3$ per Dirac fermion in $d=4$. Our Lagrangian \basicLag contains two-component Weyl fermions in $d=4$, which can be equivalently thought of as $d=4$ Majorana fermions. Despite the lack of precise mappings, we treat each of our two-component Weyl fermions in $d=4$ as containing 2 Majorana fermions in $d=3$. However, there is evidence that the story is not so simple when one wants to include global flavour symmetries. For instance, it is believed that though the GNY model, (\ref{basicLagD}) for $N_s=1$ with $\hat{y}=0$ and $y_{ab}=y\lsp\delta_{ab}$, has $U(N_f)$ symmetry in 4d, the Dirac fermions will break apart such that the theory has $O(2N_f)^2\rtimes\mathbb{Z}_2$ symmetry rather than the $O(4N_f)$ symmetry one might naively expect\cite{Jack:2023zjt, Erramilli:2022kgp, Kubota:2001kk}.

\subsection{Beta functions for scalar-fermion theories}
The beta functions for the coupling constants $\lambda$ and $y$ for scalar-fermion theories are well known~\cite{Jack:2023zjt, Jack:1983sk, Jack:1984vj, Machacek:1983fi, Machacek:1984zw}, and depend only on the type of fermion that is used in that the number of independent components contained by the spinor can appear as a prefactor in certain terms. For the Weyl Lagrangian in (\ref{basicLag}), the beta functions are given at one loop by\foot{Here and hereafter we rescale the quartic coupling by $16\pi^2$ and the Yukawa coupling by $4\pi$.}
\begin{align}
    \beta^{\lambda,W}_{ijkl}&=-\varepsilon\lsp\lambda_{ijkl}+S_{3,ijkl}\lsp\lambda_{ijmn}\lambda_{mnkl}-4\lsp S_{6,ijkl}\lsp y_{iab}\lsp y^*{\!\!\!}_{jbc}\lsp y_{kcd}\lsp y^*{\!\!\!}_{lda}\nonumber\\
    &\hspace{5cm}+\tfrac{1}{2}S_{4,ijkl}\lsp(y_{iab}\lsp y^*{\!\!\!}_{mba}+y_{mab}\lsp y^*{\!\!\!}_{iba})\lambda_{mjkl}\,,\label{eq:betalambdaweyl}\\
    \beta^{y,W}_{iab}&=-\tfrac{1}{2}\varepsilon\lsp y_{iab}+\tfrac{1}{2}(y_{jac}\lsp y^*{\!\!\!}_{jcd}\lsp y_{idb}+y_{iac}\lsp y^*{\!\!\!}_{jcd}\lsp y_{jdb})+2\lsp y_{jac}\lsp y^*{\hspace{-5pt}}_{icd}\lsp y_{jdb}+\tfrac12(y_{icd}\lsp y^*{\!\!\!}_{jdc}+y_{jcd}\lsp y^*{\hspace{-5pt}}_{idc})\lsp y_{jab}\,,\label{eq:betayweyl}
\end{align}
where $S_{n,ijkl}$ sums over the $n$ inequivalent index permutations of the indices $i,j,k,l$. Operationally, fixed points may be calculated using the above beta functions, at which point one takes $N_f$ to be the number of desired Majorana fermions in 3d divided by the appropriate factor of two. The difficulties associated with attributing the global symmetry group in three dimensions to the starting symmetry group in four dimensions makes this method subtle. 

There is another, perhaps more naive way of continuing across dimensions, by attempting to directly guess the form the Yukawa interaction ought to take in 3d. If we divide the $N_f$ 4d Weyl fermions into $2N_f$ 3d Majorana fermions as
\begin{equation}
    \psi_{a}=\zeta_{a}+i\lsp\xi_{a}\,,
\end{equation}
the Yukawa interaction in (\ref{basicLag}) becomes
\begin{equation}
\phi_i\big(2\Rea(y_{iab})(\zeta_a\zeta_b-\xi_a\xi_b)+2\Ima(y_{iab})(\zeta_a\xi_b+\zeta_b\xi_a)\big)=g_{iAB}\phi_i\theta_A\theta_B\,,
\label{eq:weyltomajorana}
\end{equation}
where $A,B=1,\ldots,2N_f$, $\zeta_a=\theta_a$ and $\xi_a=\theta_{N_f+a}$. In four dimensions, the Lorentz symmetry requires that the coefficient of the $\zeta\zeta$ and $\xi\xi$ terms be equal, up to a sign, so that the theory's degrees of freedom can be written in terms of Weyl, or Dirac spinors. In three dimensions, however, this is no longer the case, and relaxing this condition we conjecture that these theories have the more general Yukawa interaction $y_{iAB}\phi_i\theta_A\theta_B$, where now $y_{iAB}$ is simply a real tensor symmetric with respect to its $A,B$ indices. Importantly, this includes both three-dimensional theories found by taking the $\varepsilon\rightarrow1$ limit of (\ref{basicLag}), and three-dimensional theories which cannot be lifted into four-dimensional theories in a Lorentz-invariant manner. 

The beta functions for the two-component Majorana fermions can be written down based on this conjecture. As the basic scalar-fermion interaction vertex still takes the standard Yukawa form, the coupling $y_{iab}$ will be renormalised with precisely the same diagrams as those needed to renormalise $g_{iab}$. There are only two distinctions which we must account for: first that the fermion loops will only see half of the degrees of freedom in the 3d Majorana case as they do in the Weyl case, and second that our conjecture posits that $y_{iab}$ is real, so that both $g_{iab}$ and $g^*{\hspace{-5pt}}_{iab}$ must be replaced by $y_{iab}$ wherever they appear in a diagram. Implementing these rules in (\ref{eq:betalambdaweyl}) and (\ref{eq:betayweyl}) yields the new beta functions
\begin{align}
    \beta^{\lambda,M}_{ijkl}&=-\varepsilon\lsp\lambda_{ijkl}+S_{3,ijkl}\lsp\lambda_{ijmn}\lambda_{mnkl}-4\lsp S_{3,ijkl}\lsp y_{iab}\lsp y_{jbc}\lsp y_{kcd}\lsp y_{lda}\nonumber\\
    &\hspace{7cm}+\tfrac{1}{2}S_{4,ijkl}\lsp y_{iab}\lsp y_{mba}\lambda_{mjkl}\,,\\
    \beta^{y,M}_{iab}&=-\tfrac{1}{2}\varepsilon\lsp y_{iab}+\tfrac{1}{2}(y_{jac}\lsp y_{jcd}\lsp y_{idb}+y_{iac}\lsp y_{jcd}\lsp y_{jdb})+2\lsp y_{jac}\lsp y_{icd}\lsp y_{jdb}+\tfrac12\lsp y_{jcd}\lsp y_{icd}\lsp y_{jab}\,.
\end{align}
To notationally unify the Weyl and 3d Majorana beta functions we can define the symmetric $O(N_s)$ tensors
\begin{equation}
\begin{aligned}
Z_{ij}&=\tfrac{1}{2}\alpha\big(y_{iab}\lsp y^{*}{\hspace{-5pt}}_{jba}+y^{*}{\hspace{-5pt}}_{iab}\lsp y_{jba}\big)=\tfrac{1}{2}\alpha\Tr(y_iy^*{\!\!\!}_j+y^*{\!\!\!}_iy_j)\,,\\
X_{ijkl}&=\tfrac{1}{2}\alpha\lsp S_{6,ijkl}\lsp y_{iab}\lsp y^{*}{\hspace{-5pt}}_{jbc}\lsp y_{kcd}\lsp y^{*}{\hspace{-5pt}}_{lda}=\tfrac{1}{2}\alpha\lsp S_{6,ijkl}\Tr(y_{i}\lsp y^{*}{\hspace{-5pt}}_{j}\lsp y_{k}\lsp y^{*}{\hspace{-5pt}}_{l})\,,
\label{eq:zweyl}
\end{aligned}
\end{equation}
where $\alpha=1$ for two-component Majorana fermions and $\alpha=2$ for Weyl fermions. One can view $\alpha$ as counting the number of 3d Majorana fermions the 4d fermions will produce when we take the $\varepsilon\rightarrow1$ limit. The factors of $\alpha$ account for the different number of independent components in the distinct types of fermions. In the two-component Majorana case, one must simply remember that the coupling $y_{iab}$ is real to remove the complex conjugation.

Tensors we have encountered so far are not irreducible representations (irreps) of $O(N_s)$, and can be decomposed as
\begin{equation}
\begin{aligned}
    Z_{ij}&=Z_0\lsp\delta_{ij}+Z_{2,ij}\,, \\
    X_{ijkl}&=X_0\lsp S_{3,ijkl}\lsp \delta_{ij}\delta_{kl}+S_{6,ijlk}\lsp X_{2,ij}\lsp \delta_{kl}+X_{4,ijkl}\,,\\
    \lambda_{ijkl}&=d_0\lsp S_{3,ijkl}\lsp\delta_{ij}\delta_{kl}+S_{6,ijlk}\lsp d_{2,ij}\delta_{kl}+d_{4,ijkl}\,,
    \label{eq:tensordecomposition}
    \end{aligned}
\end{equation}
where $Z_2$, $X_2$, $X_4$, $d_2$ and $d_4$ are all symmetric, traceless tensors. In terms of these tensors, the beta functions simplify in both cases to
\begin{align}
    \beta^{\lambda}_{ijkl}&=-\varepsilon\lsp \lambda_{ijkl}+S_{3,ijkl}\lambda_{ijmn}\lambda_{mnkl}-4\lsp X_{ijkl}
    +\tfrac{1}{2}S_{4,ijkl}\lsp Z_{im}\lambda_{mjkl}\,,
    \label{eq:betalambda}\\
    \beta^{y}_{iab}&=-\tfrac{1}{2}\varepsilon \lsp y_{iab}+\tfrac{1}{2}\big(y_{jac}\lsp y^{*}{\hspace{-5pt}}_{jcd}\lsp y_{idb} + y_{iac}\lsp y^{*}{\hspace{-5pt}}_{jcd}\lsp y_{jdb} + 4\lsp y_{jac}\lsp y^{*}{\hspace{-5pt}}_{icd}\lsp y_{jdb}+Z_{ij}\lsp y_{jab}\big)\,,
    \label{eq:betay}
\end{align}
where again in the Majorana case we simply use the fact that the couplings $y$ are real to remove the complex conjugation. The anomalous dimensions of the fields are also known to one loop order\cite{Machacek:1983tz}, and are given by the eigenvalues of the matrices
\begin{equation}
\gamma^\phi_{ij}=\tfrac{1}{4}\alpha\lsp Z_{ij}\,, \qquad
\gamma^\psi_{ab}=\tfrac{1}{2}\lsp y_{iac}\lsp y^{*}{\hspace{-5pt}}_{icb}\,.
    \label{eq:anomalousdimensions}
\end{equation}
These dimensions will be useful in determining how the global symmetry of the free theory is broken at different fixed points.

\subsubsection{Doubling effect} \label{doubling}
Solving for fixed points in some situations produces pairs of distinct fixed points related by a rescaling of the scalar couplings. Suppose we have some solution $(\lambda,y)$ to the beta function equations and consider the new set of couplings $(\lambda'=a\lambda,y)$. This rescaled set of couplings may now also be a fixed point of the same beta function equations, i.e.\ it may set (\ref{eq:betalambda}) and (\ref{eq:betay}) to zero. One could absorb the rescaling in a field redefinition, but that would not preserve the form of the kinetic term in the Lagrangian. Thus, demanding a canonical kinetic term shows that these fixed points will be distinct physical theories.

Let us now assume that $(a\lambda,y)$ also describes a fixed point. Then, starting from the fixed point $(\lambda,y)$ we have 
\eqn{-\varepsilon\lsp a_0+a_1+2S-4N_s(N_s+2)X_0+\tfrac{1}{2}Z_{ij}\lambda_{ijkk}=0\,,}[eqnonresc]
where $S=\lambda_{ijkl}\lambda_{ijkl}$, while for the rescaled fixed point $(a\lambda,y)$ we have
\eqn{-\varepsilon\lsp a\lsp a_0+a^2\lsp a_1+2\lsp a^2\lsp S-4N_s(N_s+2)X_0+\tfrac{1}{2}Z_{ij}\lsp a\lsp \lambda_{ijkk}=0\,.}[eqresc]
These two equations can be simultaneously satisfied for a unique value of $a\ne1$. Subtracting $1/a$ times \eqresc from \eqnonresc yields the condition which $a$ must satisfy:
\begin{equation}
    (1-a)(a_1+2S)-\bigg(1-\frac{1}{a}\bigg)4N_s(N_s+2)X_0=0\,.
    \label{eq:acondition}
\end{equation}
Assuming $a\ne1$ and using the fact that $X_0$ only vanishes when all fermions are free, the last expression becomes
\begin{equation}
    \frac{a_1+2S}{4N_s(N_s+2)X_0}=-\frac{1}{a}.
    \label{eq:doubling}
\end{equation}
This is the doubling effect: for certain $(\lambda,y)$ fixed points, there exist a physically inequivalent $(a\lambda,y)$ fixed points with $a$ given by \eqref{eq:doubling}. As everything is positive on the left-hand side of \eqref{eq:doubling}, it follows that $a<0$.

Notice that as $\lambda=0$ is not a solution to $\beta^\lambda=0$ for non-zero $y$, $a\neq1$ always exists for fixed points with interacting fermions. However, this does not a priori guarantee the existence of a doubled fixed point satisfying the same beta function equations. For that to be the case, it must hold that
\begin{equation}
    S_{3,ijkl}\lambda_{ijmn}\lambda_{mnkl}=-\frac{4}{a}X_{ijkl}
    \label{eq:betadoubled}
\end{equation}
for all $i,j,k$ and $l$, which is a non-trivial condition. 

Note that if $N_s=1$ the indices in \eqref{eq:betadoubled} disappear so that this equation trivially reduces to (\ref{eq:doubling}), and thus every fixed point is doubled in that case. When the fermions are all free or in a purely scalar theory the only solution to \eqref{eq:acondition} is $a=1$ and thus there is no doubling effect in those cases.

\subsection{The stability matrix}
\subsubsection{Eigenvalues and their interpretation}\label{stabmatrixeigenvalues}
It is crucial to notice that because $\beta^y$ is independent of $\lambda$ for both Weyl and two-component Majorana fermions, in both cases the stability matrix will necessarily take the schematic form
\begin{equation}
    S_{IJ}=\begin{pmatrix}H_{ij} & B_{ib}\\ 0 & C_{ab}\end{pmatrix}
    \label{eq:stabilitymatrixform}
\end{equation}
for some $H$, $B$ and $C$, where $I$, $i$ and $a$ represent generalised indices. Here, $C$ will be the stability matrix of the Yukawa couplings. Because of the zero matrix in the lower-left block, eigenvectors of $H$ will remain eigenvectors of the whole stability matrix as long as we simply append $\frac12 N_sN_f(N_f+1)$ zeros to the end of the vector.

Many of the results concerning eigenvalues of this stability matrix carry over virtually unchanged from those found for the purely scalar case\cite{Osborn:2020cnf}. The stability matrix will always have an eigenvector with eigenvalue $\kappa=\varepsilon$ which is related to the value of the couplings by
\begin{equation}
    v_I=\begin{pmatrix}\lambda_i\\ \frac{1}{2}y_a\end{pmatrix}.
    \label{eq:eig}
\end{equation}
To show that (\ref{eq:eig}) is indeed an eigenvector of $S_{IJ}$ we must first show that $y_a$ is an eigenvector of $C_{ab}$ with unit eigenvalue. To see that this is the case, notice that when taking derivatives of $\beta^y$ with respect to $y_{iab}$ one will get Kronecker deltas in the $O(N_s)$ indices and Kronecker deltas in the $O(N_f)$ indices symmetrised over $a$ and $b$, so that upon contracting with $y_{iab}$ one gets back the same terms as in $\beta^y$ but now with a prefactor equal to the number of $y$'s in each term. That is,
\begin{equation}
    \frac{d\beta^{y}_{iab}}{dy_{jdc}}y_{jdc}=-\tfrac12\varepsilon\lsp y_{iab}+\tfrac{3}{2}\big(y_{jac}y_{jcd}y_{idb}+y_{iac}y_{jcd}y_{jdb}+4\lsp y_{jac}y_{icd}y_{jdb}+Z_{ij}y_{jab}\big)\,,
\end{equation}
which, using the fact that $\beta^y=0$ at the fixed point, reduces immediately to
\begin{equation}
    \frac{d\beta^{y}_{iab}}{dy_{jdc}}y_{jdc}=\varepsilon\lsp y_{iab}\,.
\end{equation}
Thus, we have that 
\begin{equation}
    S_{aJ}v_J=\varepsilon\lsp v_a\,.
\end{equation}
To demonstrate that $v_I$ is an eigenvector, all that is left is then to show that $S_{iJ}v_J=\varepsilon\lsp v_i$. To see this, notice that as before contracting $S$ with $\lambda$ or $y$ will simply introduce factors equal to the power of the couplings that appear in $\beta^\lambda$ so that
\begin{equation}
\begin{aligned}
    S_{ijklmnpq}\lambda_{mnpq}&=-\varepsilon\lsp\lambda_{ijkl}+2\lsp S_{3,ijkl}\lambda_{ijmn}\lambda_{klmn}+\tfrac{1}{2} S_{4,ijkl}Z_{im}\lambda_{mjkl}\,,\\
    \tfrac{1}{2} S_{ijklmab}y_{mab}&=-8\lsp S_{3,ijkl}\text{Tr}(y_iy_jy_ky_l)+\tfrac{1}{2} S_{4,ijkl}Z_{im}\lambda_{mjkl}\,.
\end{aligned}
\end{equation}
Summing these and using the fact that $\beta^\lambda=0$ at the fixed point we arrive at the eigenvalue equation
\begin{equation}
    S_{iJ}v_J=\varepsilon\lsp v_i\,.
\end{equation}
Hence, $v_I$ is an eigenvector of the stability matrix with eigenvalue $\varepsilon$. Though the proof was performed for the two-component Majorana theory, it is also valid for the Weyl theory with only minor modifications to account for the fact that one now has both $y$ and $y^*$ couplings. In the case of decoupled theories, there will be multiple $\kappa=\varepsilon$ eigenvectors, one for each of the decoupled theories.

There may also exist Yukawa eigenvectors of the stability matrix with $\kappa=0$, which are related to the breaking of the $O(N_s)$ and $O(N_f)$ symmetries. For the $O(N_f)$ symmetry, define the vector
\begin{equation}
    v_I=\begin{pmatrix}0 \\ \omega_{ac}y_{icb}+\omega_{bc}y_{iac}\end{pmatrix},
\end{equation}
where $\omega_{ab}$ is antisymmetric and thus exists in the Lie Algebra of $O(N_f)$. Then, notice that because of the antisymmetry of $\omega$, terms arising from contracted fermionic indices in the beta functions will cancel pairwise, for
\begin{equation}
    \frac{dy_{iab}}{dy_{kde}}y_{jbc}v_{kde}=\omega_{ae}y_{ieb}y_{jbc}+\omega_{ce}y_{iab}y_{jbe}\,.
\end{equation}
Consequently, to determine the action of $S_{IJ}$ on $v_I$ we simply need to attach a factor of $\omega$ to each free fermionic index. We see immediately that
\begin{equation}
    S_{iJ}v_J=0\,,
\end{equation}
as all of the fermionic indices in $S_{iJ}$ are summed over in traces. For $\beta^y$, we get two separate $\omega$ terms
\begin{equation}
\begin{aligned}
    S_{iabjcd}v_{jcd}&=\omega_{ae}\big(-\tfrac{1}{2}\varepsilon\lsp y_{ieb}+\tfrac{1}{2}(y_{jec}y_{jcd}y_{idb}+y_{iec}y_{jcd}y_{jdb}+4\lsp y_{jec}y_{icd}y_{jdb}+Z_{ij}y_{jeb})\big)\\&\quad+\omega_{be}\big(-\tfrac{1}{2}\varepsilon\lsp y_{iae}+\tfrac{1}{2}(y_{jac}y_{jcd}y_{ide}+y_{iac}y_{jcd}y_{jde}+4\lsp y_{jac}y_{icd}y_{jde}+Z_{ij}y_{jae})\big),
    \end{aligned}
\end{equation}
which one can see vanishes when applying the fact that $\beta^y=0$ twice. If $y_{iab}$ is $O(N_f)$ invariant, then $\omega_{ac}y_{icb}+\omega_{bc}y_{iac}=0$ for all $\omega$ by virtue of an infinitesimal $O(N_f)$ rotation. Thus, non-zero $v_I$ indicate the presence of broken symmetry, with certain $O(N_f)$ generators acting non-trivially on the $y$'s. Similarly, one can define the vector
\begin{equation}
    u_I=\begin{pmatrix}\omega_{im}\lambda_{mjkl}+\omega_{jm}\lambda_{imkl}+\omega_{km}\lambda_{ijml}+\omega_{lm}\lambda_{ijkm} \\ \omega_{ij}y_{jab}\end{pmatrix},
\end{equation}
where now the antisymmetric $\omega_{ij}$ act as generators of the scalar $O(N_s)$ symmetry. One can see that the antisymmetry of $\omega$ will again, up to applications of the beta function equations, lead to vanishing eigenvalue in exactly the same way as for $v_I$. That is,
\begin{equation}
    S_{IJ}u_{J}=0\,.
\end{equation}
Here, non-zero $u_I$ demonstrate the presence of broken $O(N_s)$ generators, so that by examining the form of the eigenvectors with $\kappa=0$ at a fixed point one can see how the $O(N_s)$ and $O(N_f)$ symmetries are broken and thus determine the remaining symmetry group.

\subsubsection{RG stability} \label{stabilitysection}
The fact that $\beta^y$ has no dependence on the scalar coupling $\lambda$, at least at one loop, means that to this order the problem of finding fixed points in scalar-fermion theories splits into two separate problems: first, one can find all possible fixed points in $y$ by considering $\beta^y$, and then one can insert these $y$'s into $\beta^\lambda$ as constants to find the corresponding values of $\lambda$. In practice this splitting was not needed to determine the fixed points with numerics, but this observation allows one to observe quite simply some crucial analytic properties of the system.

In the purely scalar case, one can capture the behaviour of the beta functions and the stability matrix by writing down the $O(N_s)$ scalar $A$-function for the system\cite{Rychkov:2018vya},
\begin{equation}
    A=-\tfrac{1}{2}\varepsilon\lambda_{ijkl}\lambda_{ijkl}+\lambda_{ijkl}\lambda_{klmn}\lambda_{mnij}\,,
    \label{eq:scalara}
\end{equation}
which is defined such that
\begin{equation}
    \beta^\lambda_{ijkl}=\frac{\delta A}{\delta \lambda_{ijkl}}\,.
\end{equation}
The stability matrix is then given by the Hessian of the $A$-function in coupling space. Determining fixed points and their stability then simply becomes an exercise in extremising the $A$-function and classifying the various extrema. The uniqueness of stable fixed points in the purely scalar system is thus guaranteed by a theorem of Michel which states that the $A$-function given in (\ref{eq:scalara}) has a unique minimum\cite{Michel:1983in}.

Let us then suppose that we have previously determined the possible values of $y$ by solving the $\beta^y=0$ equation, and wish now to find the allowed $\lambda$'s for one specific $y$ of our choice. As we have, in a sense, now just simply added constants to the purely scalar beta function, we can again write down an $A$-function which generalises (\ref{eq:scalara}) to the case of non-zero fermions. Here, it is important to note that now the Hessian of $A$ does not describe the full stability matrix for the theory, but only the part for perturbations with purely scalar operators, labeled by $H$ in (\ref{eq:stabilitymatrixform}). The $A$-function which generates (\ref{eq:betalambda}) is
\begin{equation}
    A_y(\lambda)=-\tfrac{1}{2}\varepsilon\lsp(\lambda,\lambda)+(\lambda,\lambda,\lambda)-4\lsp(X,\lambda)+(\lambda,\lambda)_Z\,,
    \label{eq:afunctiondef}
\end{equation}
where we have defined the notation
\begin{equation}
    \begin{aligned}
        (a,b)&=a_{ijkl}b_{ijkl}\,,\\
        (a,b,c)&=a_{ijkl}b_{klmn}c_{mnij}\,,\\
        (a,b)_Z&=Z_{ij}a_{iklm}b_{jklm}\,.
    \end{aligned}
\end{equation}
The only difference here between the Weyl and two-component Majorana theories is which value of $\alpha$ one uses in the definition of $X_{ijkl}$ and $Z_{ij}$, (\ref{eq:zweyl}). As this is the only distinction, we will not specify which theory we are working with, as the following theorem and proof are identical in both cases. If $\Lambda_y(H)$ is the set of couplings $(\lambda,y)$ invariant under some subgroup $H<O(N_s)$ of scalar rotations, then we then find the following generalisation of Michel's theorem:
\begin{theorem}
If $(\lambda_1,y)$ and $(\lambda_2,y)$ are two non-identical fixed points which lie in $\Lambda_y(H)$ such that $A_y(\lambda_1)>A_y(\lambda_2)$, then $(\lambda_1,y)$ cannot be stable against perturbations within $\Lambda_y(H)$ at one loop.\footnote{Note that our proof does not apply if $A_y(\lambda_1)=A_y(\lambda_2)$.}
\end{theorem}
\begin{proof}
At a fixed point, we will have that
\begin{equation}
    (\beta^\lambda,\lambda)=0=-\varepsilon\lsp(\lambda,\lambda)+3\lsp(\lambda,\lambda,\lambda)-4\lsp(X,\lambda)+2\lsp(\lambda,\lambda)_Z\,,
\end{equation}
so that the $A$-function will thus be reduced to
\begin{equation}
    A_*(\lambda)=-\tfrac{1}{6}\varepsilon\lsp(\lambda,\lambda)-\tfrac{8}{3}\lsp(X,\lambda)+\tfrac{1}{3}\lsp(\lambda,\lambda)_Z\,.
    \label{eq:afixed}
\end{equation}
Then, suppose that we have two distinct fixed points $(\lambda_1,y)$ and $(\lambda_2,y)$ with $A_y(\lambda_2)< A_y(\lambda_1)$, where by distinct we simply mean that $\lambda_1$ and $\lambda_2$ are unequal tensors. The fixed point $(\lambda_1,y)$ will only be stable against perturbations towards $(\lambda_2,y)$ if the Hessian, $H$, of $A(\lambda_1)$ is positive-semidefinite in the $\lambda_1$-$\lambda_2$ plane. Expanding $A(\lambda_1+s\lambda_1+t\lambda_2)$ and collecting terms of orders $s^2$, $st$ and $t^2$, we find that at $\lambda_1$ this Hessian is
\begin{equation}
    H=\begin{pmatrix}
    A'_1-8\lsp(X,\lambda_1) & B-8\lsp(X,\lambda_1) \\ B-8\lsp(X,\lambda_1) & 2\lsp B-A'_2-8\lsp(X,\lambda_1)
    \end{pmatrix},
    \label{eq:hessian}
\end{equation}
where we have defined the $O(N_s)$ invariants
\begin{equation}
    A'_i=-6A_y(\lambda_i)\,,\qquad
    B=\varepsilon\lsp(\lambda_1,\lambda_2)+8\lsp(X,\lambda_1+\lambda_2)-2\lsp(\lambda_1,\lambda_2)_Z\,.
\end{equation}
Notice that we can separate $H$ as $H=J+K$ for
\begin{equation}
    J=\begin{pmatrix} A'_1 & B \\ B & 2B-A'_2
    \end{pmatrix}, \qquad
    K=-8\lsp(X,\lambda_1)\begin{pmatrix}1 & 1 \\ 1 & 1\end{pmatrix}.
\end{equation}
As the matrix $K$ has an eigenvector $v=\begin{pmatrix}
    1 & -1
\end{pmatrix}^T$ with eigenvalue $0$, if $H$ is to be positive-semidefinite we must have that
\begin{equation}
    v^TJv=A_1'-A_2'\geq0\,.
\end{equation}
However, as $A_y(\lambda_2)< A_y(\lambda_1)$, $A_1'<A_2'$, so that $v^TJv<0$ and thus $H$ is not positive semi-definite. Hence, $\lambda_1$ will not be a stable fixed point. 
\end{proof}

It is important to notice that, unlike in the purely scalar case, the mixed scalar-fermion theory may have multiple stable fixed points, as the theorem only applies when the fixed points have the same solution for the Yukawa couplings. The different allowed values of $y$ give us different `levels' of fixed points, each of which may have its own stable fixed point.

\subsection{Bounds}
The beta function equations allow one to derive various inequalities that must be obeyed by $O(N_s)\times O(N_f)$ invariants at the fixed points. We define the invariants
\begin{equation}
    \begin{aligned}
        a_0&=N_s(N_s+2)d_0=\lambda_{iijj}\,,\\
        a_1&=\lambda_{iimn}\lambda_{jjmn}\,,\\
        a_2&=||(N_s+4)d_2||^2=a_1-\frac{1}{N_s}a_0^2\,,\\
        S&=\lambda_{ijkl}\lambda_{ijkl}=||\lambda||^2\,,\\
        b_0&=Z_{ii}=N_sZ_0\,,\qquad \tilde{b}_0=a_0Z_0\,,\\
        b_1&=||Z_2||^2\,,\\
        b_2&=||(N_s+4)d_2+Z_2||^2=a_2+b_1+2(N_s+4)d_2\cdot Z_2\,,\\
        b_3&=X_{iijj}=N_s(N_s+2)X_0\,,\\
        Y&=||y_iy^*{\hspace{-5pt}}_i\llsp ||^2=\Tr(y_{i}y^*{\hspace{-5pt}}_{i}y_{j}y^*{\!\!\lnsp}_{j})\,.
    \end{aligned}
    \label{eq:invariants}
\end{equation}
We can then derive relations between these invariants by considering contractions of the beta functions (\ref{eq:betalambda}) and (\ref{eq:betay}). By considering an expansion of $\lambda$ and $y$ in terms of $\varepsilon$ we can ignore factors of $\varepsilon$ at leading order.

Considering first $\beta^\lambda$, we find that by contracting indices in pairs in this beta function we derive the following expression which must hold at all fixed points:
\begin{equation}
\beta^{\lambda}_{iijj}=-\varepsilon\lsp a_0+a_1+2S-4\lsp b_3+2\lsp \tilde{b}_0+2(N_s+4)d_2\cdot Z_2=0\,.
\end{equation} 
Expanding the inner product in the invariant $b_2$ and then replacing $a_2$ in terms of $a_1$ and $a_0$ we find
\begin{equation}
    a_1+2(N_s+4)d_2\cdot Z_2=b_2-b_1+\frac{1}{N_s}a_0^2\,,
\end{equation}
so that we have
\begin{equation}
    -\frac{1}{N_s}(a_0^2-N_s\lsp\varepsilon\lsp a_0)=2 S+b_2-b_1-4\lsp b_3+2\lsp \tilde{b}_0\,.
    \label{eq:four}
\end{equation}
Completing the square on the left-hand side this immediately gives us the inequality
\begin{equation}
    S+\tfrac{1}{2}\lsp b_2-\tfrac{1}{2}\lsp b_1-2\lsp b_3+\tilde{b}_0\leq\tfrac{1}{8}N_s\lsp\varepsilon^2\,.
    \label{eq:three}
\end{equation}

Similarly, at a fixed point we must have $\beta^{y}_{iab}y^*_{iba}=0$, so that
\begin{equation}
    \frac{1}{\alpha}\varepsilon\lsp Z_{ii}=\frac{1}{\alpha}\varepsilon\lsp b_0=2\lsp Y+4\Tr( y_{j}y^*{\hspace{-5pt}}_{i}y_{j}y^*{\hspace{-5pt}}_{i})+\frac{1}{\alpha}Z_{ij}Z_{ij}=2\lsp Y+4\Tr( y_{j}y^*{\hspace{-5pt}}_{i}y_{j}y^*{\hspace{-5pt}}_{i})+\frac{1}{\alpha N_s}b_0^2+\frac{1}{\alpha}b_1\,.
    \label{eq:one}
\end{equation}
To express the remaining trace term in terms of more familiar invariants, notice that
\begin{equation}
    \frac{4}{\alpha}b_3=4\Tr(y_{i}y^*{\hspace{-5pt}}_{i}y_{j}y^*{\!\!\!}_{j}+y_{j}y^*{\hspace{-5pt}}_{i}y_{j}y^*{\hspace{-5pt}}_{i}+y_{i}y^*{\!\!\!}_{j}y_{j}y^*{\hspace{-5pt}}_{i})=8\lsp Y+4\Tr(y_{j}y^*{\hspace{-5pt}}_{i}y_{j}y^*{\hspace{-5pt}}_{i})\,,
\end{equation}
where the factor of $\alpha$ in this expression comes from the fact that there are twice as many terms needed to symmetrise $\text{Tr}(y_iy^*{\!\!\!}_jy_{k}y^*{\!\!\!}_l)$ in the Weyl case. Plugging this in, we see that
\begin{equation}
    -\frac{1}{\alpha N_s}(b_0^2-N_s\lsp\varepsilon\lsp b_0)=-6\lsp Y+\frac{4}{\alpha}\lsp b_3+\frac{1}{\alpha}\lsp b_1\,.
    \label{eq:five}
\end{equation}
We now see that by completing the square on the left-hand side we can turn this into the bound
\begin{equation}
    b_1+4\lsp b_3-6\lsp\alpha\lsp Y=\frac{1}{4}N_s\lsp\varepsilon^2-\frac{1}{N_s}\Big(b_0-\frac{1}{2}N_s\lsp\varepsilon\Big)^2\leq\frac{1}{4}N_s\lsp\varepsilon^2\,.
    \label{eq:two}
\end{equation}
We may obtain a single inequality relating both fermionic and scalar invariants by eliminating $b_1$ and $b_3$ from  (\ref{eq:two}) and (\ref{eq:three}):
\begin{equation}\label{simplebound}
S+\tfrac{1}{2}\lsp b_2+\tilde{b}_0-3\lsp \alpha\lsp Y\leq
    \tfrac{1}{4}N_s\lsp\varepsilon^2\,.
\end{equation}

It would be simpler to have an inequality relating $S$, the norm of $\lambda$, and $Y$, the norm of $y_iy^*{\hspace{-5pt}}_i$, only. In fact, there is a way to remove the possibly negative $\tilde{b}_0$ term from \eqref{simplebound}. To see this, notice that the $\tilde{b}_0$ term in (\ref{eq:four}) means that there is an additional term proportional to $a_0$, and thus an additional way to complete the square. Moving this term over to the other side we have that
\begin{equation}
    S+\frac{1}{2}\lsp b_2-\frac{1}{2}\lsp b_1-2\lsp b_3=-\frac{1}{2N_s}(a_0^2-(N_s\lsp\varepsilon-2\lsp b_0)a_0)=\frac{1}{2N_s}\bigg(b_0-\frac{1}{2}N_s\lsp\varepsilon\bigg)^2-\frac{1}{2N_s}\bigg(a_0+b_0-\frac{1}{2}N_s\lsp\varepsilon\bigg)^2\,,
    \label{eq:newsquare}
\end{equation}
and combining with \eqref{eq:two} we get
\begin{equation}\label{curve1}
    S+\frac12\lsp b_2-3\lsp\alpha\lsp Y=\frac{1}{8}N_s\lsp\varepsilon^2-\frac{1}{2N_s}\Big(a_0+b_0-\frac{1}{2}N_s\lsp\varepsilon\Big)^2\,,
\end{equation}
which implies
\begin{equation}\label{simplestbound}
    S+\tfrac12\lsp b_2-3\lsp\alpha\lsp Y\leq \tfrac{1}{8}N_s\lsp\varepsilon^2\,.
\end{equation}
This bound directly generalises the bound obtained in~\cite{Osborn:2020cnf, Hogervorst:2020gtc}, which reads
\begin{equation}\label{scalarbound}
    S+\tfrac12\lsp a_2\leq\tfrac18 N_s\lsp\varepsilon^2\qquad \text{(scalars only)}\,.
\end{equation}
As we observe, the only difference with the scalar case is to replace $a_2$ with $b_2-6\lsp\alpha\lsp Y$. If we have only scalars or the fermions are decoupled, then $b_2$ is equal to $a_2$ and $Y$ is zero, so that \eqref{simplestbound} reduces to \eqref{scalarbound} in those cases. Note that $b_2$ can be omitted from the left-hand side of \eqref{simplestbound} since it is positive, which gives the bound \eqref{bintro} discussed in the introduction (when $\alpha=2$).

Additionally, using \eqref{eq:newsquare} and \eqref{eq:two} (or, equivalently, substituting \eqref{eq:tensordecomposition} in \eqref{eq:betalambda} and using \eqref{eq:five}) we find
\begin{equation}\label{a0b0bound}
    (a_0+b_0)\lsp\varepsilon-\frac{1}{N_s}\bigg(\frac{N_s+8}{N_s+2}a_0^2+b_0^2\bigg)-\frac{2}{N_s}a_0b_0+6\lsp\alpha\lsp Y\geq0\,,
\end{equation}
which becomes the bound $a_0\lsp\varepsilon-\frac{N_s+8}{N_s(N_s+2)}a_0^2\geq 0\Rightarrow 0\leq a_0\leq \frac{N_s(N_s+2)}{N_s+8}\lsp\varepsilon$ when there are no fermions. Obviously from \eqref{a0b0bound} we also have that 
\begin{equation}\label{a0b0bound2}
    (a_0+b_0)\lsp\varepsilon-\frac{2}{N_s}a_0b_0+6\lsp\alpha\lsp Y\geq0\,.
\end{equation}

We would also like to point out here that with
\begin{equation}
    T=S+\tfrac12\lsp b_2-3\lsp\alpha\lsp Y\,,\qquad R=\frac{1}{\sqrt{2N_s}}\Big(a_0+b_0-\frac12N_s\lsp\varepsilon\Big)\,,
\end{equation}
we may write \eqref{curve1} as
\begin{equation}
    R^2+T=\tfrac18N_s\lsp\varepsilon^2\,,
    \label{parabolicbound}
\end{equation}
which defines a parabola on which all fixed points lie. Since $b_2\geq 0$, if we define $T'=T-\frac12 b_2\leq T$, then the points in the $R$-$T'$ plane corresponding to fixed points will satisfy $R^2+T'\leq \tfrac18N_s\lsp\varepsilon^2$. Note that fixed points that saturate the bound on $R^2+T'$ do not necessarily have a unique quadratic invariant.

\subsubsection{Anomalous dimensions}
From the beta functions (\ref{eq:betalambda}) and (\ref{eq:betay}) one can also immediately write down bounds limiting the allowed anomalous dimensions of the renormalised fields $\phi$ and $\psi$.\footnote{While these bounds apply even to non-interacting fixed points, in that case they merely express the fact that purely scalar theories have a vanishing one loop anomalous dimension. For purely scalar theories one can bound the two-loop anomalous dimension instead by $\gamma_\phi\leq\frac{1}{8}N_s\varepsilon^2$\cite{Osborn:2020cnf}.} Consider the following combination of couplings, which must vanish at all of the fixed points:
\begin{equation}
    \alpha\lsp\beta^y_{iab}y^*_{iab}=-\frac{1}{2}\varepsilon\lsp b_0+\frac{1}{2}Z_{ij}Z_{ji}+\text{Tr}\big((y_iy^*{\!\!\!}_j+y_jy^*{\hspace{-5pt}}_i)(y_iy^*{\!\!\!}_j+y_jy^*{\hspace{-5pt}}_i)\big)-\Tr(y_iy^*{\!\!\!}_jy_jy^*{\hspace{-5pt}}_i)=0\,.
\end{equation}
Let us assume that
\begin{equation}
    \text{Tr}\big((y_iy^*{\!\!\!}_j+y_jy^*{\hspace{-5pt}}_i)(y_iy^*{\!\!\!}_j+y_jy^*{\hspace{-5pt}}_i)\big)-\Tr(y_iy^*{\!\!\!}_jy_jy^*{\hspace{-5pt}}_i)\geq0
    \label{assumption}
\end{equation}
for all fixed point solutions $y_{iab}$, so that
\begin{equation}
    -\varepsilon\lsp b_0+Z_{ij}Z_{ji}\leq 0\,.
\end{equation}
As $Z_{ij}$ is a symmetric matrix, it will always be possible to use an $O(N_s)$ field redefinition to place it in a diagonal form.\footnote{Though this transformation is always possible, we will find it more useful to leave $Z_{ij}$ general in most cases, because for any specific theory the necessary form of the rotation matrix $R_{ij}$ will not be known until a solution has been found.} Using the fact that $Z_{ij}$ is positive definite, the above inequality reduces to $\varepsilon\lsp b_0\geq b_1+\frac{1}{N_s}b_0^2$ which implies that
\begin{equation}
    0\leq b_0\leq N_s\lsp\varepsilon\,.
\end{equation}
From (\ref{eq:anomalousdimensions}) this becomes a bound on the allowed anomalous dimensions of the scalar fields $\phi_i$,
\begin{equation}
    \gamma_\phi\leq\tfrac{1}{4}\alpha\lsp N_s\lsp\varepsilon\,.
\end{equation}

Subject to the assumption \eqref{assumption}, this bound on $\gamma_\phi$ applies universally to all scalar-fermion theories with a Yukawa interaction, at least to one-loop order. This bound can also be related to a bound for the allowed anomalous dimensions of the fermion fields $\psi_a$ by noting that
\begin{equation}
    \text{Tr}\big(\gamma^\psi\big)=\frac{1}{2}\alpha\lsp b_0\leq\frac{1}{2}\alpha\lsp N_s\lsp \varepsilon\,.
\end{equation}
As the unitary bound prohibits negative anomalous dimensions, this becomes the bound
\begin{equation}
    \gamma_\psi\leq\tfrac{1}{2}\alpha\lsp N_s\lsp\varepsilon\,,
\end{equation}
which, again, is applicable to all scalar-fermion theories at leading order in perturbation theory.

These derivations rest on the assumption (\ref{assumption}), the validity of which we must now examine. When $N_s=1$, the left-hand side of \eqref{assumption} becomes $3\Tr(yy^*yy^*)$,
which is the norm of the matrix $yy^*$, and is thus strictly non-negative. For $N_f=1$, the trace will be trivial so that the left-hand side of \eqref{assumption} is equal to $3\lsp|y_i|^2|y_j|^2$,
which is again strictly non-negative. Thus, the bounds on anomalous dimensions must hold if $N_s=1$ or $N_f=1$. When $N_s,N_f>1$, the situation becomes more difficult. While we have not been able to find an analytic proof of (\ref{assumption}) for the general case, we have verified numerically that it holds at all fixed points we have found for $N_s,N_f\leq4$. It thus seems likely that (\ref{assumption}) holds for solutions to the beta functions with arbitrary $N_s$ and $N_f$, so that our bounds on anomalous dimensions become general.

\section{Some well-known models}\label{sec:wellknown}
Up until this point the discussion has been completely general, with no mention made of specifying the global symmetry group $G\leq O(N_s)\times O(N_f)$ for certain theories. While the majority of the fixed points we will find in our brute-force search will not have been specifically investigated, we will locate a number of well-known and well-explored theories. Thus, for completeness we review  here two of the better known: the Gross--Neveu--Yukawa and Nambu--Jona-Lasinio--Yukawa models.

\subsection{Gross--Neveu--Yukawa model}\label{secGNY}
First, we consider the Gross--Neveu--Yukawa (GNY) model, which contains a single real scalar and $N_D$ Dirac fermions with the Lagrangian
\begin{equation}\label{lagGNY}
    \mathscr{L}_{\text{GNY}}= \tfrac{1}{2}\partial^\mu\phi\lsp\partial_\mu\phi + i\lsp \overbar{\Psi}_a\slashed{\partial}\Psi_a+ y\lsp\phi\overbar{\Psi}_a\Psi_a +\tfrac{1}{4!}\lambda \phi^4\,,
\end{equation}
for $a=1,\ldots,N_D$. As we are ultimately interested in deriving results in 3d, we will find it useful to define $N=4N_D$, which is the number of 3d Majorana fermions we will have in the end. This theory has a well-known interacting fixed point at which (taking $\varepsilon\rightarrow1$)
\begin{equation}
    y^2=\frac{1}{N+6}\,,\qquad \lambda=\frac{\sqrt{P_N}-N+6}{6(N+6)}\,,
    \label{gnyfp}
\end{equation}
for $P_N=N^2+132N+36$. Using (\ref{eq:diractoweyl}) and then (\ref{eq:weyltomajorana}), we see that in terms of 3d Majorana fermions we will obtain the Yukawa interaction
\begin{equation}
    y_{AB}=\begin{pmatrix}
         & & y\lsp\delta_{ab} & \textbf{0}\\
         &  & \textbf{0} &-y\lsp\delta_{ab}\\ y\lsp\delta_{ab} & \textbf{0} &  &  \\
      \textbf{0} & -y\lsp\delta_{ab} &   & 
    \end{pmatrix}\,,
\end{equation}
where $A,B=1,\ldots,N$. This squares to a multiple of the identity,
\begin{equation}
    y_{AC}\lsp y_{CB}=y^2\delta_{AB}\,,
\end{equation}
so that the invariants at this fixed point will be very simple:
\begin{equation}
\begin{aligned}
    b_0&=Ny^2=\frac{N}{N+6}\,,\\
    Y&=Ny^4=\frac{N}{(N+6)^2}\,,\\
    S&=\lambda^2=\frac{\big(\sqrt{P_N}-N+6\big)^2}{36(N+6)^2}\,,
\end{aligned}
\end{equation}
where we have implicitly assumed that $N_D$ is a positive integer so that $N\geq4$. In \cref{analyticsec} we will derive all of the interacting fixed points for $N_s=1$, $N_f=4$, and will thus find it useful to note that plugging in $N=4$ to the above equations yields
\begin{equation}
    S=\frac{(1+\sqrt{145})^2}{900}\,,\qquad Y=\frac{1}{25}\,,
\end{equation}
so that we will be able to identify the GNY fixed point for one Dirac fermion.

It has been suggested that one can derive information about theories with $N<4$ Majorana fermions in three dimensions by beginning with a non-integer number of Dirac fermions in four dimensions. Specifically, a connection has been drawn between the GNY model with $N_D=1/4$ and a minimal $\mathcal{N}=1$ supersymmetric theory of a real superfield in three dimensions with a cubic superpotential\cite{Fei:2016sgs}. In agreement with (\ref{gnyfp}) for $N=1$, this theory is expected to have a fixed point at
\begin{equation}
    y^2=\tfrac{1}{7}\,,\qquad\lambda=\tfrac{3}{7}\,.
\end{equation}
Plugging in $N=1$ to the above equations for $S$ and $Y$ yields
\begin{equation}
    S=\tfrac{9}{49}\,,\qquad Y=\tfrac{1}{49}\,.
\end{equation}

\subsection{Nambu--Jona-Lasinio--Yukawa model}\label{secNJLY}
The Nambu--Jona-Lasinio--Yukawa (NJLY) model provides a generalisation of the GNY model to two scalar fields $\phi_1$ and $\phi_2$ (though really $\phi_2$ is a pseudoscalar) and $N_D$ Dirac fermions $\Psi_a, a=1,\ldots,N_D$. Its Lagrangian is
\begin{equation}\label{lagNJLY}
    \mathscr{L}_{\text{NJLY}}= \tfrac{1}{2}\partial^\mu\phi_1\lsp\partial_\mu\phi_1 + \tfrac{1}{2}\partial^\mu\phi_2\lsp\partial_\mu\phi_2 + i\lsp \overbar{\Psi}_a\slashed{\partial}\Psi_a + y\lsp\overbar{\Psi}_a(\phi_1+i\lsp\gamma^5\phi_2)\Psi_a+\tfrac{1}{4!}\lambda(\phi_1^2+\phi_2^2)^2\,,
\end{equation}
which enjoys not only the $U(N_D)$ flavour symmetry of the fermions but also a $U(1)$ chiral symmetry which rotates both the scalars and the fermions. There is again a non-trivial interacting fixed point with, similarly defining $N=4N_D$,
\begin{equation}
    y^2=\frac{1}{N+4}\,,\qquad \lambda=\frac{\sqrt{R_N}-N+4}{20(N+4)}\,,
\end{equation}
where $R_N=N^2+152\lsp N+16$. When expanding the Dirac fermions in terms of 3d Majorana fermions, we see that the $\phi_1$ Yukawa interaction is simply that of the GNY model, so that
\begin{equation}
    y_{1AB}=\begin{pmatrix}
         & & y\lsp\delta_{ab} & \textbf{0}\\
         &  & \textbf{0} &-y\lsp\delta_{ab}\\ y\lsp\delta_{ab} & \textbf{0} &  &  \\
      \textbf{0} & -y\lsp\delta_{ab} &   & 
    \end{pmatrix}\,,
\end{equation}
where $A,B=1,\ldots,N$. Using the Weyl representation of the gamma matrices, the $\phi_2$ Yukawa interaction becomes, in terms of Weyl fermions,
\begin{equation}
    i\phi_2\overbar{\Psi}_a\gamma^5\Psi_a=i\phi_2(-\chi_a\eta_a+\overbar{\chi}^a\overbar{\eta}^a)\,.
\end{equation}
Applying (\ref{eq:diractoweyl}) and then (\ref{eq:weyltomajorana}), we have
\begin{equation}
    y_{2AB}=\begin{pmatrix}
        & & \textbf{0} & -y\lsp\delta_{ab}\\
        & & -y\lsp\delta_{ab} & \textbf{0}\\
        \textbf{0} & -y\lsp\delta_{ab} & &\\
        -y\lsp\delta_{ab} & \textbf{0} & & 
    \end{pmatrix}\,.
\end{equation}
As before, both of these matrices square to the same multiple of the identity
\begin{equation}
    y_{1AC}\lsp y_{1CB}=y_{2AC}\lsp y_{2CB}=y^2\delta_{AB}\,,
\end{equation}
yielding simple forms for the invariants we use to characterise fixed points:
\begin{equation}
\begin{aligned}
    b_0&=2Ny^2=\frac{2N}{N+4}\,,\\
    Y&=4Ny^4=\frac{4N}{(N+4)^2}\,,\\
    S&=24\lambda^2=\frac{3\big(\sqrt{R_N}-N+4\big)^2}{50(N+4)^2}\,,
\end{aligned}
\end{equation}
where we again have assumed that $N_D\in\mathbb{N}$. In order to identify the NJLY model in our fixed points, it will again be useful to plug $N=4$ into the above expressions for $S$ and $Y$. For future reference these values are
\begin{equation}
    S=\tfrac{3}{5}\,,\qquad Y=\tfrac{1}{4}\,.
    \label{NJLY1}
\end{equation}

As with the GNY model, one can consider taking $N_D$ to be non-integer to investigate theories with small numbers of Majorana fermions in three dimensions. For $N_D=1/2$, that is $N=2$, it is believed that the IR fixed point of the NLJY model is the same as that of a Wess--Zumino model for a four-component Majorana fermion with a cubic superpotential. Here, the $U(1)$ chiral symmetry becoming the Wess--Zumino model's $R$-symmetry\cite{Fei:2016sgs}. For $N=2$ one finds
\begin{equation}
    S=\tfrac{2}{3}\,,\qquad Y=\tfrac{2}{9}\,.
\end{equation}

\section{Analytic results for small \texorpdfstring{$\boldsymbol{N_s,N_f}$}{N\_s,N\_f}}\label{analyticsec}
As the scalar coupling $\lambda$ does not appear in $\beta^y$, we can always look for solutions with $y_{iab}=0$. However, this reduces the system to $N_s$ massless scalars with a general quartic interaction alongside $N_f$ decoupled free fermions. As this system has already been investigated in previous papers\cite{Osborn:2020cnf}, in what follows we will always assume that $y_{iab}$ is not identically zero. It is important to note that in both this and the following section we will assume for simplicity that the Yukawa coupling tensor $y_{iab}$ is real and symmetric on its $a,\,b$ indices, so that we will take $\alpha=1$. Using (\ref{eq:weyltomajorana}), one can rewrite any fixed point of (\ref{basicLag}) in terms of a real and symmetric $y_{iab}$ which will also satisfy (\ref{eq:betalambda},\,\ref{eq:betay}), so that this assumption comes at the loss of no generality. This amounts to working with 3d fermions in four dimensions, which at the level of beta functions simply changes the number of fermionic degrees of freedom in fermion loops. However, it is not always clear how one can embed these fermions in 4d Weyl or Dirac fermions. Unless otherwise indicated, we will now take $\varepsilon=1$.

We will not be interested in examining every solution to (\ref{eq:betalambda},\,\ref{eq:betay}), because there will be many solutions which are, in some ways, trivial. Given two fixed points with Lagrangians $\mathscr{L}_1$ and $\mathscr{L}_2$, one can quite easily see that the Lagrangian $\mathscr{L}_1+\mathscr{L}_2$ will also lie at a fixed point, though now with a larger number of scalars and fermions. As there is no interaction between the different sectors of this Lagrangian, this is not really a new fixed point, with the invariants (\ref{eq:invariants}) simply being the sum of the invariants of theories 1 and 2. Similarly, one could consider adding a number of free fermions or scalars to get the Lagrangian $\mathscr{L}_1+\mathscr{L}_{\text{free}}$. As the new fields will not interact with the old sector, this Lagrangian will once again be a solution to (\ref{eq:betalambda},\,\ref{eq:betay}) with precisely the same invariants as theory 1, but now with more fields. To remove these uninteresting cases, we must determine their signature and screen fixed points accordingly.

First, notice that for decoupled theories one can follow \cref{stabmatrixeigenvalues} to find a $\kappa=1$ eigenvalue for each decoupled sector, with the eigenvector being given by (\ref{eq:eig}) using only the couplings present in that sector. However, notice that the lack of a purely fermionic interaction necessitates that each decoupled sector must contain at least one scalar, else it will simply be a sector of free fermions. Thus, we must discard fixed points where the degeneracy of the $\kappa=1$ eigenvalue is greater than 1 and less than or equal to $N_s$. Second, to discard theories with free fermions we will simply need to screen out those which have a zero eigenvalue of $\gamma^\psi_{ab}$. For scalars, it will not be sufficient to look for zero eigenvalues of $\gamma^\phi_{ij}$ because to one-loop only the Yukawa interaction will contribute to the anomalous dimension.\footnote{While scalars which interact only with other scalars have zero one-loop anomalous dimension, this will be corrected by a two-loop result, where $\gamma^\phi_{ij}\supset\frac{1}{12}\lambda_{iklm}\lambda_{jklm}$.} Instead, from (\ref{eq:stabilitymatrixform}) one can see that free scalars will be marked by a column containing only $-1$ on the diagonal entry. Thus, to eliminate fixed points with free scalars we must remove those containing any $\kappa=-1$ eigenvalues.

An interesting feature of the beta functions becomes apparent when we examine fixed points analytically. Because $y_{iab}$ is an $O(N_s)$ vector, all of the Yukawa invariants we can construct, such as those in  (\ref{eq:invariants}), will contain even powers of $y$. If we take $y_{iab}$ to be diagonal in its fermionic $a$ and $b$ indices, which happens for example at fixed points which retain $U(N_f)$ or $O(N_f)$ symmetry for Weyl fermions and two-component Majorana fermions respectively, then the invariants will be sums of even powers of the diagonal elements and hence be unchanged if we invert the signs of a number of these elements. Inverting these signs also has no effect on the beta functions (\ref{eq:betalambda}) and (\ref{eq:betay}), and thus will also be a fixed point solution. However, these signs can alter the symmetry of the fixed point, and will thus alter the number of $\kappa=0$ eigenvalues of the stability matrix. 

Explicitly one can see this in the simplest case, where $N_s=1$ and $N_f=2$. Here, the Yukawa couplings become the matrix
\begin{equation}
    y=\begin{pmatrix}
        y_1 & y_2 \\ y_2 & y_3
    \end{pmatrix}\,,
\end{equation}
and their the beta-functions can be written as
\begin{equation}
    \beta^y=\begin{pmatrix}
        y_1\big(-\frac{1}{2}+\frac{7}{2}y_1^2+4y_2^2+\frac{1}{2}y_3^2\big)+3(y_1+y_2)y_2^2 & y_2\big(-\frac{1}{2}+\frac{7}{2}(y_1^2+y_3^2)+3y_1y_3+4y_2^2\big) \\
        y_2\big(-\frac{1}{2}+\frac{7}{2}(y_1^2+y_3^2)+3y_1y_3+4y_2^2\big) & y_3\big(-\frac{1}{2}+\frac{7}{2}y_3^2+4y_2^2+\frac{1}{2}y_1^2\big)+3(y_1+y_2)y_2^2
    \end{pmatrix}\,.
\end{equation}
Clearly, the off-diagonal entries can be solved by simply taking $y_2=0$, at which point one finds that
\begin{equation}
    y_1^2=y_3^2=\tfrac{1}{8}.
\end{equation}
Noting that for $y_2=0$,
\begin{equation}
    \text{Tr}(y^2)=y_1^2+y_3^2=\tfrac{1}{4}\,,\qquad\text{Tr}(y^4)=y_1^4+y_3^4=\tfrac{1}{32}\,,
\end{equation}
one finds that (\ref{eq:betalambda}) is solved by
\begin{equation}
    \lambda=\frac{1+\sqrt{19}}{12}\,,
\end{equation}
independently of how we assign signs to $y_1$ and $y_3$. We then have two distinct solutions: one in which we assign $y_1$ and $y_3$ the same sign and one in which they have opposite signs.\footnote{We can always use the field redefinition $\phi\rightarrow-\phi$ to take the sign of $y_1$ to be positive.} To see that these solutions are genuinely distinct, we write down their potentials:
\begin{equation}
    V_1(\phi,\psi)=\frac{1+\sqrt{19}}{288}\phi^4+\frac{1}{\sqrt{32}}\lsp\phi(\psi_1^2+\psi_2^2)\,,
    \label{eq:potential1}
\end{equation}
\begin{equation}
    V_2(\phi,\psi)=\frac{1+\sqrt{19}}{288}\phi^4+\frac{1}{\sqrt{32}}\lsp\phi(\psi_1^2-\psi_2^2)\,.
    \label{eq:potential2}
\end{equation}
By inspection, one can see that the first potential retains the full $O(2)$ fermionic rotation symmetry, while the second only retains a $\mathbb{Z}_2^2\rtimes\mathbb{Z}_2$ global symmetry. To verify this, one can examine the eigenvalues of the respective stability matrices to look for $\kappa=0$ eigenvalues arising due to broken symmetry generators. These eigenvalues are, respectively,
\begin{equation}
    \{\tfrac12\sqrt{19},\,1,\,\tfrac{3}{4},\,\tfrac{3}{4}\}\,,\qquad
    \{\tfrac12\sqrt{19},\,1,\,\tfrac{3}{4},0\}\,.
\end{equation}
One sees that, indeed, (\ref{eq:potential1}) is associated with no symmetry breaking while (\ref{eq:potential2}) sees the one-dimensional $O(2)$ symmetry broken down to some discrete subgroup. The peculiarity of this setup is the fact that the $\mathbb{Z}_2\times O(2)$ invariants one can construct will involve only even powers of the couplings $y_{ab}$, so that from the point of view of invariants such as $Y$ these two fixed points will be indistinguishable.\footnote{At higher loop order terms in the beta function will arise causing these fixed points to separate.} Crucially, this includes the $A$-function (\ref{eq:afunctiondef}), which will take the value
\begin{equation}
    A=-\frac{28+19\sqrt{19}}{864}
\end{equation}
at both fixed points. Examining the eigenvalues of the stability matrix, one notices that both fixed points are stable, indicating that there will be no RG flow connecting them.

Another explicit example can be seen in \cref{tab:14}. Examining section \ref{secGNY}, we see that for $N_s=1$, $N_f=4$ the $\big(S,Y\big)=\big(\frac{(1+\sqrt{145})^2}{900},\frac{1}{25}\big)$ fixed point arises from the GNY model with one Dirac fermion. Depending on whether or not chiral symmetry is broken by the interaction in the 4d Lagrangian, the symmetry in 3d may be either $O(4)$ or $O(2)^2\rtimes\mathbb{Z}_2$\cite{Jack:2023zjt,Erramilli:2022kgp,Kubota:2001kk}. The $O(3)\times\mathbb{Z}_2$ fixed point cannot arise from a Lorentz-invariant 4d Lagrangian, because there is no way to package the 3d fermions into 4d fermions consistent with both the flavour and spacetime symmetries.\footnote{The $N_s=1$ results analysed here for $N_f=2,4$ can be extended to a family of theories for every $N_f$. This was pointed out to us by H.\ Osborn.}

Also of interest is that for $N_f=1$ one is always able to use an $O(N_s)$ field redefinition so that $y_i=y\lsp\delta_{i1}$. The $\beta^y_i$ then reduces to
\begin{equation}
\beta^y=\big(-\tfrac{1}{2}\varepsilon+\tfrac12(6+\alpha)|y|^2\big)y=0\,,
\label{nfequal1}
\end{equation}
which only has a single non-zero solution. This single level will be verified in both the following analytic and numerical results.

To find analytic solutions to the beta functions, we follow the same procedure used to investigate RG stability. Namely, we solve (\ref{eq:betay}) first to find the allowed real Yukawa solutions, and then plug those into (\ref{eq:betalambda}) as constants to find the associated real solutions for the scalar coupling. Both beta functions will be polynomials in the components of the couplings, so that these problems are simply finding the shared roots of $\frac12 N_sN_f(N_f+1)$ and $\frac{1}{4!}N_s(N_s+1)(N_s+2)(N_s+3)$ polynomials respectively. We achieved this by using \textit{Mathematica's} \texttt{GroebnerBasis} function to find a simpler equivalent system of polynomials, which \textit{Mathematica's} \texttt{Solve} function was then able to dissect.\footnote{This was inspired by the discussion in \cite{434507}.} To identify the symmetries of the fixed point, we then extracted the anomalous dimensions $\{\gamma_s\}$ and $\{\gamma_f\}$ and the eigenvalues of the stability matrix $\{\kappa\}$. The results for two-component Majorana fermions with $(N_s,N_f)=(1,1)$, $(1,2)$, $(1,3)$, $(1,4)$, and $(2,1)$ are exhibited in \cref{tab:11,tab:12,tab:13,tab:14,tab:21}. Table \ref{tab:22} contains the fully interacting stable fixed point for $N_s=N_f=2$. This table also notes fixed points at which the invariants $a_2$ and $b_1$ are non-zero by use of stars, e.g.\ $\star,\,\star$ indicates a fixed point at which both are non-zero. Where these invariants do not vanish, there will be additional rank-two $O(N_s)$ tensors, so that $\phi_i\phi_i$ will not be the only quadratic singlet. All of these fixed points strictly obey the bound (\ref{simplestbound}), which can be easily seen by plotting them as in \cref{fig:11,fig:12,fig:13,fig:14,fig:21}. These plots also show the Ising fixed point for $N_s=1$, which is identical to the usual Ising fixed point but now with $N_f$ free fermions attached.

For equal values of $N_s$ and $N_f$, the theory  will be supersymmetric in three dimensions \cite{Fei:2016sgs,Liendo:2021wpo} as long as the scalar and Yukawa couplings obey the relation
\begin{equation}
\lambda_{ijkl}=S_{3,ijkl}y_{ijm}y_{mkl}\,.
    \label{eq:supersymmetry}
\end{equation}
Using this, one finds that the $N_s=N_f=1$ $\mathbb{Z}_2$ fixed point with $S=\frac{9}{49}$ and the stable $N_s=N_f=2$ $O(2)$ fixed point have emergent supersymmetry. From sections \ref{secGNY} and \ref{secNJLY} we see that these correspond to emergent supersymmetry in the $N=1$ GNY and $N=2$ NJLY model respectively.

\begin{table}[H]
\centering
\begin{tabular}{|c c c c c c c|} 
 \hline
Symmetry & $S$ & $Y$ & \makecell{$\#$ different \\[-3pt] $\gamma_\phi$(degeneracies)} & \makecell{$\#$ different \\[-3pt] $\gamma_\psi$(degeneracies)} & $\#\,\kappa<0$, =0 & $a_2,\,b_1\neq0$ \\ [0.5ex] 
 \hline
 $\mathbb{Z}_2$ & $\frac{16}{441}$ & $\frac{1}{49}$ & 1(1) & 1(1) & 1, 0 & \\ 
  $\mathbb{Z}_2$ & $\frac{9}{49}$ & $\frac{1}{49}$  & 1(1) & 1(1) & 0, 0 &\\ 
 [1ex] 
 \hline
\end{tabular}
\caption{The two interacting fixed points for $N_f=N_s=1$.}
\label{tab:11}
\end{table}

\begin{table}[H]
\centering
\begin{tabular}{|c c c c c c c|} 
 \hline
Symmetry & $S$ & $Y$ & \makecell{$\#$ different \\[-3pt] $\gamma_\phi$(degeneracies)} & \makecell{$\#$ different \\[-3pt] $\gamma_\psi$(degeneracies)} & $\#\,\kappa<0$, =0 & $a_2,\,b_1\neq0$ \\ [0.5ex] 
 \hline
  
 $\mathbb{Z}_2^2\rtimes\mathbb{Z}_2$ & $\frac{(1-\sqrt{19})^2}{144}$ & $\frac{1}{32}$ & 1(1) & 1(2) & 1, 1 &\\ 
 
$O(2)$ & $\frac{(1-\sqrt{19})^2}{144}$ & $\frac{1}{32}$ & 1(1) & 1(2) & 1, 0 &\\ 
$\mathbb{Z}_2^2\rtimes\mathbb{Z}_2$ & $\frac{(1+\sqrt{19})^2}{144}$ & $\frac{1}{32}$  & 1(1) & 1(2) & 0, 1 &\\ 
 
 $O(2)$ & $\frac{(1+\sqrt{19})^2}{144}$ & $\frac{1}{32}$  & 1(1) & 1(2) & 0, 0 &\\[1ex] 
 \hline
\end{tabular}
\caption{The four interacting fixed points for $N_s=1$ and $N_f=2$.}
\label{tab:12}
\end{table}

\begin{table}[H]
\centering
\begin{tabular}{|c c c c c c c|} 
 \hline
Symmetry & $S$ & $Y$ & \makecell{$\#$ different \\[-3pt] $\gamma_\phi$(degeneracies)} & \makecell{$\#$ different \\[-3pt] $\gamma_\psi$(degeneracies)} & $\#\,\kappa<0$, =0 & $a_2,\,b_1\neq0$ \\ [0.5ex] 
 \hline
 $O(2)\times\mathbb{Z}_2$ & $\frac{1}{9}$ & $\frac{1}{27}$ & 1(1) & 1(3) & 1, 2 &\\  
 
 $O(3)$ & $\frac{1}{9}$ & $\frac{1}{27}$ & 1(1) & 1(3) & 1, 0 &\\
 
 $O(2)\times\mathbb{Z}_2$ & $\frac{16}{81}$ & $\frac{1}{27}$  & 1(1) & 1(3) & 0, 2 &\\ 
 
 $O(3)$ & $\frac{16}{81}$ & $\frac{1}{27}$  & 1(1) & 1(3) & 0, 0 &\\[1ex] 
 \hline
\end{tabular}
\caption{The four interacting fixed points for $N_s=1$ and $N_f=3$.}
\label{tab:13}
\end{table}

\begin{table}[H]
\centering
\begin{tabular}{|c c c c c c c|} 
 \hline
Symmetry & $S$ & $Y$ & \makecell{$\#$ different \\[-3pt] $\gamma_\phi$(degeneracies)} & \makecell{$\#$ different \\[-3pt] $\gamma_\psi$(degeneracies)} & $\#\,\kappa<0$, =0 & $a_2,\,b_1\neq0$ \\ [0.5ex] 
 \hline
  $O(2)^2\rtimes\mathbb{Z}_2$ & $\frac{(1-\sqrt{145})^2}{900}$ & $\frac{1}{25}$ & 1(1) & 1(4) & 1, 4 &\\ 
  
 $O(3)\times\mathbb{Z}_2$ & $\frac{(1-\sqrt{145})^2}{900}$ & $\frac{1}{25}$ & 1(1) & 1(4) & 1, 3 &\\   
 
  $O(4)$ & $\frac{(1-\sqrt{145})^2}{900}$ & $\frac{1}{25}$ & 1(1) & 1(4) & 1, 0 &\\ 
  
   $O(2)^2\rtimes\mathbb{Z}_2$ & $\frac{(1+\sqrt{145})^2}{900}$ & $\frac{1}{25}$  & 1(1) & 1(4) & 0, 4 &\\ 

  $O(3)\times\mathbb{Z}_2$ & $\frac{(1+\sqrt{145})^2}{900}$ & $\frac{1}{25}$  & 1(1) & 1(4) & 0, 3 &\\ 

   $O(4)$ & $\frac{(1+\sqrt{145})^2}{900}$ & $\frac{1}{25}$ & 1(1) & 1(4) & 0, 0 &\\[1ex] 
 \hline
\end{tabular}
\caption{The six interacting fixed points for $N_s=1$ and $N_f=4$.}
\label{tab:14}
\end{table}

\begin{table}[H]
\centering
\begin{tabular}{|c c c c c c c|} 
 \hline
Symmetry & $S$ & $Y$ & \makecell{$\#$ different \\[-3pt] $\gamma_\phi$(degeneracies)} & \makecell{$\#$ different \\[-3pt] $\gamma_\psi$(degeneracies)} & $\#\,\kappa<0$, =0 & $a_2,\,b_1\neq0$ \\ [0.5ex] 
 \hline
 $\mathbb{Z}_2$ & $0.142072$ & $\frac{1}{49}$ & 2(1,1) & 1(1) & 4, 1 & $\star,\,\star$\\

  $\mathbb{Z}_2^2$ & $0.233467$ & $\frac{1}{49}$ & 2(1,1) & 1(1) & 2, 1 & $\star,\,\star$ \\ 
 
  $\mathbb{Z}_2^2$ & $0.296351$ & $\frac{1}{49}$  & 2(1,1) & 1(1) & 0, 1 & $\star,\,\star$ \\ [1ex] 
 \hline
\end{tabular}
\caption{The three interacting fixed points with $N_s=2$ and $N_f=1$.}
\label{tab:21}
\end{table}

\begin{table}[ht]
\centering
\begin{tabular}{|c c c c c c c|} 
 \hline
Symmetry & $S$ & $Y$ & \makecell{$\#$ different \\[-3pt] $\gamma_\phi$(degeneracies)} & \makecell{$\#$ different \\[-3pt] $\gamma_\psi$(degeneracies)} & $\#\,\kappa<0$, =0 & $a_2,\,b_1\neq0$ \\ [0.5ex] 
 \hline
 $O(2)$ & $\frac{2}{3}$ & $\frac{2}{9}$ & 1(2) & 1(2) & 0, 1 & \\ 
  [1ex] 
 \hline
\end{tabular}
\caption{Stable, fully interacting fixed point for $N_s=N_f=2$. Note that this point corresponds to the GNY model with $N=2$ and is supersymmetric. Further note that there is a stable fixed point consisting of two decoupled stable $\mathbb{Z}_2$ theories.}
\label{tab:22}
\end{table}

\begin{figure}[H]
\centering
\begin{tikzpicture}
\begin{axis}[
    xmin=0, xmax=0.2,
    ymin=0, ymax=0.025,
    xlabel= $S$,
    ylabel= $Y$,
    ylabel style={rotate=-90},
    xticklabel style={/pgf/number format/fixed, /pgf/number format/precision=3},
    xtick distance=0.025,
    yticklabel style={scaled ticks=false, /pgf/number format/fixed, /pgf/number format/precision=3},
    title={Scalar-Fermion Fixed Points for $N_s=1$ and $N_f=1$},
    legend pos = outer north east
]
\addplot+[
    only marks,
    mark=*,
    mark options={color=Dark2-A,fill=Dark2-A},
    mark size=2pt]
table{Datafiles/1,1/stable.dat};
\addplot+[
    only marks,
    mark=square*,
    mark options={color=Dark2-F,fill=Dark2-F},
    mark size=1.5pt]
table{Datafiles/1,1/semistable.dat};
\addplot+[
    only marks,
    mark=diamond*,
    mark options={color=Dark2-B,fill=Dark2-B},
    mark size=3pt]
table{Datafiles/1,1/unstable.dat};
\addplot[domain=0:2, 
    samples=2, 
    color=Dark2-C]{x/3-1/24};
\legend{Stable, Mixed Stability, Unstable, Bound on Invariants (\ref{simplestbound})}
\end{axis}
\end{tikzpicture}
 \caption{Analytic interacting fixed points for $N_s=N_f=1$ given by Table \ref{tab:11}. Note that the Ising fixed point ($Y=0$ while $S\neq0$) is not included in the table as it only has free fermions.}
    \label{fig:11}
\end{figure}
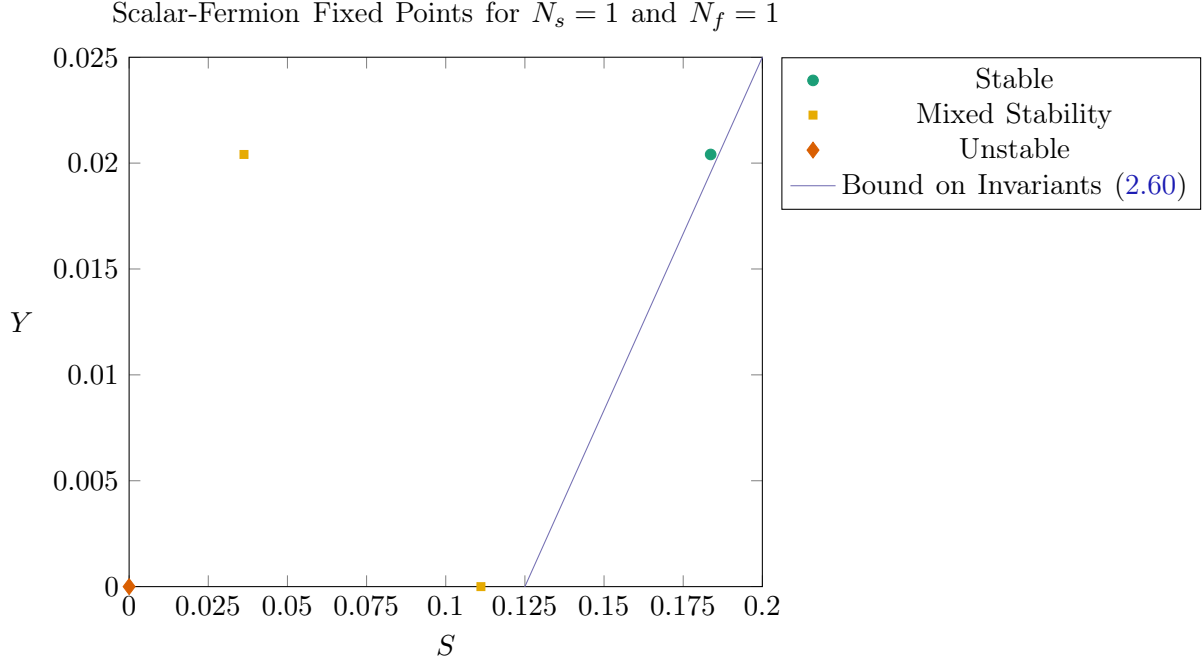

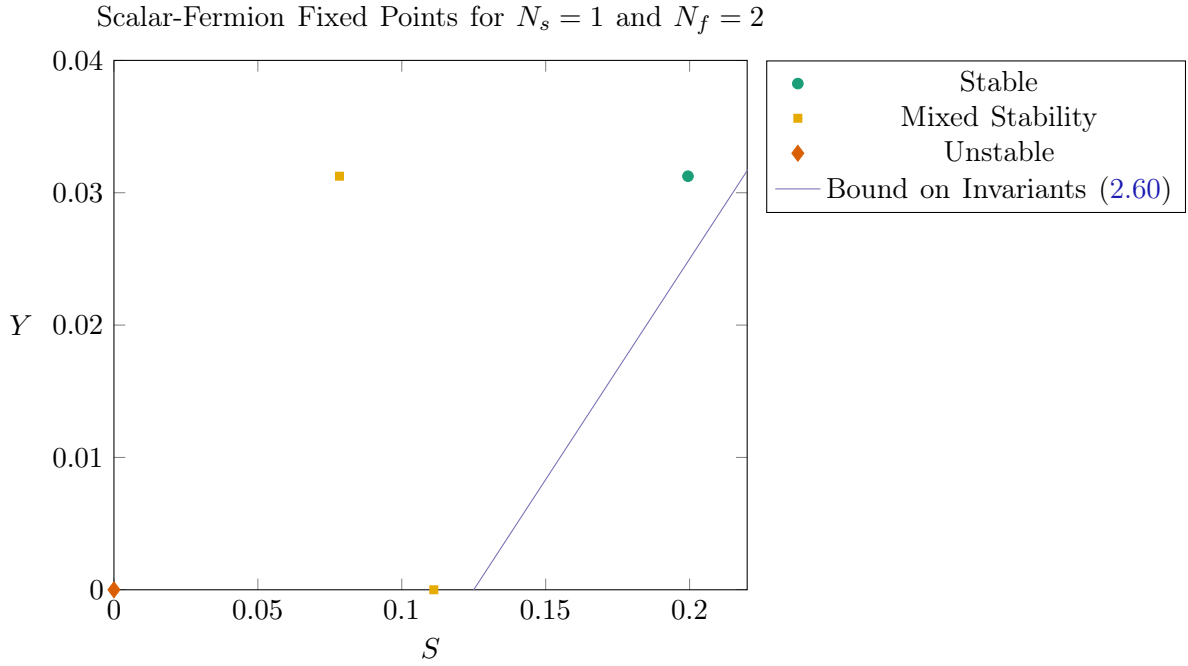
\begin{figure}[H]
\centering
\begin{tikzpicture}
\begin{axis}[
    xmin=0, xmax=0.22,
    ymin=0, ymax=0.04,
    xlabel= $S$,
    ylabel= $Y$,
    ylabel style={rotate=-90},
    xticklabel style={scaled ticks=false, /pgf/number format/fixed, /pgf/number format/precision=3},
    yticklabel style={scaled ticks=false, /pgf/number format/fixed, /pgf/number format/precision=3},
    title={Scalar-Fermion Fixed Points for $N_s=1$ and $N_f=2$},
    legend pos = outer north east
]
\addplot+[
    only marks,
    mark=*,
    mark options={color=Dark2-A,fill=Dark2-A},
    mark size=2pt]
table{Datafiles/1,2/stable.dat};
\addplot+[
    only marks,
    mark=square*,
    mark options={color=Dark2-F,fill=Dark2-F},
    mark size=1.5pt]
table{Datafiles/1,2/semistable.dat};
\addplot+[
    only marks,
    mark=diamond*,
    mark options={color=Dark2-B,fill=Dark2-B},
    mark size=3pt]
table{Datafiles/1,2/unstable.dat};
\addplot[domain=0:2, 
    samples=2, 
    color=Dark2-C]{x/3-1/24};
\legend{Stable, Mixed Stability, Unstable, Bound on Invariants (\ref{simplestbound})}
\end{axis}
\end{tikzpicture}
 \caption{Analytic interacting fixed points for $N_s=1$ and $N_f=2$ given by Table \ref{tab:12} along with the Ising fixed point.}
    \label{fig:12}
\end{figure}

\begin{figure}[H]
\centering
\begin{tikzpicture}
\begin{axis}[
    xmin=0, xmax=0.25,
    ymin=0, ymax=0.04,
    xlabel= $S$,
    ylabel= $Y$,
    ylabel style={rotate=-90},
    xticklabel style={scaled ticks=false, /pgf/number format/fixed, /pgf/number format/precision=3},
    yticklabel style={scaled ticks=false, /pgf/number format/fixed, /pgf/number format/precision=3},
    title={Scalar-Fermion Fixed Points for $N_s=1$ and $N_f=3$},
    legend pos = outer north east
]
\addplot+[
    only marks,
    mark=*,
    mark options={color=Dark2-A,fill=Dark2-A},
    mark size=2pt]
table{Datafiles/1,3/stable.dat};
\addplot+[
    only marks,
    mark=square*,
    mark options={color=Dark2-F,fill=Dark2-F},
    mark size=1.5pt]
table{Datafiles/1,3/semistable.dat};
\addplot+[
    only marks,
    mark=diamond*,
    mark options={color=Dark2-B,fill=Dark2-B},
    mark size=3pt]
table{Datafiles/1,3/unstable.dat};
\addplot[domain=0:2, 
    samples=2, 
    color=Dark2-C]{x/3-1/24};
\legend{Stable, Mixed Stability, Unstable, Bound on Invariants (\ref{simplestbound})}
\end{axis}
\end{tikzpicture}
 \caption{Analytic interacting fixed points for $N_s=1$ and $N_f=3$ given by Table \ref{tab:13} along with the Ising fixed point.}
    \label{fig:13}
\end{figure}
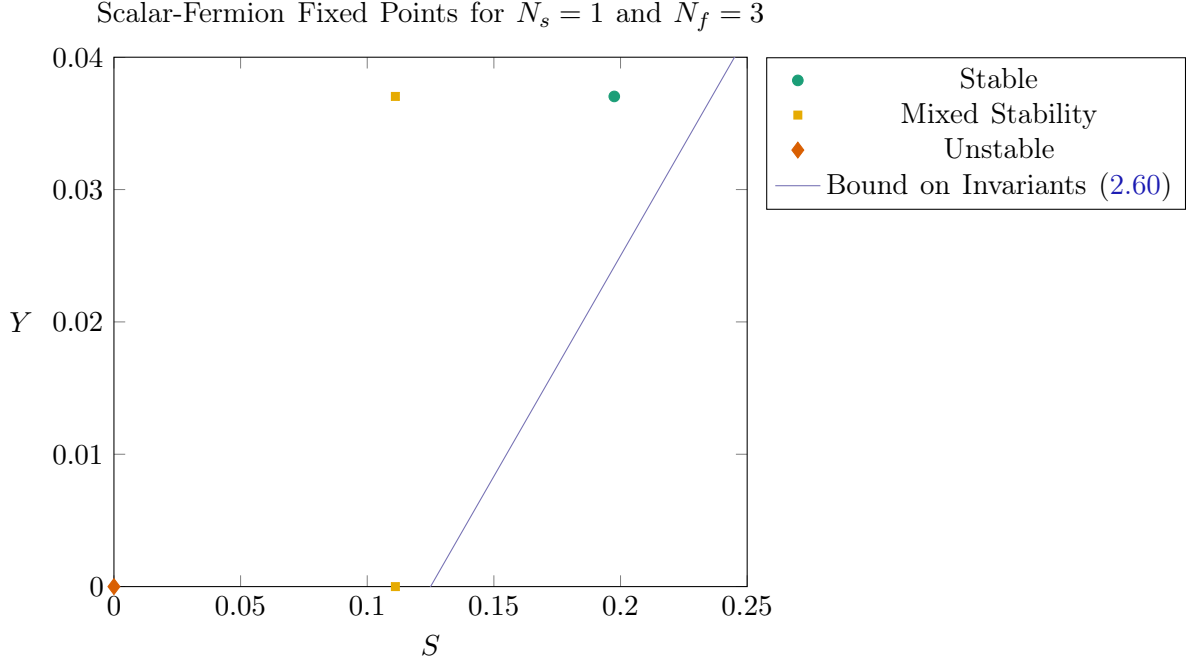

\begin{figure}[H]
\centering
\begin{tikzpicture}
\begin{axis}[
    xmin=0, xmax=0.25,
    ymin=0, ymax=0.05,
    xlabel= $S$,
    ylabel= $Y$,
    ylabel style={rotate=-90},
    xticklabel style={scaled ticks=false, /pgf/number format/fixed, /pgf/number format/precision=3},
    yticklabel style={scaled ticks=false, /pgf/number format/fixed, /pgf/number format/precision=3},
    title={Scalar-Fermion Fixed Points for $N_s=1$ and $N_f=4$},
    legend pos = outer north east
]
\addplot+[
    only marks,
    mark=*,
    mark options={color=Dark2-A,fill=Dark2-A},
    mark size=2pt]
table{Datafiles/1,4/stable.dat};
\addplot+[
    only marks,
    mark=square*,
    mark options={color=Dark2-F,fill=Dark2-F},
    mark size=1.5pt]
table{Datafiles/1,4/semistable.dat};
\addplot+[
    only marks,
    mark=diamond*,
    mark options={color=Dark2-B,fill=Dark2-B},
    mark size=3pt]
table{Datafiles/1,4/unstable.dat};
\addplot[domain=0:2, 
    samples=2, 
    color=Dark2-C]{x/3-1/24};
\legend{Stable, Mixed Stability, Unstable, Bound on Invariants (\ref{simplestbound})}
\end{axis}
\end{tikzpicture}
 \caption{Analytic interacting fixed points for $N_s=1$ and $N_f=4$ given by Table \ref{tab:14} along with the Ising fixed point.}
    \label{fig:14}
\end{figure}
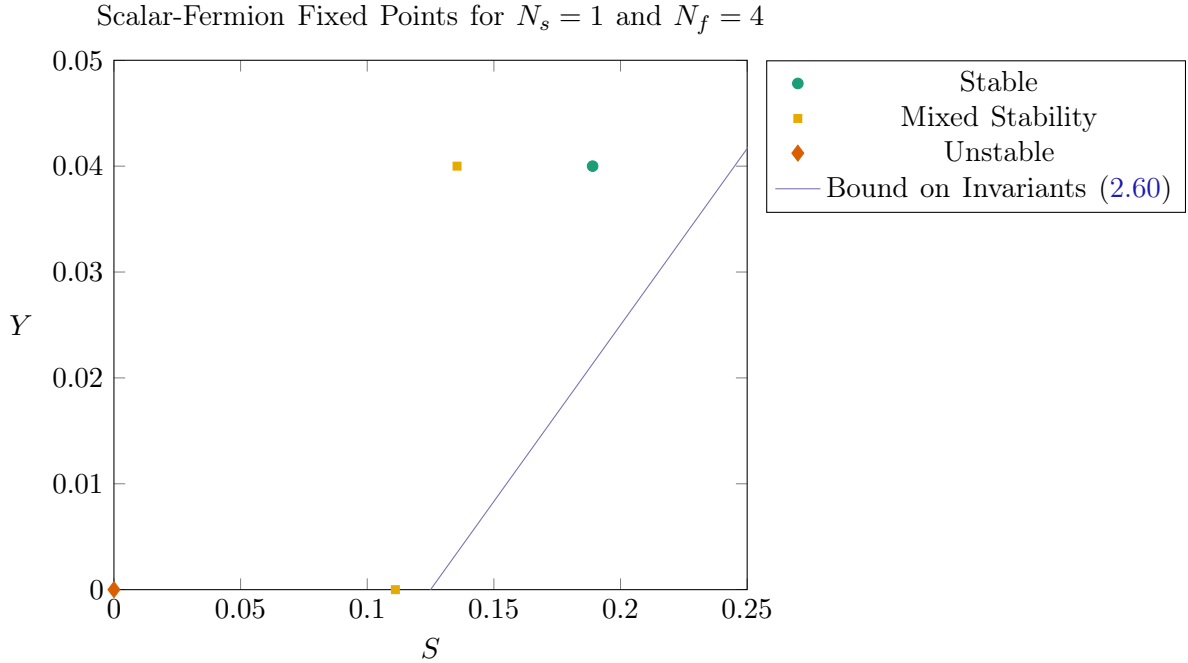

\begin{figure}[H]
\centering
\begin{tikzpicture}
\begin{axis}[
    xmin=0, xmax=0.4,
    ymin=0, ymax=0.025,
    xlabel= $S$,
    ylabel= $Y$,
    ylabel style={rotate=-90},
    xticklabel style={scaled ticks=false, /pgf/number format/fixed, /pgf/number format/precision=3},
    yticklabel style={scaled ticks=false, /pgf/number format/fixed, /pgf/number format/precision=3},
    title={Scalar-Fermion Fixed Points for $N_s=2$ and $N_f=1$},
    legend pos = outer north east
]
\addplot+[
    only marks,
    mark=*,
    mark options={color=Dark2-A,fill=Dark2-A},
    mark size=2pt]
table{Datafiles/2,1/stable.dat};
\addplot+[
    only marks,
    mark=square*,
    mark options={color=Dark2-F,fill=Dark2-F},
    mark size=1.5pt]
table{Datafiles/2,1/semistable.dat};
\addplot+[
    only marks,
    mark=diamond*,
    mark options={color=Dark2-B,fill=Dark2-B},
    mark size=3pt]
table{Datafiles/2,1/unstable.dat};
\addplot[domain=0:3, 
    samples=2, 
    color=Dark2-C]{x/3-2/24};
\legend{Stable, Mixed Stability, Unstable, Bound on Invariants (\ref{simplestbound})}
\end{axis}
\end{tikzpicture}
 \caption{Analytic interacting fixed points for $N_s=2$ and $N_f=1$ given by Table \ref{tab:21}. To compare with the purely scalar case, we have also included the $O(2)$ fixed point with $Y=0$.}
    \label{fig:21}
\end{figure}
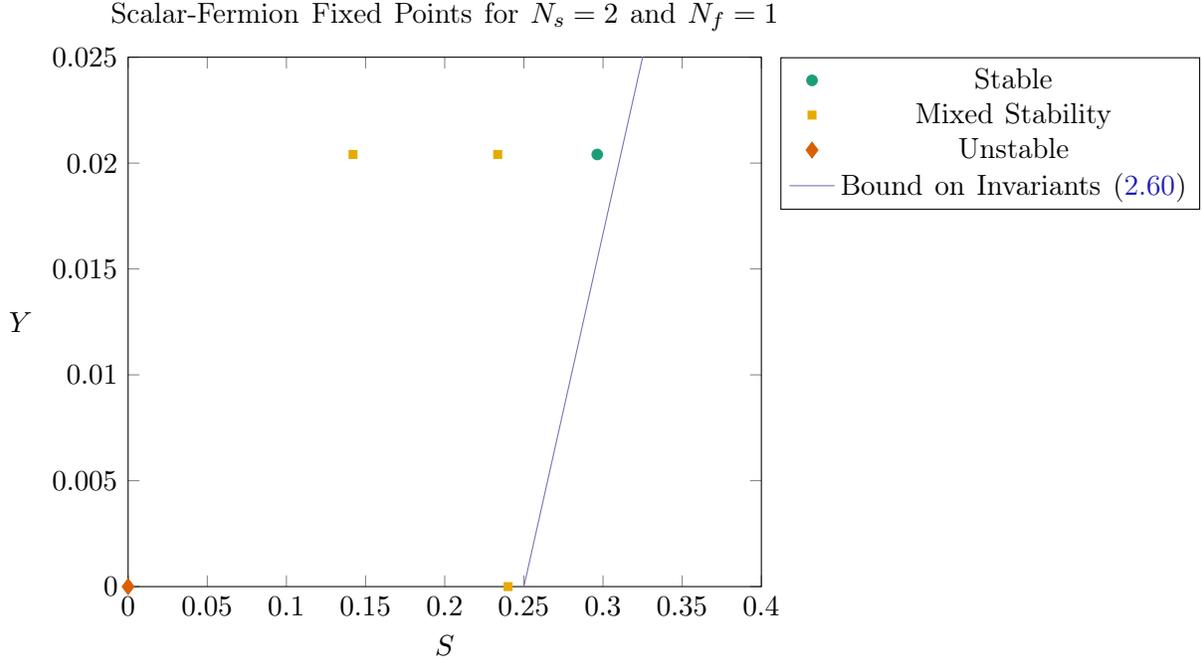

\section{Numerical searches}\label{numericsec}
The time required to implement the techniques used to find analytic fixed points rapidly increases as $N_s$ and $N_f$ increase, necessitating numerical methods for efficient searches. For low enough values of $N_s$ and $N_f$, \textit{Mathematica} is still powerful enough for our purposes, and in this paper we rely on its \texttt{FindRoot} function, selecting only those solutions with an error less than $5\times 10^{-15}$ in order to ensure accuracy. To verify convergence, we ask \texttt{FindRoot} to use these solutions as initial points to obtain new, iterated solutions with 100 digits of predision. Ultimately, we then only accept those iterated solutions with an error less than $10^{-80}$. \texttt{FindRoot} relies on an initial guess for a solution, so to ensure that the solver would be capable of locating all of the possible fixed points we selected random initial points where each of the components of $\lambda_{ijkl}$ and $y_{iab}$ are chosen to lie in the interval $[0,1]$. This is not a systematic search for all possible fixed points, but as our intention is to examine the size of the space of fixed points, a broad survey will be sufficient. One drawback of this method is that it is difficult for the solver to then find the purely scalar fixed points when $N_s$ and $N_f$ are not small, because of how increasingly rare it becomes that the initial guess has $y_{iab}$ small. However, as these fixed points have been previously explored\cite{Osborn:2020cnf}, this does not substantially affect the findings of this work. As in \cref{analyticsec} we will assume that the couplings $y_{iab}$ are all real, and we will discard any non-fully interacting fixed points.

First, to verify the efficacy of our numerical program, we applied it to the values of $N_s$ and $N_f$ in \ref{analyticsec}. In each case we began with 10,000 random starting points, and found complete agreement between with the results given in \cref{tab:11,tab:12,tab:13,tab:14,tab:21}. In every case but for $N_s=1$ and $N_f=4$, the number of starting points was sufficient for the numerical solver to correctly obtain all of the interacting fixed points. For $N_s=1$ and $N_f=4$, the solver only found eleven out of the sixteen points listed in \cref{tab:14}, but this difference is eliminated as long as we increase the number of starting points. This highlights a weakness in our method of choosing initial guesses. The use of \texttt{FindRoot} can make it difficult to find certain solutions, e.g.\ solutions with $y_{iab}$ identically vanishing, unless the initial guess is finely tuned, making our random guesses a pure game of chance with regard to these points. However, for low enough $N_s$ and $N_f$ this method is sufficient to find the majority of theories. In addition to the fully interacting fixed points, the solver was able to sporadically find solutions corresponding to interacting scalar theories accompanied by free fermions. When they appeared they could be matched with known fixed points found in \cite{Osborn:2020cnf}.

\subsection{Numerical results for small \texorpdfstring{$N_s$, $N_f$}{N\_s, N\_f}}

Using our numerical solver, we can extend our fixed point search past the values of $N_s$ and $N_f$ considered analytically. In \cref{,fig:22,fig:23,fig:24,fig:31} the values of $S$ are plotted against $Y$ for $(N_s,N_f)=(2,2)$, $(2,3)$, $(2,4)$ and $(3,1)$, where the stable fixed points are plotted as green circles, unstable fixed points are red diamonds, mixed stability fixed points are yellow squares, and the bound on invariants (\ref{simplestbound}) as the blue line. In Appendix \ref{dataappendix} we list \cref{tab:22n,tab:23,tab:24,tab:31} which contain the fixed points in the aforementioned figures, where we have used \textit{Mathematica's} \texttt{Rationalize} function to interpret the numerical values of $Y$ and $S$ where possible.

A few important properties of scalar-fermion solutions can be noted already by looking at these fixed points. Firstly, it is interesting to note that while for purely scalar $N_s=2$ there is only one non-trivial fixed point, the $O(2)$ point, introducing fermions into the system and increasing the value of $N_f$ leads to the rapid proliferation of solutions, not only by increasing the number of lines of constant $Y$, but by increasing the number of solutions for $\lambda$ within each level. This continues to be true for $N_s=3$ in \cref{fig:31} and \cref{tab:31}, where we find 7 solutions with fermions as opposed to the 3 found in the purely scalar case. Secondly, the stable fixed points are distributed in agreement with section \ref{stabilitysection}, demonstrating an important distinction between the scalar and scalar-fermion systems. If one calculates the $A$-function at each of the fixed points, one finds that the stable solution minimises it for that specific Yukawa coupling, and in the case where there are multiple stable fixed points, as with $N_s=N_f=2$, they appear only for different solutions to $\beta_y$. The presence of distinct stable fixed points means that, unlike in the purely scalar case, the RG flows arising from the beta functions can be divided into different basins of stability. The IR theory one ends up in, if one does not flow away to infinity, is not fixed but rather depends upon the initial position in theory space and perturbation one introduces.

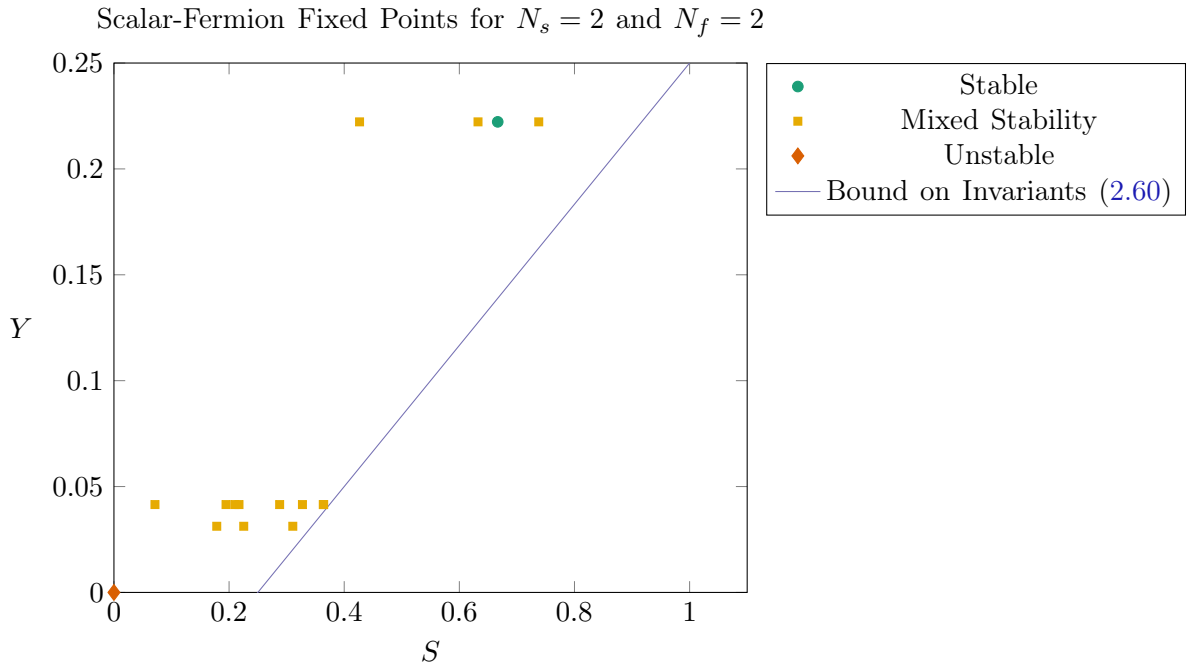
\begin{figure}[t]
\centering
\begin{tikzpicture}
\begin{axis}[
    xmin=0, xmax=1.1,
    ymin=0, ymax=0.25,
    xlabel= $S$,
    ylabel= $Y$,
    ylabel style={rotate=-90},
    yticklabel style={scaled ticks=false, /pgf/number format/fixed, /pgf/number format/precision=3},
    title={Scalar-Fermion Fixed Points for $N_s=2$ and $N_f=2$},
    legend pos = outer north east
]
\addplot+[
    only marks,
    mark=*,
    mark options={color=Dark2-A,fill=Dark2-A},
    mark size=2pt]
table{Datafiles/2,2/stable.dat};
\addplot+[
    only marks,
    mark=square*,
    mark options={color=Dark2-F,fill=Dark2-F},
    mark size=1.5pt]
table{Datafiles/2,2/semistable.dat};
\addplot+[
    only marks,
    mark=diamond*,
    mark options={color=Dark2-B,fill=Dark2-B},
    mark size=3pt]
table{Datafiles/2,2/unstable.dat};
\addplot[domain=0:1, 
    samples=2, 
    color=Dark2-C]{x/3-2/24};
\legend{Stable, Mixed Stability, Unstable, Bound on Invariants (\ref{simplestbound})}
\end{axis}
\end{tikzpicture}
 \caption{Results of numerical search for $N_s=N_f=2$ beginning with 10,000 points. The points are listed in \cref{tab:22n}. We find a total of 19 fully interacting fixed points. The stable fixed point agrees with that found in \cref{tab:22}.}
    \label{fig:22}
\end{figure}

It is important to note that the blue line requires the simultaneous saturation of both the scalar bound (\ref{eq:three}) and the fermionic bound (\ref{eq:two}). While a number of fixed points are able to saturate one of these bounds, none of the fixed points have saturated both. However, as the plots indicate there are a number of points which are very nearly able to saturate (\ref{simplestbound}), making it unlikely that an improved bound exists when the number of fields involved is small. This feature is unlike the purely scalar case\cite{Osborn:2020cnf}.

\begin{figure}[t]
\centering
\begin{tikzpicture}
\begin{axis}[
    xmin=0, xmax=1,
    ymin=0, ymax=0.25,
    xlabel= $S$,
    ylabel= $Y$,
    ylabel style={rotate=-90},
    yticklabel style={scaled ticks=false, /pgf/number format/fixed, /pgf/number format/precision=3},
    title={Scalar-Fermion Fixed Points for $N_s=2$ and $N_f=3$},
    legend pos = outer north east
]
\addplot+[
    only marks,
    mark=*,
    mark options={color=Dark2-A,fill=Dark2-A},
    mark size=2pt]
table{Datafiles/2,3/stable.dat};
\addplot+[
    only marks,
    mark=square*,
    mark options={color=Dark2-F,fill=Dark2-F},
    mark size=1.5pt]
table{Datafiles/2,3/semistable.dat};
\addplot+[
    only marks,
    mark=diamond*,
    mark options={color=Dark2-B,fill=Dark2-B},
    mark size=3pt]
table{Datafiles/2,3/unstable.dat};
\addplot[domain=0:1, 
    samples=2, 
    color=Dark2-C]{x/3-2/24};
\legend{Stable, Mixed Stability, Unstable, Bound on Invariants (\ref{simplestbound})}
\end{axis}
\end{tikzpicture}
    \caption{Results of numerical search for $N_s=2$, $N_f=3$ beginning with 10,000 points. The points are listed in \cref{tab:23} We find 33 distinct fixed points, of which one is stable.}
    \label{fig:23}
\end{figure}
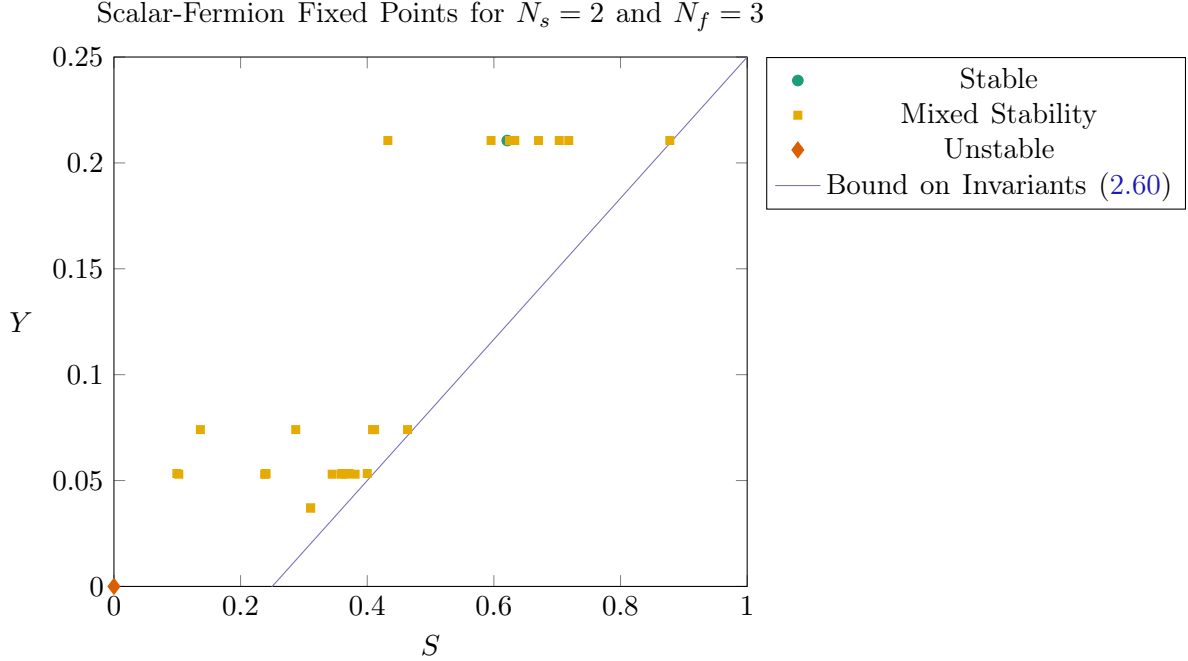

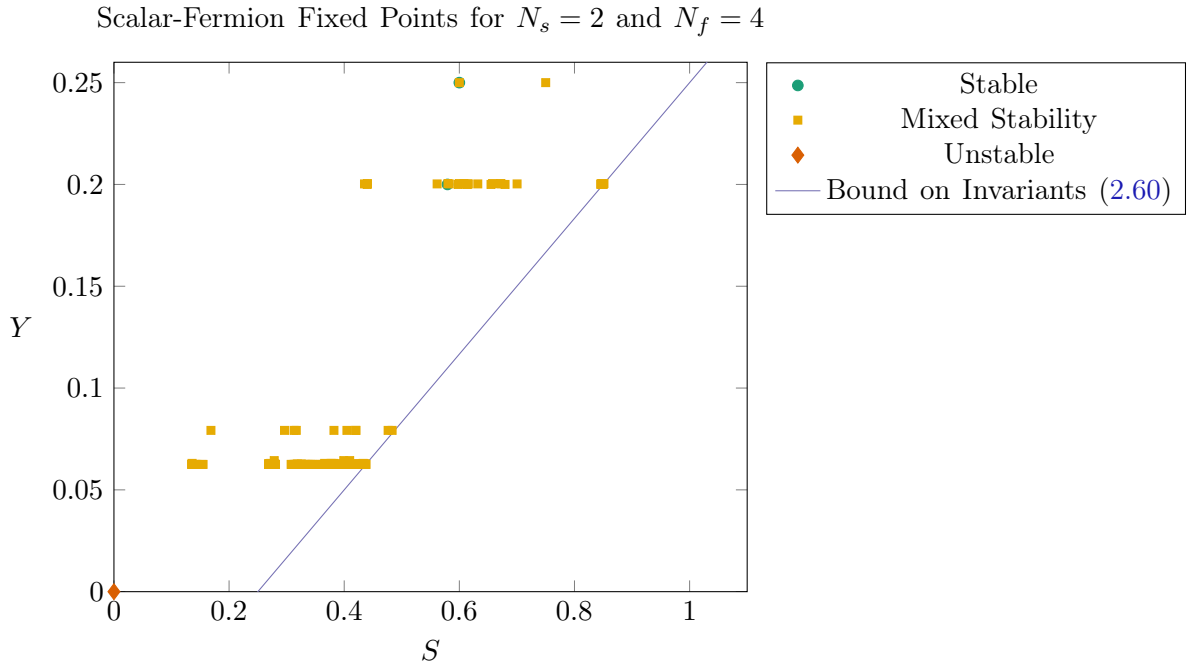
\begin{figure}[H]
\centering
\begin{tikzpicture}
\begin{axis}[
    xmin=0, xmax=1.1,
    ymin=0, ymax=0.26,
    xlabel= $S$,
    ylabel= $Y$,
    ylabel style={rotate=-90},
    yticklabel style={scaled ticks=false, /pgf/number format/fixed, /pgf/number format/precision=3},
    title={Scalar-Fermion Fixed Points for $N_s=2$ and $N_f=4$},
    legend pos = outer north east
]
\addplot+[
    only marks,
    mark=*,
    mark options={color=Dark2-A,fill=Dark2-A},
    mark size=2pt]
table{Datafiles/2,4/stable.dat};
\addplot+[
    only marks,
    mark=square*,
    mark options={color=Dark2-F,fill=Dark2-F},
    mark size=1.5pt]
table{Datafiles/2,4/semistable.dat};
\addplot+[
    only marks,
    mark=diamond*,
    mark options={color=Dark2-B,fill=Dark2-B},
    mark size=3pt]
table{Datafiles/2,4/unstable.dat};
\addplot[domain=0:2, 
    samples=2, 
    color=Dark2-C]{x/3-2/24};
\legend{Stable, Mixed Stability, Unstable, Bound on Invariants (\ref{simplestbound})}
\end{axis}
\end{tikzpicture}
    \caption{Results of numerical search for $N_s=2$, $N_f=4$ beginning with 10,000 points. The points are listed in \cref{tab:24}. Note that no fixed points manage to completely saturate (\ref{simplestbound}). In total we find 254 distinct fixed points, two of which are stable. The NJLY model appears as the stable fixed point at $(0.6,0.25)$.}
    \label{fig:24}
\end{figure}

Even for these small numbers of scalars and fermions, \cref{fig:22,fig:23,fig:24,fig:31} are sufficient to identify some qualitative trends in how the solutions arrange themselves. As emphasised before, the independence of $\beta^y$ from the scalar coupling $\lambda$, a feature specific to one loop,\footnote{While higher-loop terms in the beta function will remove this $\lambda$-independence, there is expected to be a one-to-one correspondence between solutions of the one loop beta function and higher loop fixed points. This correspondence suggests that the higher order terms will simply perturb the fixed points to no longer exactly lie on the same line.} means that there will be many fixed point solutions which share the same value of the Yukawa coupling, perhaps up to an $O(N_s)\times O(N_f)$ field redefinition. These fixed points may, however, have markedly distinct scalar potentials, leading to the lines of solutions with constant $Y$ but variable $S$ which feature prominently in the figures.

Also interesting to note is the fact that in all instances $Y$ is rational, unlike $S$, and this may be a general feature of the one loop beta function (\ref{eq:betay}). This can clearly be seen by examining Tables \ref{tab:11} to \ref{tab:31}. For $N_f=1$ this immediately follows from (\ref{nfequal1}). For $N_s=1$, the beta function for $y_{ab}$ becomes
\begin{equation}
    \big[\big(-\tfrac{1}{2}\varepsilon+\Tr y^2\big)\delta_{ab}+3(y^2)_{ab}\big]y_{bc}=0\,.
\end{equation}
For fully interacting fixed points, $y_{ab}$ must be invertible, so that one can see that the only allowed value of $Y$ will be
\begin{equation}
    Y=\frac{N_f}{(N_f+6)^2}\varepsilon\,.
\end{equation}
This permits not only the GNY model, but also those mentioned previously arising from flipping individual signs in $y_{ab}$. Though we do not have a proof of this property for general $N_s$ and $N_f$, we have verified it in all of the cases identified numerically. To demonstrate that these trends are not peculiar to $N_s=1$ or $2$, we include plots of $N_s=3,4$ and $N_f\leq 4$ in \cref{fig:32,fig:33,fig:34,fig:41,fig:42,fig:43,fig:44}. Here, the number of fixed points becomes too large to reasonably list in tabular form. It is interesting to note that the solver is unable to find any stable fixed points for $N_s=3,\,4$; however this may be due to limitations of the solver rather than a feature of the beta functions.

Examining \cref{tab:22n,tab:23,tab:24,tab:31}, one sees another peculiarity, namely the existence of fixed points with the same invariants, yet distinct numbers of $\kappa=0$ eigenvalues. In the analytic solutions it was seen that this may be associated with sign changes in Yukawa couplings which break the symmetry of the fixed point without either spoiling the beta functions or changing the invariants, which depend upon $y$ only in even powers. It seems likely that a similar effect is happening here. 

\begin{figure}[H]
\centering
\begin{tikzpicture}
\begin{axis}[
    xmin=0, xmax=0.45,
    ymin=0, ymax=0.025,
    xlabel= $S$,
    ylabel= $Y$,
    ylabel style={rotate=-90},
    yticklabel style={scaled ticks=false, /pgf/number format/fixed, /pgf/number format/precision=3},
    title={Scalar-Fermion Fixed Points for $N_s=3$ and $N_f=1$},
    legend pos = outer north east
]
\addplot+[
    only marks,
    mark=square*,
    mark options={color=Dark2-F,fill=Dark2-F},
    mark size=1.5pt]
table{Datafiles/3,1/semistable.dat};
\addplot+[
    only marks,
    mark=diamond*,
    mark options={color=Dark2-B,fill=Dark2-B},
    mark size=3pt]
table{Datafiles/3,1/unstable.dat};
\addplot[domain=0:1, 
    samples=2, 
    color=Dark2-C]{x/3-3/24};
\legend{Mixed Stability, Unstable, Bound on Invariants (\ref{simplestbound})}
\end{axis}
\end{tikzpicture}
    \caption{Results of numerical search for $N_s=3$, $N_f=1$ beginning with 10,000 points. The points are listed in \cref{tab:31}. We find 7 distinct interacting fixed points, of which none are stable.}
    \label{fig:31}
\end{figure}
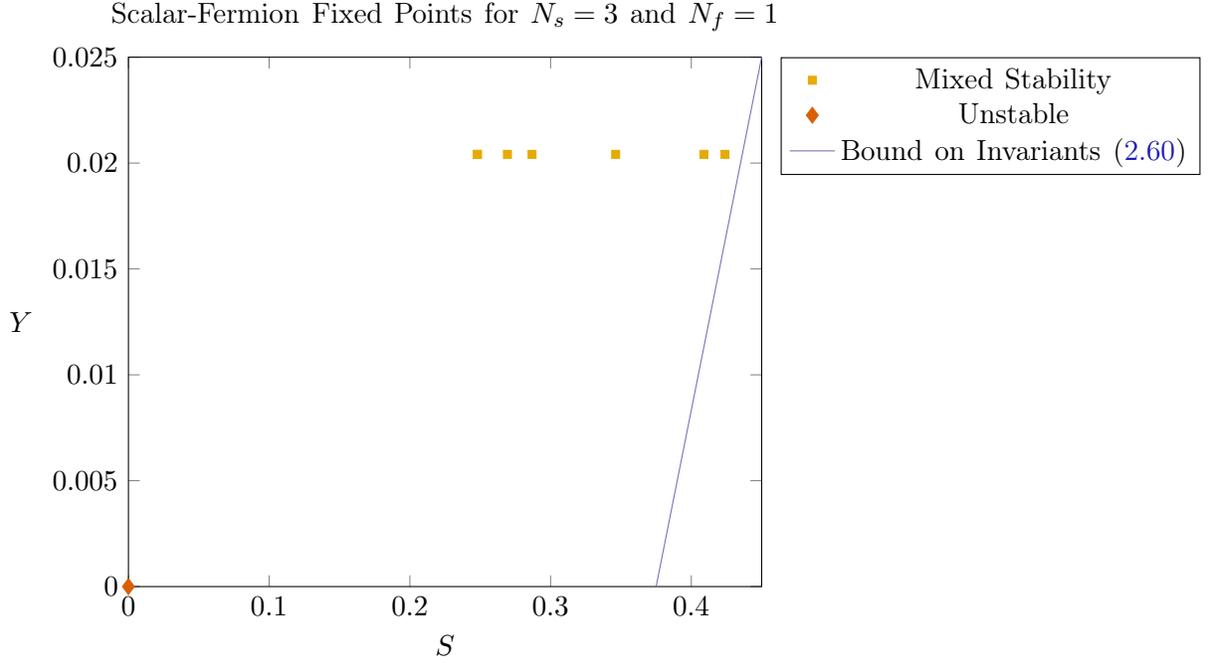

\begin{figure}[H]
\centering
\begin{tikzpicture}
\begin{axis}[
    xmin=0, xmax=1.2,
    ymin=0, ymax=0.25,
    xlabel= $S$,
    ylabel= $Y$,
    ylabel style={rotate=-90},
    yticklabel style={scaled ticks=false, /pgf/number format/fixed, /pgf/number format/precision=3},
    title={Scalar-Fermion Fixed Points for $N_s=3$ and $N_f=2$},
    legend pos = outer north east
]
\addplot+[
    only marks,
    mark=square*,
    mark options={color=Dark2-F,fill=Dark2-F},
    mark size=1.5pt]
table{Datafiles/3,2/semistable.dat};
\addplot+[
    only marks,
    mark=diamond*,
    mark options={color=Dark2-B,fill=Dark2-B},
    mark size=3pt]
table{Datafiles/3,2/unstable.dat};
\addplot[domain=0:2, 
    samples=2, 
    color=Dark2-C]{x/3-3/24};
\legend{Mixed Stability, Unstable, Bound on Invariants (\ref{simplestbound})}
\end{axis}
\end{tikzpicture}
    \caption{Results of numerical search for $N_s=3$, $N_f=2$ beginning with 10,000 points. We find 35 distinct fixed points, of which none are stable.}
    \label{fig:32}
\end{figure}
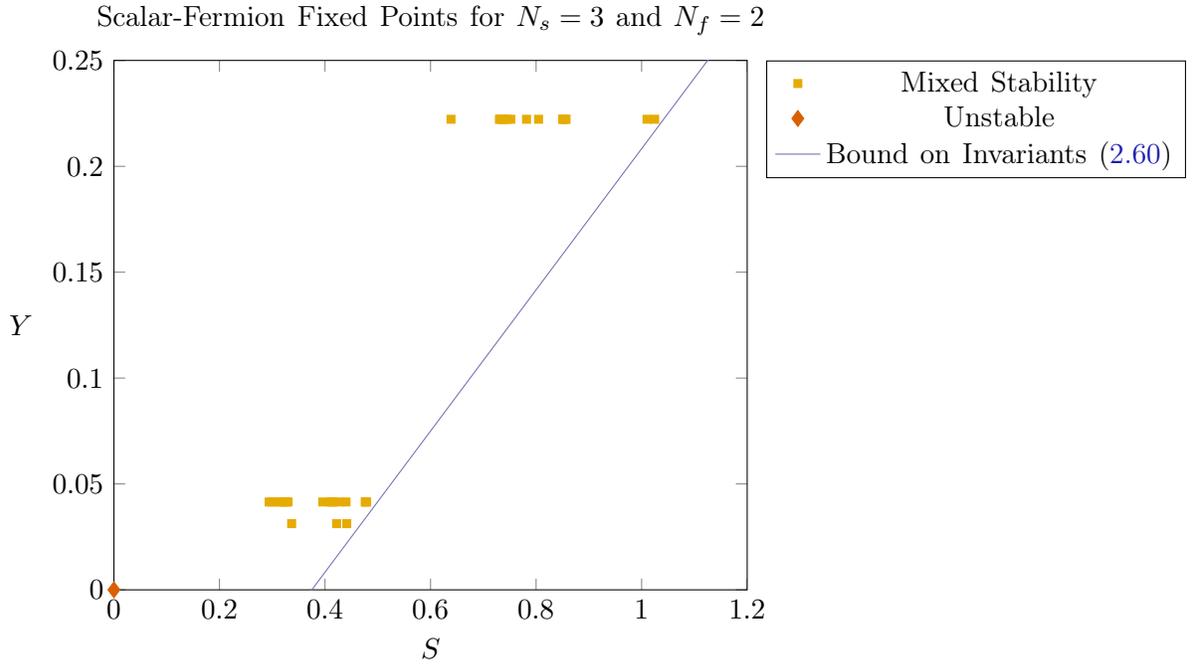

\begin{figure}[H]
\centering
\begin{tikzpicture}
\begin{axis}[
    xmin=0, xmax=1.6,
    ymin=0, ymax=0.41,
    xlabel= $S$,
    ylabel= $Y$,
    ylabel style={rotate=-90},
    title={Scalar-Fermion Fixed Points for $N_s=3$ and $N_f=3$},
    legend pos = outer north east
]
\addplot+[
    only marks,
    mark=square*,
    mark options={color=Dark2-F,fill=Dark2-F},
    mark size=1.5pt]
table{Datafiles/3,3/semistable.dat};
\addplot+[
    only marks,
    mark=diamond*,
    mark options={color=Dark2-B,fill=Dark2-B},
    mark size=3pt]
table{Datafiles/3,3/unstable.dat};
\addplot[domain=0:2, 
    samples=2, 
    color=Dark2-C]{x/3-3/24};
\legend{Mixed Stability, Unstable, Bound on Invariants (\ref{simplestbound})}
\end{axis}
\end{tikzpicture}
    \caption{Results of numerical search for $N_s=N_f=3$ beginning with 10,000 points. We find 171 distinct fixed points, of which none are stable.}
    \label{fig:33}
\end{figure}
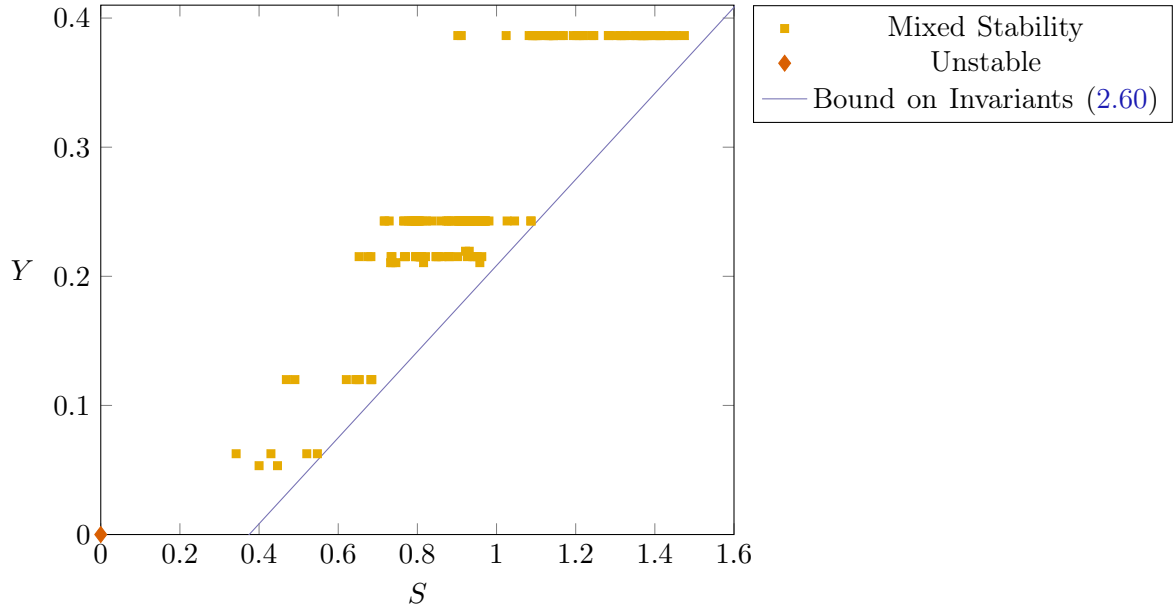

\begin{figure}[H]
\centering
\begin{tikzpicture}
\begin{axis}[
    xmin=0, xmax=3.3,
    ymin=0, ymax=1.1,
    xlabel= $S$,
    ylabel= $Y$,
    ylabel style={rotate=-90},
    title={Scalar-Fermion Fixed Points for $N_s=3$ and $N_f=4$},
    legend pos = outer north east
]
\addplot+[
    only marks,
    mark=square*,
    mark options={color=Dark2-F,fill=Dark2-F},
    mark size=1.5pt]
table{Datafiles/3,4/semistable.dat};
\addplot+[
    only marks,
    mark=diamond*,
    mark options={color=Dark2-B,fill=Dark2-B},
    mark size=3pt]
table{Datafiles/3,4/unstable.dat};
\addplot[domain=0:4, 
    samples=2, 
    color=Dark2-C]{x/3-3/24};
\legend{Stable, Mixed Stability, Unstable, Bound on Invariants (\ref{simplestbound})}
\end{axis}
\end{tikzpicture}
    \caption{Results of numerical search for $N_s=3$, $N_f=4$ beginning with 10,000 points. We find 518 distinct fixed points, none of which are stable.}
    \label{fig:34}
\end{figure}
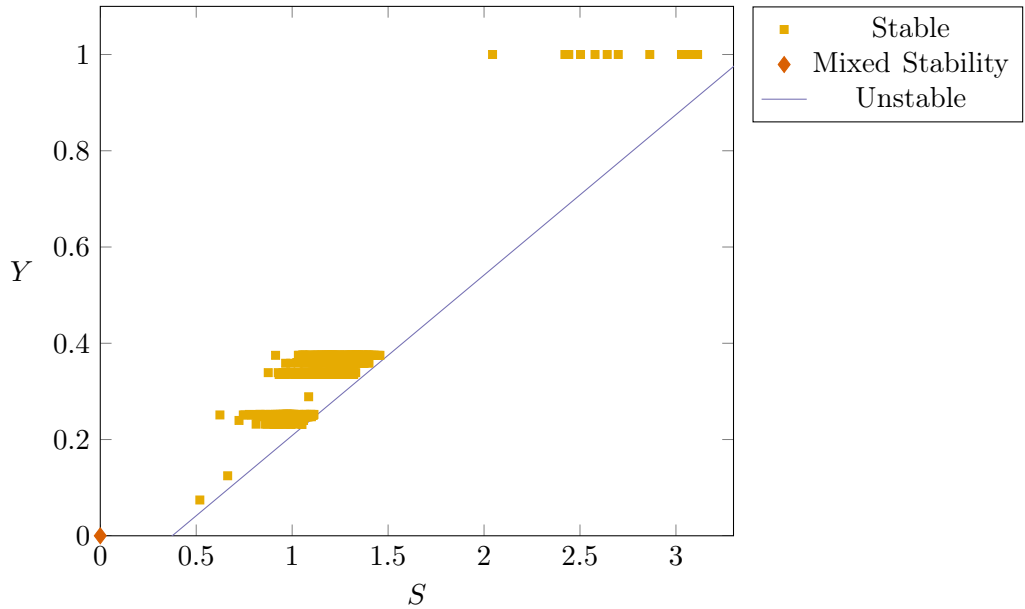

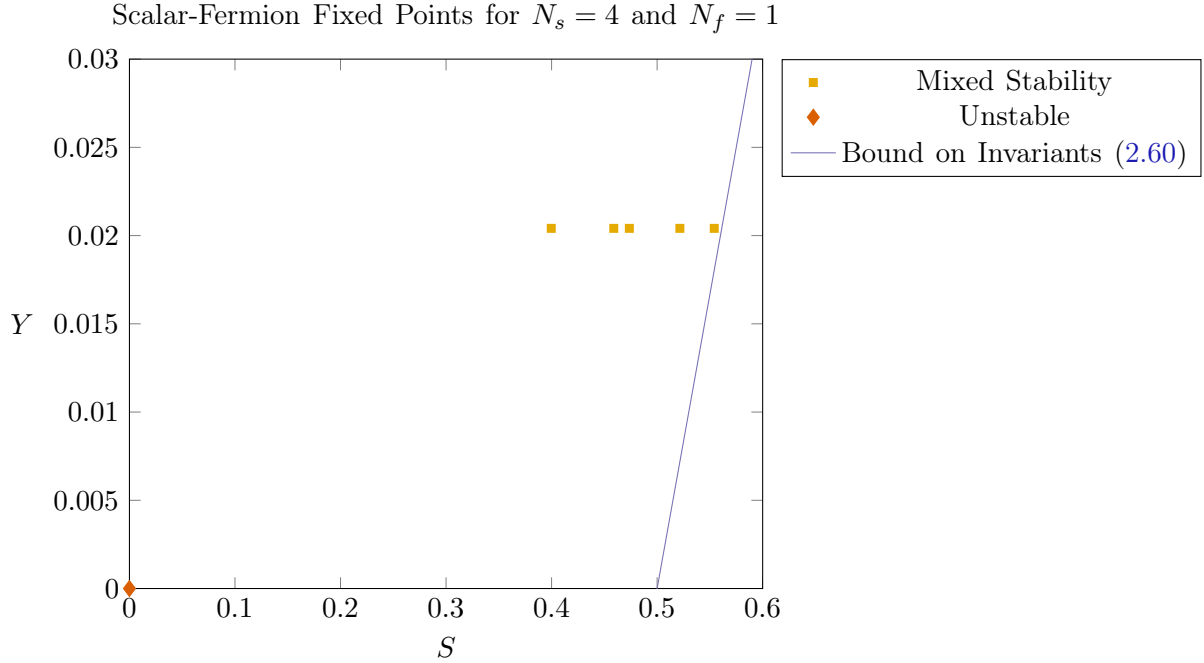
\begin{figure}[H]
\centering
\begin{tikzpicture}
\begin{axis}[
    xmin=0, xmax=0.6,
    ymin=0, ymax=0.03,
    xlabel= $S$,
    ylabel= $Y$,
    ylabel style={rotate=-90},
    yticklabel style={scaled ticks=false, /pgf/number format/fixed, /pgf/number format/precision=3},
    title={Scalar-Fermion Fixed Points for $N_s=4$ and $N_f=1$},
    legend pos = outer north east
]
\addplot+[
    only marks,
    mark=square*,
    mark options={color=Dark2-F,fill=Dark2-F},
    mark size=1.5pt]
table{Datafiles/4,1/semistable.dat};
\addplot+[
    only marks,
    mark=diamond*,
    mark options={color=Dark2-B,fill=Dark2-B},
    mark size=3pt]
table{Datafiles/4,1/unstable.dat};
\addplot[domain=0:5, 
    samples=2, 
    color=Dark2-C]{x/3-4/24};
\legend{Mixed Stability, Unstable, Bound on Invariants (\ref{simplestbound})}
\end{axis}
\end{tikzpicture}
    \caption{Results of numerical search for $N_s=4$, $N_f=1$ beginning with 10,000 points. We find 5 distinct fixed points, of which none are stable.}
    \label{fig:41}
\end{figure}

\begin{figure}[H]
\centering
\begin{tikzpicture}
\begin{axis}[
    xmin=0, xmax=1.5,
    ymin=0, ymax=0.3,
    xlabel= $S$,
    ylabel= $Y$,
    ylabel style={rotate=-90},
    yticklabel style={scaled ticks=false, /pgf/number format/fixed, /pgf/number format/precision=3},
    title={Scalar-Fermion Fixed Points for $N_s=4$ and $N_f=2$},
    legend pos = outer north east
]
\addplot+[
    only marks,
    mark=square*,
    mark options={color=Dark2-F,fill=Dark2-F},
    mark size=1.5pt]
table{Datafiles/4,2/semistable.dat};
\addplot+[
    only marks,
    mark=diamond*,
    mark options={color=Dark2-B,fill=Dark2-B},
    mark size=3pt]
table{Datafiles/4,2/unstable.dat};
\addplot[domain=0:2, 
    samples=2, 
    color=Dark2-C]{x/3-4/24};
\legend{Mixed Stability, Unstable, Bound on Invariants (\ref{simplestbound})}
\end{axis}
\end{tikzpicture}
    \caption{Results of numerical search for $N_s=4$, $N_f=2$ beginning with 10,000 points. We find 34 distinct fixed points, of which none are stable.}
    \label{fig:42}
\end{figure}
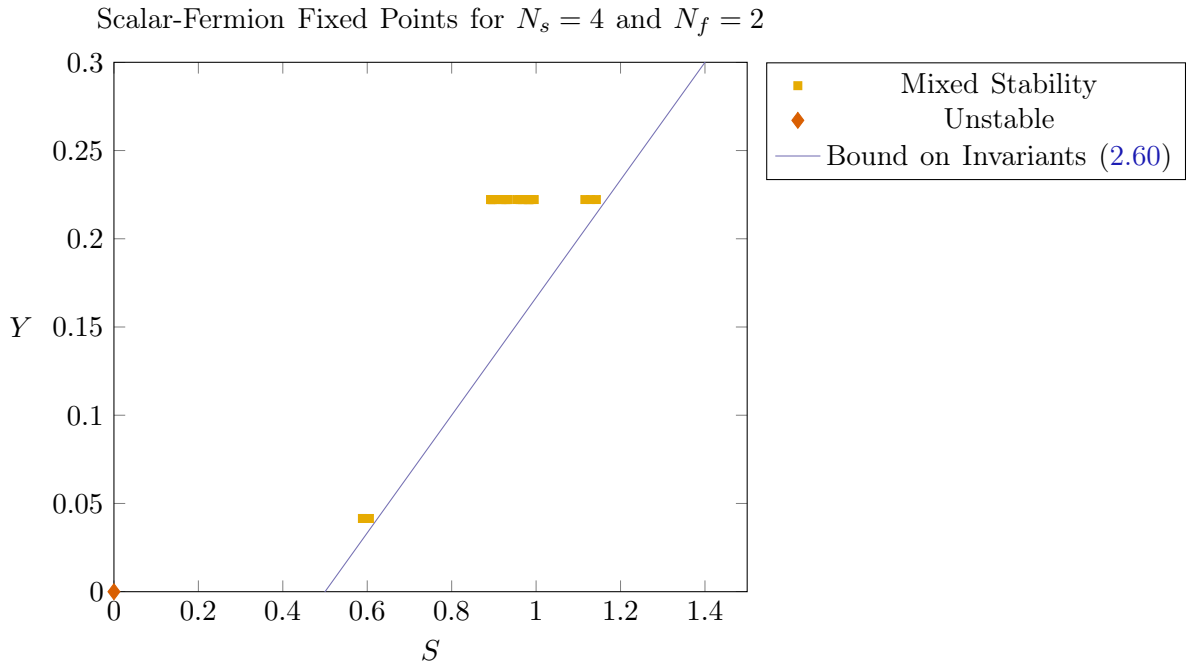

\begin{figure}[H]
\centering
\begin{tikzpicture}
\begin{axis}[
    xmin=0, xmax=1.7,
    ymin=0, ymax=0.45,
    xlabel= $S$,
    ylabel= $Y$,
    ylabel style={rotate=-90},
    title={Scalar-Fermion Fixed Points for $N_s=4$ and $N_f=3$},
    legend pos = outer north east
]
\addplot+[
    only marks,
    mark=square*,
    mark options={color=Dark2-F,fill=Dark2-F},
    mark size=1.5pt]
table{Datafiles/4,3/semistable.dat};
\addplot+[
    only marks,
    mark=diamond*,
    mark options={color=Dark2-B,fill=Dark2-B},
    mark size=3pt]
table{Datafiles/4,3/unstable.dat};
\addplot[domain=0:3, 
    samples=2, 
    color=Dark2-C]{x/3-4/24};
\legend{Mixed Stability, Unstable, Bound on Invariants (\ref{simplestbound})}
\end{axis}
\end{tikzpicture}
    \caption{Results of numerical search for $N_s=4$, $N_f=3$ beginning with 10,000 points. We find 88 distinct interacting fixed points, of which none are stable.}
    \label{fig:43}
\end{figure}
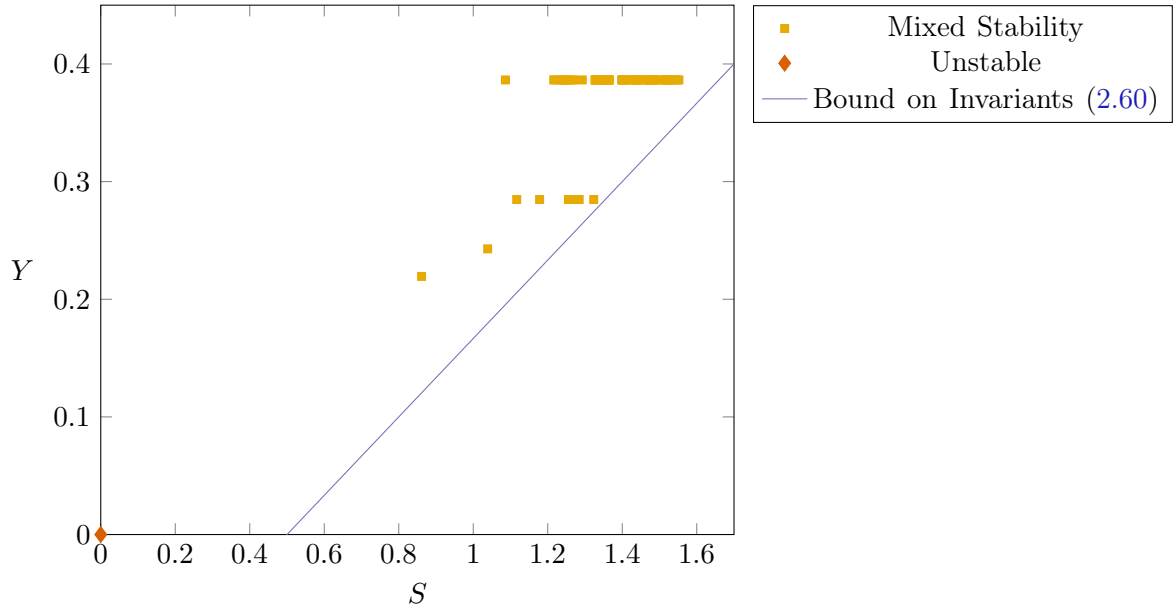

\begin{figure}[H]
\centering
\begin{tikzpicture}
\begin{axis}[
    xmin=0, xmax=2.2,
    ymin=0, ymax=0.6,
    xlabel= $S$,
    ylabel= $Y$,
    ylabel style={rotate=-90},
    title={Scalar-Fermion Fixed Points for $N_s=4$ and $N_f=4$},
    legend pos = outer north east
]
\addplot+[
    only marks,
    mark=square*,
    mark options={color=Dark2-F,fill=Dark2-F},
    mark size=1.5pt]
table{Datafiles/4,4/semistable.dat};
\addplot+[
    only marks,
    mark=diamond*,
    mark options={color=Dark2-B,fill=Dark2-B},
    mark size=3pt]
table{Datafiles/4,4/unstable.dat};
\addplot[domain=0:27, 
    samples=2, 
    color=Dark2-C]{x/3-4/24};
\legend{Mixed Stability, Unstable, Bound on Invariants (\ref{simplestbound})}
\end{axis}
\end{tikzpicture}
    \caption{Results of numerical search for $N_s=4$, $N_f=4$ beginning with 10,000 points. We find 240 distinct fixed points, of which none are stable.}
    \label{fig:44}
\end{figure}
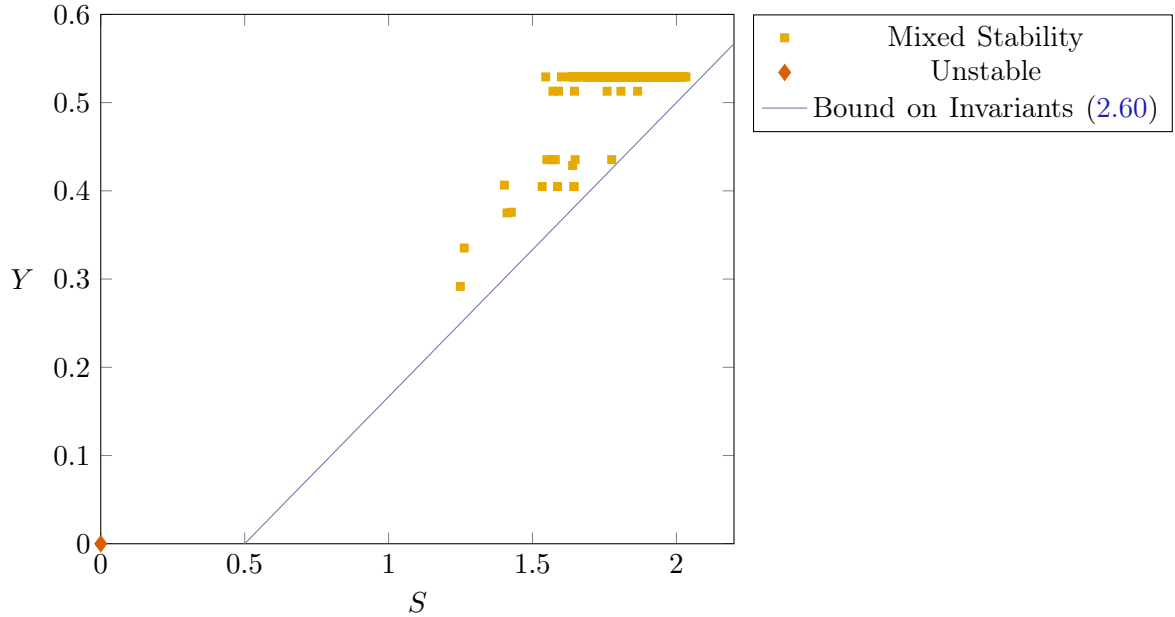

\FloatBarrier

\section{Conclusion}\label{conc}
Beginning with the beta functions for the couplings in a general scalar-fermion theory with a Yukawa-type interaction, we have derived bounds on linear combinations of coupling invariants which must be obeyed at all fixed points, beginning with either Dirac or Weyl fermions in four dimensions. When reduced to three dimensions, these theories, presuming we begin with an integer number of fermions, will always give an even number of three-dimensional Majorana fermions. Using the Weyl fermion beta functions as a basis, we have also considered a naive extension to generalise the results to include all possible Yukawa-type theories in three dimensions. We have also shown that, for a given solution to the Yukawa beta functions, there exists at most one stable fixed point.

While we presented the numerical results with a simple linear bound in the $S$-$Y$ plane using \eqref{simplestbound}, this is not the only way to organise the data. Instead, one could plot the fixed points in the $R$-$T'$ plane using the bound (\ref{parabolicbound}), which is a direct generalisation of a purely scalar bound\cite{Osborn:2020cnf}. We plot this bound for a few combinations of $N_s$ and $N_f$ in \cref{fig:33parabolic,fig:34parabolic,fig:43parabolic,fig:44parabolic}, which are analogous to Figs.\ 4 to 6 in \cite{Osborn:2020cnf}. Here, the fermionic levels do not appear as straight lines, and are instead distributed throughout the allowed region. Interestingly, one does not find the same clumping of fixed points as in the purely scalar case, with the points, even at each fermionic level, being distributed seemingly randomly.

\begin{figure}[H]
\centering
\begin{tikzpicture}
\begin{axis}[
    xmin=-1, xmax=1,
    ymin=-0.2, ymax=0.4,
    xlabel= $R$,
    ylabel= $T'$,
    ylabel style={rotate=-90},
    title={Scalar-Fermion Fixed Points for $N_s=3$ and $N_f=3$},
    legend pos = outer north east,
    colorbar,
    colormap/Spectral,
    colorbar style={ylabel=$Y$,ylabel style={rotate=-90}}
]
\addplot+[
    only marks,
    mark=*,
    scatter,
    mark size=2pt,
    scatter src = explicit]
table[x=R,y=T,meta=b4]{Datafiles/3,3/parabolicbound.dat};
\addplot[domain=-2:2, 
    samples=1000, 
    color=Dark2-C]{-x^2+3/8};
\end{axis}
\end{tikzpicture}
    \caption{Results of numerical search for $N_s=3$, $N_f=3$ beginning with 10,000 points. This includes a total of 171 fixed points, and does not include the free fixed point.}
\label{fig:33parabolic}
\end{figure}
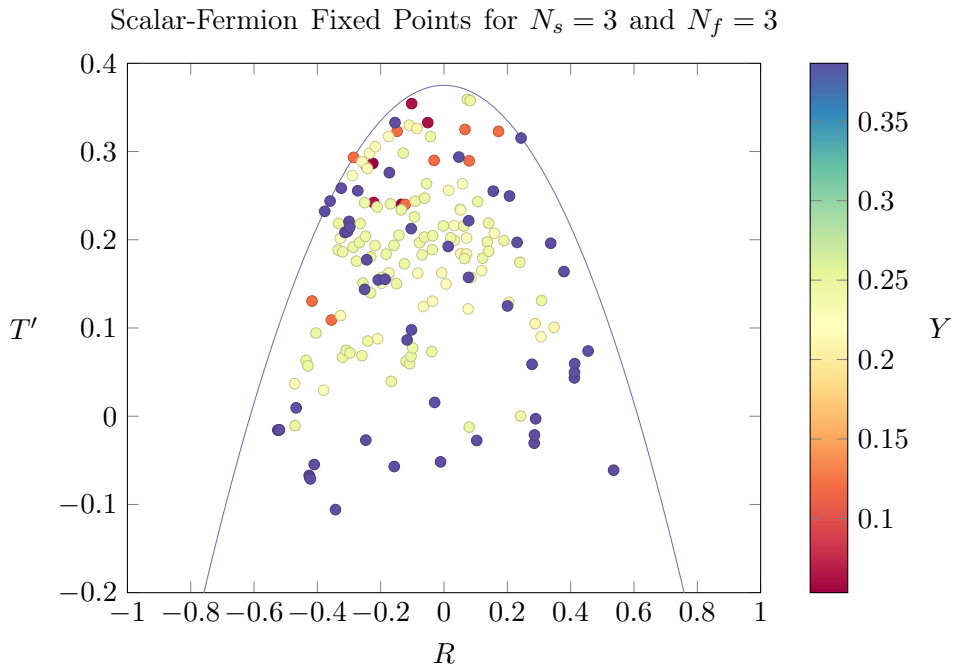

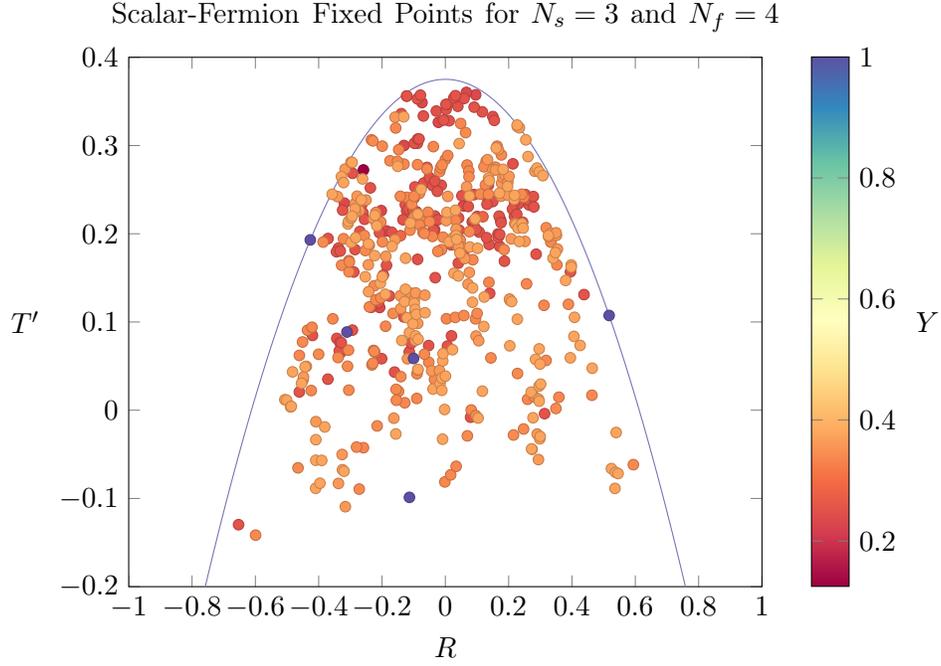
\begin{figure}[H]
\centering
\begin{tikzpicture}
\begin{axis}[
    xmin=-1, xmax=1,
    ymin=-0.2, ymax=0.4,
    xlabel= $R$,
    ylabel= $T'$,
    ylabel style={rotate=-90},
    title={Scalar-Fermion Fixed Points for $N_s=3$ and $N_f=4$},
    legend pos = outer north east,
    colorbar,
    colormap/Spectral,
    colorbar style={ylabel=$Y$,ylabel style={rotate=-90}}
]
\addplot+[
    only marks,
    mark=*,
    scatter,
    mark size=2pt,
    scatter src = explicit]
table[x=R,y=T,meta=b4]{Datafiles/3,4/parabolicbound.dat};
\addplot[domain=-2:2, 
    samples=1000, 
    color=Dark2-C]{-x^2+3/8};
\end{axis}
\end{tikzpicture}
    \caption{Results of numerical search for $N_s=3$, $N_f=4$ beginning with 10,000 points. This includes a total of 518 fixed points, and does not include the free fixed point.}
\label{fig:34parabolic}
\end{figure}

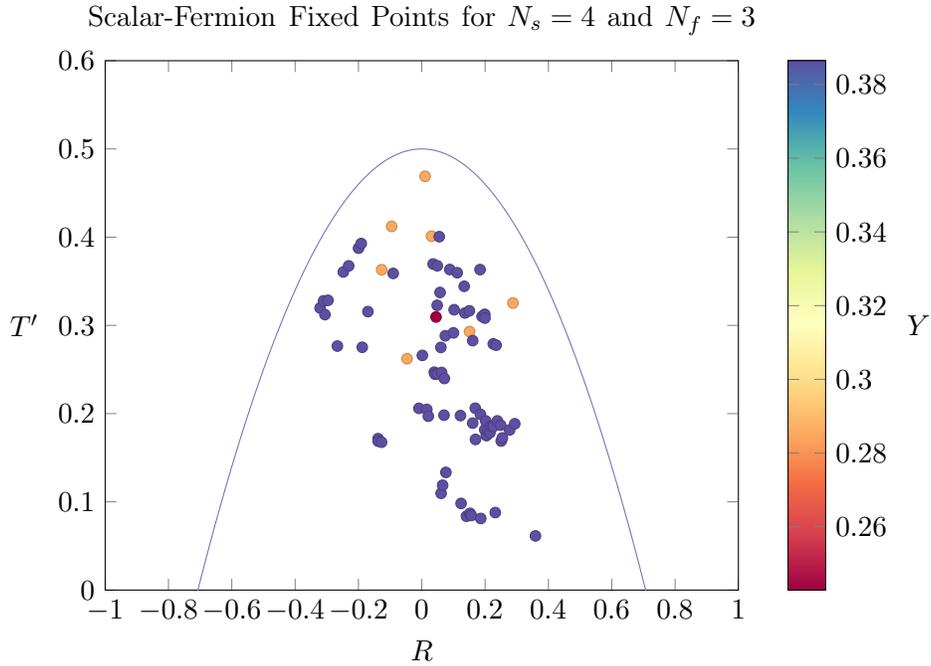
\begin{figure}[H]
\centering
\begin{tikzpicture}
\begin{axis}[
    xmin=-1, xmax=1,
    ymin=0, ymax=0.6,
    xlabel= $R$,
    ylabel= $T'$,
    ylabel style={rotate=-90},
    title={Scalar-Fermion Fixed Points for $N_s=4$ and $N_f=3$},
    legend pos = outer north east,
    colorbar,
    colormap/Spectral,
    colorbar style={ylabel=$Y$,ylabel style={rotate=-90}}
]
\addplot+[
    only marks,
    mark=*,
    scatter,
    mark size=2pt,
    scatter src = explicit]
table[x=R,y=T,meta=b4]{Datafiles/4,3/parabolicbound.dat};
\addplot[domain=-2:2, 
    samples=1000, 
    color=Dark2-C]{-x^2+4/8};
\end{axis}
\end{tikzpicture}
    \caption{Results of numerical search for $N_s=4$, $N_f=3$ beginning with 10,000 points. This includes a total of 88 fixed points, and does not include the free fixed point.}
\label{fig:43parabolic}
\end{figure}

\begin{figure}[H]
\centering
\begin{tikzpicture}
\begin{axis}[
    xmin=-0.6, xmax=0.6,
    ymin=-0.1, ymax=0.6,
    xlabel= $R$,
    ylabel= $T'$,
    ylabel style={rotate=-90},
    title={Scalar-Fermion Fixed Points for $N_s=4$ and $N_f=4$},
    legend pos = outer north east,
    colorbar,
    colormap/Spectral,
    colorbar style={ylabel=$Y$,ylabel style={rotate=-90}}
]
\addplot+[
    only marks,
    mark=*,
    scatter,
    mark size=2pt,
    scatter src = explicit]
table[x=R,y=T,meta=b4]{Datafiles/4,4/parabolicbound.dat};
\addplot[domain=-2:2, 
    samples=1000, 
    color=Dark2-C]{-x^2+4/8};
\end{axis}
\end{tikzpicture}
    \caption{Results of numerical search for $N_s=4$, $N_f=4$ beginning with 10,000 points. This includes a total of 240 fixed points, and does not include the free fixed point.}
\label{fig:44parabolic}
\end{figure}
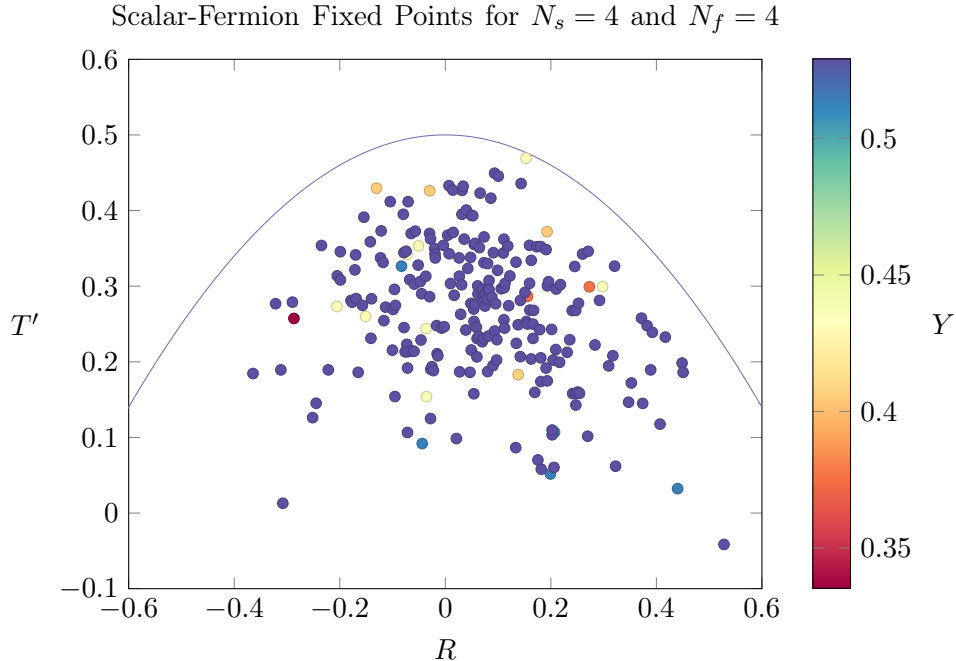

We have searched for all possible scalar-fermion fixed points both analytically and numerically, but we cannot guarantee that our numerical searches have found all existing fixed points. The numerical methods used extend very naturally to higher $N_s$ and $N_f$, and we were only prevented from listing more numeric results by a lack of space. For large enough $N_s$ and $N_f$, \emph{Mathematica} may no longer have sufficient power to efficiently canvas the entire space of fixed points, at which point it would be necessary to consider alternatives.

Statements made about the properties of beta functions for arbitrary potentials become statements about a wide class of theories. What these statements may lack in depth and specificity is made up for in the breadth of their application, especially as the potential used becomes more and more general. To that end, one could generalise the Lagrangian we began with in a number of different ways. There exist important models, for example the chiral Heisenberg model~\cite{Zerf:2017zqi,Rosenstein:1993zf} for $N_s=3$ scalars, which couple scalars and fermions in a more non-trivial way. These are not encompassed by (\ref{basicLag}) in the general case. One could also attempt to add spin-one gauge fields into the system, either by gauging one of the global symmetries or by introducing a fresh gauge group. Such a system would produce new fixed points for which the gauge coupling would be non-zero.

\ack{We would like to thank Hugh Osborn for enlightening discussions and comments on the manuscript. AS is funded by the Royal Society under grant URF{\textbackslash}R1{\textbackslash}211417.}

\FloatBarrier

\begin{appendices}

\section{Fixed point data for small \texorpdfstring{$\boldsymbol{N_s,N_f}$}{N\_s,N\_f}}\label{dataappendix}
Let us here add a remark concerning the number of $\kappa=0$ eigenvalues. The equation
\begin{equation}
    S_{IJ}v_J=0
\end{equation} for some non-zero vector $v_J$ corresponds to a non-trivial constraint which the tensors $\lambda_{ijkl}$ and $y_{iab}$ must obey at the fixed point. From \cref{stabmatrixeigenvalues} we see that broken symmetry generators can lead to these constraints, but one can additionally imagine that the relatively simple form of the one-loop beta-functions could allow additional 'accidental' constraints. In these additional cases the extra zero $\kappa$ eigenvalue will be lifted at higher loop order because they are not related to symmetry. For example for $N_s=2, N_f=3$ we can have a maximum of four broken generators of the free-theory global symmetry group $O(2)\times O(3)$. However, in \cref{tab:23} there are entries with five zero $\kappa$ eigenvalues. By considering the two-loop beta-functions one can indeed verify that these accidental eigenvalues are lifted at higher loop order, and that only the four zero $\kappa$ eigenvalues remain. Similar comments apply to \cref{tab:24}, and will hold also for points with larger numbers of scalars and fermions.

\begin{center}
\begin{longtable}{|c c c c c c|}\caption{Fixed points found for $N_s=N_f=2$.} \\
\hline
$S$ & $Y$ & \makecell{$\#$ different \\[-3pt] $\gamma_\phi$(degeneracies)} & \makecell{$\#$ different \\[-3pt] $\gamma_\psi$(degeneracies)} & $\#\,\kappa<0$, =0  & $a_2,\,b_1\neq0$ \\ [0.5ex] 
 \hline
\endfirsthead

\multicolumn{6}{c}%
{{\bfseries \tablename\ \thetable{} -- continued from previous page}} \\
\hline
$S$ & $Y$ & \makecell{$\#$ different \\[-3pt] $\gamma_\phi$(degeneracies)} & \makecell{$\#$ different \\[-3pt] $\gamma_\psi$(degeneracies)} & $\#\,\kappa<0$, =0  & $a_2,\,b_1\neq0$\\ [0.5ex] 
 \hline
\endhead

\hline \multicolumn{6}{|r|}{{Continued on next page}} \\ \hline
\endfoot

\hline
\endlastfoot \label{tab:22n}
$\!\!0.178588$ & $\frac{1}{32}$ & $\text{2(1,1)}$ & $\text{1(2)}$ & $\text{6, 2}$ & $\star,\,\star$ \\
$0.178588$ & $\frac{1}{32}$ & $\text{2(1,1)}$ & $\text{1(2)}$ & $\text{6, 1}$ & $\star,\,\star$ \\
$0.225456$ & $\frac{1}{32}$ & $\text{2(1,1)}$ & $\text{1(2)}$ & $\text{4, 2}$ & $\star,\,\star$ \\
$0.225456$ & $\frac{1}{32}$ & $\text{2(1,1)}$ & $\text{1(2)}$ & $\text{4, 1}$ & $\star,\,\star$ \\
$0.310693$ & $\frac{1}{32}$ & $\text{2(1,1)}$ & $\text{1(2)}$ & $\text{3, 2}$ & $\star,\,\star$ \\
$0.310693$ & $\frac{1}{32}$ & $\text{2(1,1)}$ & $\text{1(2)}$ & $\text{3, 1}$ & $\star,\,\star$ \\
$0.0712582$ & $\frac{493}{11881}$ & $\text{2(1,1)}$ & $\text{2(1,1)}$ & $\text{6, 2}$ & $\star,\,\star$ \\
$0.194755$ & $\frac{493}{11881}$ & $\text{2(1,1)}$ & $\text{2(1,1)}$ & $\text{5, 2}$ & $\star,\,\star$ \\
$0.194755$ & $\frac{493}{11881}$ & $\text{2(1,1)}$ & $\text{2(1,1)}$ & $\text{5, 2}$ & $\star,\,\star$ \\
$0.210175$ & $\frac{493}{11881}$ & $\text{2(1,1)}$ & $\text{2(1,1)}$ & $\text{5, 2}$ & $\star,\,\star$ \\
$0.217217$ & $\frac{493}{11881}$ & $\text{2(1,1)}$ & $\text{2(1,1)}$ & $\text{4, 2}$ & $\star,\,\star$ \\
$0.287836$ & $\frac{493}{11881}$ & $\text{2(1,1)}$ & $\text{2(1,1)}$ & $\text{3, 2}$ & $\star,\,\star$ \\
$0.327445$ & $\frac{493}{11881}$ & $\text{2(1,1)}$ & $\text{2(1,1)}$ & $\text{4, 2}$ & $\star,\,\star$ \\
$0.363636$ & $\frac{493}{11881}$ & $\text{2(1,1)}$ & $\text{2(1,1)}$ & $\text{1, 2}$ & $\star,\,\star$ \\
$\frac{4329}{11881}$ & $\frac{493}{11881}$ & $\text{2(1,1)}$ & $\text{2(1,1)}$ & $\text{2, 2}$ & $\star,\,\star$ \\
$\frac{32}{75}$ & $\frac{2}{9}$ & $\text{1(2)}$ & $\text{1(2)}$ & $\text{5, 1}$ & $\text{}$ \\
$0.632444$ & $\frac{2}{9}$ & $\text{1(2)}$ & $\text{1(2)}$ & $\text{3, 2}$ & $\star,$ \\
$\frac{2}{3}$ & $\frac{2}{9}$ & $\text{1(2)}$ & $\text{1(2)}$ & $\text{0, 1}$ & $\text{}$ \\
$0.737926$ & $\frac{2}{9}$ & $\text{1(2)}$ & $\text{1(2)}$ & $\text{1, 2}$ & $\star,$ \\
\end{longtable}
\end{center}

\begin{center}
\begin{longtable}{|c c c c c c|}\caption{Fixed points found for $N_s=2$ and $N_f=3$.} \\
\hline
$S$ & $Y$ & \makecell{$\#$ different \\[-3pt] $\gamma_\phi$(degeneracies)} & \makecell{$\#$ different \\[-3pt] $\gamma_\psi$(degeneracies)} & $\#\,\kappa<0$, =0 & $a_2,\,b_1\neq0$ \\ [0.5ex] 
 \hline
\endfirsthead

\multicolumn{6}{c}%
{{\bfseries \tablename\ \thetable{} -- continued from previous page}} \\
\hline
$S$ & $Y$ & \makecell{$\#$ different \\[-3pt] $\gamma_\phi$(degeneracies)} & \makecell{$\#$ different \\[-3pt] $\gamma_\psi$(degeneracies)} & $\#\,\kappa<0$, =0  & $a_2,\,b_1\neq0$ \\ [0.5ex] 
 \hline
\endhead

\hline \multicolumn{6}{|r|}{{Continued on next page}} \\ \hline
\endfoot

\hline
\endlastfoot \label{tab:23}
$\!\!0.310449$ & $\frac{1}{27}$ & $\text{2(1,1)}$ & $\text{1(3)}$ & $\text{6, 3}$ & $\star,\,\star$ \\
$0.310719$ & $\frac{1}{27}$ & $\text{2(1,1)}$ & $\text{1(3)}$ & $\text{7, 3}$ & $\star,\,\star$ \\
$0.10243$ & $\frac{1191}{22472}$ & $\text{2(1,1)}$ & $\text{3(1,1,1)}$ & $\text{8, 4}$ & $\star,\,\star$ \\
$0.238032$ & $\frac{1191}{22472}$ & $\text{2(1,1)}$ & $\text{3(1,1,1)}$ & $\text{6, 4}$ & $\star,\,\star$ \\
$0.239854$ & $\frac{1191}{22472}$ & $\text{2(1,1)}$ & $\text{3(1,1,1)}$ & $\text{7, 4}$ & $\star,\,\star$ \\
$0.344635$ & $\frac{1191}{22472}$ & $\text{2(1,1)}$ & $\text{3(1,1,1)}$ & $\text{4, 4}$ & $\star,\,\star$ \\
$0.359239$ & $\frac{1191}{22472}$ & $\text{2(1,1)}$ & $\text{3(1,1,1)}$ & $\text{5, 4}$ & $\star,\,\star$ \\
$0.365948$ & $\frac{1191}{22472}$ & $\text{2(1,1)}$ & $\text{3(1,1,1)}$ & $\text{3, 4}$ & $\star,\,\star$ \\
$0.380859$ & $\frac{1191}{22472}$ & $\text{2(1,1)}$ & $\text{3(1,1,1)}$ & $\text{6, 4}$ & $\star,\,\star$ \\
$0.0993507$ & $\frac{4}{75}$ & $\text{1(2)}$ & $\text{1(3)}$ & $\text{8, 4}$ & $\text{}$ \\
$\frac{6}{25}$ & $\frac{4}{75}$ & $\text{1(2)}$ & $\text{1(3)}$ & $\text{6, 5}$ & $\star,$ \\
$\frac{6}{25}$ & $\frac{4}{75}$ & $\text{1(2)}$ & $\text{1(3)}$ & $\text{4, 5}$ & $\star,$ \\
$\frac{9}{25}$ & $\frac{4}{75}$ & $\text{1(2)}$ & $\text{1(3)}$ & $\text{4, 5}$ & $\star,$ \\
$\frac{9}{25}$ & $\frac{4}{75}$ & $\text{1(2)}$ & $\text{1(3)}$ & $\text{2, 5}$ & $\star,$ \\
$0.371049$ & $\frac{4}{75}$ & $\text{1(2)}$ & $\text{1(3)}$ & $\text{3, 4}$ & $\text{}$ \\
$0.371049$ & $\frac{4}{75}$ & $\text{1(2)}$ & $\text{1(3)}$ & $\text{1, 4}$ & $\text{}$ \\
$\frac{2}{5}$ & $\frac{4}{75}$ & $\text{1(2)}$ & $\text{1(3)}$ & $\text{5, 5}$ & $\text{}$ \\
$\frac{2}{5}$ & $\frac{4}{75}$ & $\text{1(2)}$ & $\text{1(3)}$ & $\text{3, 5}$ & $\text{}$ \\
$0.136488$ & $\frac{2}{27}$ & $\text{1(2)}$ & $\text{2(1,2)}$ & $\text{7, 3}$ & $\text{}$ \\
$0.287065$ & $\frac{2}{27}$ & $\text{1(2)}$ & $\text{2(1,2)}$ & $\text{5, 4}$ & $\star,$ \\
$0.408408$ & $\frac{2}{27}$ & $\text{1(2)}$ & $\text{2(1,2)}$ & $\text{3, 4}$ & $\star,$ \\
$0.411661$ & $\frac{2}{27}$ & $\text{1(2)}$ & $\text{2(1,2)}$ & $\text{2, 3}$ & $\text{}$ \\
$\frac{338}{729}$ & $\frac{2}{27}$ & $\text{1(2)}$ & $\text{2(1,2)}$ & $\text{4, 4}$ & $\text{}$ \\
$0.432574$ & $\frac{2148}{10201}$ & $\text{2(1,1)}$ & $\text{2(2,1)}$ & $\text{5, 4}$ & $\star,\,\star$ \\
$0.432574$ & $\frac{2148}{10201}$ & $\text{2(1,1)}$ & $\text{2(2,1)}$ & $\text{5, 4}$ & $\star,\,\star$ \\
$0.595522$ & $\frac{2148}{10201}$ & $\text{2(1,1)}$ & $\text{2(2,1)}$ & $\text{3, 4}$ & $\star,\,\star$ \\
$0.621098$ & $\frac{2148}{10201}$ & $\text{2(1,1)}$ & $\text{2(2,1)}$ & $\text{0, 4}$ & $\star,\,\star$ \\
$0.624528$ & $\frac{2148}{10201}$ & $\text{2(1,1)}$ & $\text{2(2,1)}$ & $\text{4, 4}$ & $\star,\,\star$ \\
$0.63302$ & $\frac{2148}{10201}$ & $\text{2(1,1)}$ & $\text{2(2,1)}$ & $\text{3, 4}$ & $\star,\,\star$ \\
$0.670493$ & $\frac{2148}{10201}$ & $\text{2(1,1)}$ & $\text{2(2,1)}$ & $\text{2, 4}$ & $\star,\,\star$ \\
$0.703234$ & $\frac{2148}{10201}$ & $\text{2(1,1)}$ & $\text{2(2,1)}$ & $\text{1, 4}$ & $\star,\,\star$ \\
$0.718282$ & $\frac{2148}{10201}$ & $\text{2(1,1)}$ & $\text{2(2,1)}$ & $\text{2, 4}$ & $\star,\,\star$ \\
$0.877832$ & $\frac{2148}{10201}$ & $\text{2(1,1)}$ & $\text{2(2,1)}$ & $\text{3, 4}$ & $\star,\,\star$ \\
\end{longtable}
\end{center}

\begin{center}
\begin{longtable}{|c c c c c c|}\caption{Fixed points found for $N_s=2$ and $N_f=4$.} \\
\hline
$S$ & $Y$ & \makecell{$\#$ different \\[-3pt] $\gamma_\phi$(degeneracies)} & \makecell{$\#$ different \\[-3pt] $\gamma_\psi$(degeneracies)} & $\#\,\kappa<0$, =0  & $a_2,\,b_1\neq0$\\ [0.5ex] 
 \hline
\endfirsthead

\multicolumn{6}{c}%
{{\bfseries \tablename\ \thetable{} -- continued from previous page}} \\
\hline
$S$ & $Y$ & \makecell{$\#$ different \\[-3pt] $\gamma_\phi$(degeneracies)} & \makecell{$\#$ different \\[-3pt] $\gamma_\psi$(degeneracies)} & $\#\,\kappa<0$, =0 & $a_2,\,b_1\neq0$ \\ [0.5ex] 
 \hline
\endhead

\hline \multicolumn{6}{|r|}{{Continued on next page}} \\ \hline
\endfoot

\hline
\endlastfoot \label{tab:24}
$0.135$ & $\frac{1}{16}$ & $\text{1(2)}$ & $\text{1(4)}$ & $\text{6, 8}$ & $\text{}$ \\
$0.135267$ & $\frac{1}{16}$ & $\text{1(2)}$ & $\text{1(4)}$ & $\text{6, 8}$ & $\text{}$ \\
$0.136281$ & $\frac{1}{16}$ & $\text{1(2)}$ & $\text{1(4)}$ & $\text{6, 8}$ & $\text{}$ \\
$0.140128$ & $\frac{1}{16}$ & $\text{1(2)}$ & $\text{1(4)}$ & $\text{6, 8}$ & $\text{}$ \\
$0.141986$ & $\frac{1}{16}$ & $\text{1(2)}$ & $\text{1(4)}$ & $\text{6, 8}$ & $\text{}$ \\
$0.144275$ & $\frac{1}{16}$ & $\text{1(2)}$ & $\text{1(4)}$ & $\text{6, 8}$ & $\text{}$ \\
$0.153065$ & $\frac{1}{16}$ & $\text{1(2)}$ & $\text{1(4)}$ & $\text{5, 9}$ & $\text{}$ \\
$0.154858$ & $\frac{1}{16}$ & $\text{1(2)}$ & $\text{1(4)}$ & $\text{5, 8}$ & $\text{}$ \\
$0.26852$ & $\frac{1}{16}$ & $\text{1(2)}$ & $\text{1(4)}$ & $\text{4, 8}$ & $\star,$ \\
$0.268533$ & $\frac{1}{16}$ & $\text{1(2)}$ & $\text{1(4)}$ & $\text{6, 8}$ & $\star,$ \\
$0.268544$ & $\frac{1}{16}$ & $\text{1(2)}$ & $\text{1(4)}$ & $\text{4, 8}$ & $\star,$ \\
$0.268552$ & $\frac{1}{16}$ & $\text{1(2)}$ & $\text{1(4)}$ & $\text{4, 8}$ & $\star,$ \\
$0.268581$ & $\frac{1}{16}$ & $\text{1(2)}$ & $\text{1(4)}$ & $\text{4, 8}$ & $\star,$ \\
$0.268604$ & $\frac{1}{16}$ & $\text{1(2)}$ & $\text{1(4)}$ & $\text{4, 8}$ & $\star,$ \\
$0.268621$ & $\frac{1}{16}$ & $\text{1(2)}$ & $\text{1(4)}$ & $\text{4, 8}$ & $\star,$ \\
$0.26865$ & $\frac{1}{16}$ & $\text{1(2)}$ & $\text{1(4)}$ & $\text{4, 8}$ & $\star,$ \\
$0.268734$ & $\frac{1}{16}$ & $\text{1(2)}$ & $\text{1(4)}$ & $\text{6, 8}$ & $\star,$ \\
$0.268745$ & $\frac{1}{16}$ & $\text{1(2)}$ & $\text{1(4)}$ & $\text{4, 8}$ & $\star,$ \\
$0.268916$ & $\frac{1}{16}$ & $\text{1(2)}$ & $\text{1(4)}$ & $\text{4, 8}$ & $\star,$ \\
$0.268988$ & $\frac{1}{16}$ & $\text{1(2)}$ & $\text{1(4)}$ & $\text{4, 8}$ & $\star,$ \\
$0.269024$ & $\frac{1}{16}$ & $\text{1(2)}$ & $\text{1(4)}$ & $\text{4, 8}$ & $\star,$ \\
$0.269032$ & $\frac{1}{16}$ & $\text{1(2)}$ & $\text{1(4)}$ & $\text{6, 8}$ & $\star,$ \\
$0.269074$ & $\frac{1}{16}$ & $\text{1(2)}$ & $\text{1(4)}$ & $\text{4, 8}$ & $\star,$ \\
$0.269203$ & $\frac{1}{16}$ & $\text{1(2)}$ & $\text{1(4)}$ & $\text{4, 8}$ & $\star,$ \\
$0.269254$ & $\frac{1}{16}$ & $\text{1(2)}$ & $\text{1(4)}$ & $\text{4, 8}$ & $\star,$ \\
$0.269337$ & $\frac{1}{16}$ & $\text{1(2)}$ & $\text{1(4)}$ & $\text{6, 8}$ & $\star,$ \\
$0.26949$ & $\frac{1}{16}$ & $\text{1(2)}$ & $\text{1(4)}$ & $\text{4, 8}$ & $\star,$ \\
$0.269497$ & $\frac{1}{16}$ & $\text{1(2)}$ & $\text{1(4)}$ & $\text{4, 8}$ & $\star,$ \\
$0.269553$ & $\frac{1}{16}$ & $\text{1(2)}$ & $\text{1(4)}$ & $\text{4, 8}$ & $\star,$ \\
$0.269596$ & $\frac{1}{16}$ & $\text{1(2)}$ & $\text{1(4)}$ & $\text{6, 8}$ & $\star,$ \\
$0.26963$ & $\frac{1}{16}$ & $\text{1(2)}$ & $\text{1(4)}$ & $\text{4, 8}$ & $\star,$ \\
$0.269789$ & $\frac{1}{16}$ & $\text{1(2)}$ & $\text{1(4)}$ & $\text{6, 8}$ & $\star,$ \\
$0.269851$ & $\frac{1}{16}$ & $\text{1(2)}$ & $\text{1(4)}$ & $\text{4, 8}$ & $\star,$ \\
$0.269866$ & $\frac{1}{16}$ & $\text{1(2)}$ & $\text{1(4)}$ & $\text{4, 8}$ & $\star,$ \\
$0.270031$ & $\frac{1}{16}$ & $\text{1(2)}$ & $\text{1(4)}$ & $\text{6, 8}$ & $\star,$ \\
$0.270176$ & $\frac{1}{16}$ & $\text{1(2)}$ & $\text{1(4)}$ & $\text{5, 8}$ & $\star,$ \\
$0.270514$ & $\frac{1}{16}$ & $\text{1(2)}$ & $\text{1(4)}$ & $\text{5, 8}$ & $\star,$ \\
$0.270533$ & $\frac{1}{16}$ & $\text{1(2)}$ & $\text{1(4)}$ & $\text{5, 8}$ & $\star,$ \\
$0.270909$ & $\frac{1}{16}$ & $\text{1(2)}$ & $\text{1(4)}$ & $\text{5, 8}$ & $\star,$ \\
$0.271088$ & $\frac{1}{16}$ & $\text{1(2)}$ & $\text{1(4)}$ & $\text{4, 8}$ & $\star,$ \\
$0.271141$ & $\frac{1}{16}$ & $\text{1(2)}$ & $\text{1(4)}$ & $\text{5, 8}$ & $\star,$ \\
$0.271873$ & $\frac{1}{16}$ & $\text{1(2)}$ & $\text{1(4)}$ & $\text{4, 8}$ & $\star,$ \\
$0.271925$ & $\frac{1}{16}$ & $\text{1(2)}$ & $\text{1(4)}$ & $\text{6, 8}$ & $\star,$ \\
$0.272103$ & $\frac{1}{16}$ & $\text{1(2)}$ & $\text{1(4)}$ & $\text{4, 8}$ & $\star,$ \\
$0.272212$ & $\frac{1}{16}$ & $\text{1(2)}$ & $\text{1(4)}$ & $\text{5, 8}$ & $\star,$ \\
$0.272429$ & $\frac{1}{16}$ & $\text{1(2)}$ & $\text{1(4)}$ & $\text{7, 8}$ & $\star,$ \\
$0.272563$ & $\frac{1}{16}$ & $\text{1(2)}$ & $\text{1(4)}$ & $\text{5, 8}$ & $\star,$ \\
$0.27281$ & $\frac{1}{16}$ & $\text{1(2)}$ & $\text{1(4)}$ & $\text{5, 8}$ & $\star,$ \\
$0.273018$ & $\frac{1}{16}$ & $\text{1(2)}$ & $\text{1(4)}$ & $\text{4, 8}$ & $\star,$ \\
$0.273331$ & $\frac{1}{16}$ & $\text{1(2)}$ & $\text{1(4)}$ & $\text{4, 8}$ & $\star,$ \\
$0.274007$ & $\frac{1}{16}$ & $\text{1(2)}$ & $\text{1(4)}$ & $\text{5, 8}$ & $\star,$ \\
$0.274095$ & $\frac{1}{16}$ & $\text{1(2)}$ & $\text{1(4)}$ & $\text{4, 8}$ & $\star,$ \\
$0.274136$ & $\frac{1}{16}$ & $\text{1(2)}$ & $\text{1(4)}$ & $\text{3, 9}$ & $\star,$ \\
$0.27426$ & $\frac{1}{16}$ & $\text{1(2)}$ & $\text{1(4)}$ & $\text{5, 8}$ & $\star,$ \\
$0.274345$ & $\frac{1}{16}$ & $\text{1(2)}$ & $\text{1(4)}$ & $\text{3, 8}$ & $\star,$ \\
$0.27451$ & $\frac{1}{16}$ & $\text{1(2)}$ & $\text{1(4)}$ & $\text{3, 8}$ & $\star,$ \\
$0.274924$ & $\frac{1}{16}$ & $\text{1(2)}$ & $\text{1(4)}$ & $\text{3, 8}$ & $\star,$ \\
$0.275524$ & $\frac{1}{16}$ & $\text{1(2)}$ & $\text{1(4)}$ & $\text{4, 8}$ & $\star,$ \\
$0.276723$ & $\frac{1}{16}$ & $\text{1(2)}$ & $\text{1(4)}$ & $\text{3, 8}$ & $\star,$ \\
$0.27687$ & $\frac{1}{16}$ & $\text{1(2)}$ & $\text{1(4)}$ & $\text{5, 8}$ & $\star,$ \\
$0.277346$ & $\frac{1}{16}$ & $\text{1(2)}$ & $\text{1(4)}$ & $\text{5, 8}$ & $\star,$ \\
$0.277617$ & $\frac{1}{16}$ & $\text{1(2)}$ & $\text{1(4)}$ & $\text{5, 8}$ & $\star,$ \\
$0.277671$ & $\frac{1}{16}$ & $\text{1(2)}$ & $\text{1(4)}$ & $\text{5, 8}$ & $\star,$ \\
$0.277844$ & $\frac{1}{16}$ & $\text{1(2)}$ & $\text{1(4)}$ & $\text{5, 8}$ & $\star,$ \\
$0.278128$ & $\frac{1}{16}$ & $\text{1(2)}$ & $\text{1(4)}$ & $\text{5, 8}$ & $\star,$ \\
$0.278186$ & $\frac{1}{16}$ & $\text{1(2)}$ & $\text{1(4)}$ & $\text{5, 8}$ & $\star,$ \\
$0.278477$ & $\frac{1}{16}$ & $\text{1(2)}$ & $\text{1(4)}$ & $\text{5, 8}$ & $\star,$ \\
$0.278763$ & $\frac{1}{16}$ & $\text{1(2)}$ & $\text{1(4)}$ & $\text{5, 8}$ & $\star,$ \\
$0.278953$ & $\frac{1}{16}$ & $\text{1(2)}$ & $\text{1(4)}$ & $\text{5, 8}$ & $\star,$ \\
$0.279605$ & $\frac{1}{16}$ & $\text{1(2)}$ & $\text{1(4)}$ & $\text{5, 10}$ & $\star,$ \\
$0.279754$ & $\frac{1}{16}$ & $\text{1(2)}$ & $\text{1(4)}$ & $\text{4, 8}$ & $\star,$ \\
$0.280278$ & $\frac{1}{16}$ & $\text{1(2)}$ & $\text{1(4)}$ & $\text{5, 8}$ & $\star,$ \\
$0.280656$ & $\frac{1}{16}$ & $\text{1(2)}$ & $\text{1(4)}$ & $\text{5, 8}$ & $\star,$ \\
$0.307795$ & $\frac{1}{16}$ & $\text{1(2)}$ & $\text{1(4)}$ & $\text{3, 8}$ & $\star,$ \\
$0.308871$ & $\frac{1}{16}$ & $\text{1(2)}$ & $\text{1(4)}$ & $\text{4, 8}$ & $\star,$ \\
$0.312307$ & $\frac{1}{16}$ & $\text{1(2)}$ & $\text{1(4)}$ & $\text{2, 8}$ & $\star,$ \\
$0.314198$ & $\frac{1}{16}$ & $\text{1(2)}$ & $\text{1(4)}$ & $\text{3, 10}$ & $\star,$ \\
$0.314269$ & $\frac{1}{16}$ & $\text{1(2)}$ & $\text{1(4)}$ & $\text{3, 8}$ & $\star,$ \\
$0.317513$ & $\frac{1}{16}$ & $\text{1(2)}$ & $\text{1(4)}$ & $\text{5, 8}$ & $\star,$ \\
$0.3194$ & $\frac{1}{16}$ & $\text{1(2)}$ & $\text{1(4)}$ & $\text{3, 8}$ & $\star,$ \\
$0.321802$ & $\frac{1}{16}$ & $\text{1(2)}$ & $\text{1(4)}$ & $\text{3, 8}$ & $\star,$ \\
$0.325682$ & $\frac{1}{16}$ & $\text{1(2)}$ & $\text{1(4)}$ & $\text{3, 8}$ & $\star,$ \\
$0.331278$ & $\frac{1}{16}$ & $\text{1(2)}$ & $\text{1(4)}$ & $\text{5, 8}$ & $\star,$ \\
$0.333196$ & $\frac{1}{16}$ & $\text{1(2)}$ & $\text{1(4)}$ & $\text{3, 8}$ & $\star,$ \\
$0.336749$ & $\frac{1}{16}$ & $\text{1(2)}$ & $\text{1(4)}$ & $\text{3, 8}$ & $\star,$ \\
$0.337934$ & $\frac{1}{16}$ & $\text{1(2)}$ & $\text{1(4)}$ & $\text{3, 8}$ & $\star,$ \\
$0.341643$ & $\frac{1}{16}$ & $\text{1(2)}$ & $\text{1(4)}$ & $\text{3, 8}$ & $\star,$ \\
$0.344151$ & $\frac{1}{16}$ & $\text{1(2)}$ & $\text{1(4)}$ & $\text{5, 8}$ & $\star,$ \\
$0.345882$ & $\frac{1}{16}$ & $\text{1(2)}$ & $\text{1(4)}$ & $\text{3, 8}$ & $\star,$ \\
$0.352982$ & $\frac{1}{16}$ & $\text{1(2)}$ & $\text{1(4)}$ & $\text{3, 8}$ & $\star,$ \\
$0.359244$ & $\frac{1}{16}$ & $\text{1(2)}$ & $\text{1(4)}$ & $\text{3, 8}$ & $\star,$ \\
$0.361662$ & $\frac{1}{16}$ & $\text{1(2)}$ & $\text{1(4)}$ & $\text{3, 8}$ & $\star,$ \\
$0.365092$ & $\frac{1}{16}$ & $\text{1(2)}$ & $\text{1(4)}$ & $\text{3, 8}$ & $\star,$ \\
$0.372355$ & $\frac{1}{16}$ & $\text{1(2)}$ & $\text{1(4)}$ & $\text{3, 8}$ & $\star,$ \\
$0.374331$ & $\frac{1}{16}$ & $\text{1(2)}$ & $\text{1(4)}$ & $\text{3, 8}$ & $\star,$ \\
$0.375208$ & $\frac{1}{16}$ & $\text{1(2)}$ & $\text{1(4)}$ & $\text{1, 8}$ & $\text{}$ \\
$0.375619$ & $\frac{1}{16}$ & $\text{1(2)}$ & $\text{1(4)}$ & $\text{1, 8}$ & $\text{}$ \\
$0.375947$ & $\frac{1}{16}$ & $\text{1(2)}$ & $\text{1(4)}$ & $\text{2, 8}$ & $\star,$ \\
$0.377629$ & $\frac{1}{16}$ & $\text{1(2)}$ & $\text{1(4)}$ & $\text{2, 8}$ & $\star,$ \\
$0.378414$ & $\frac{1}{16}$ & $\text{1(2)}$ & $\text{1(4)}$ & $\text{1, 8}$ & $\text{}$ \\
$0.379566$ & $\frac{1}{16}$ & $\text{1(2)}$ & $\text{1(4)}$ & $\text{2, 8}$ & $\star,$ \\
$0.380593$ & $\frac{1}{16}$ & $\text{1(2)}$ & $\text{1(4)}$ & $\text{1, 8}$ & $\text{}$ \\
$0.381094$ & $\frac{1}{16}$ & $\text{1(2)}$ & $\text{1(4)}$ & $\text{1, 8}$ & $\text{}$ \\
$0.382421$ & $\frac{1}{16}$ & $\text{1(2)}$ & $\text{1(4)}$ & $\text{4, 8}$ & $\star,$ \\
$0.382862$ & $\frac{1}{16}$ & $\text{1(2)}$ & $\text{1(4)}$ & $\text{2, 8}$ & $\star,$ \\
$0.383032$ & $\frac{1}{16}$ & $\text{1(2)}$ & $\text{1(4)}$ & $\text{4, 8}$ & $\star,$ \\
$0.383406$ & $\frac{1}{16}$ & $\text{1(2)}$ & $\text{1(4)}$ & $\text{2, 8}$ & $\star,$ \\
$0.384387$ & $\frac{1}{16}$ & $\text{1(2)}$ & $\text{1(4)}$ & $\text{1, 8}$ & $\text{}$ \\
$0.384491$ & $\frac{1}{16}$ & $\text{1(2)}$ & $\text{1(4)}$ & $\text{3, 8}$ & $\text{}$ \\
$0.386203$ & $\frac{1}{16}$ & $\text{1(2)}$ & $\text{1(4)}$ & $\text{3, 8}$ & $\text{}$ \\
$0.386914$ & $\frac{1}{16}$ & $\text{1(2)}$ & $\text{1(4)}$ & $\text{1, 8}$ & $\text{}$ \\
$0.387962$ & $\frac{1}{16}$ & $\text{1(2)}$ & $\text{1(4)}$ & $\text{5, 8}$ & $\text{}$ \\
$0.389562$ & $\frac{1}{16}$ & $\text{1(2)}$ & $\text{1(4)}$ & $\text{4, 8}$ & $\text{}$ \\
$0.390925$ & $\frac{1}{16}$ & $\text{1(2)}$ & $\text{1(4)}$ & $\text{1, 8}$ & $\text{}$ \\
$0.391891$ & $\frac{1}{16}$ & $\text{1(2)}$ & $\text{1(4)}$ & $\text{4, 10}$ & $\text{}$ \\
$0.392975$ & $\frac{1}{16}$ & $\text{1(2)}$ & $\text{1(4)}$ & $\text{6, 8}$ & $\text{}$ \\
$0.39368$ & $\frac{1}{16}$ & $\text{1(2)}$ & $\text{1(4)}$ & $\text{1, 10}$ & $\text{}$ \\
$0.393713$ & $\frac{1}{16}$ & $\text{1(2)}$ & $\text{1(4)}$ & $\text{2, 8}$ & $\star,$ \\
$0.395524$ & $\frac{1}{16}$ & $\text{1(2)}$ & $\text{1(4)}$ & $\text{2, 8}$ & $\star,$ \\
$0.396745$ & $\frac{1}{16}$ & $\text{1(2)}$ & $\text{1(4)}$ & $\text{4, 8}$ & $\text{}$ \\
$0.396918$ & $\frac{1}{16}$ & $\text{1(2)}$ & $\text{1(4)}$ & $\text{4, 8}$ & $\text{}$ \\
$0.396983$ & $\frac{1}{16}$ & $\text{1(2)}$ & $\text{1(4)}$ & $\text{2, 8}$ & $\star,$ \\
$0.397389$ & $\frac{1}{16}$ & $\text{1(2)}$ & $\text{1(4)}$ & $\text{2, 8}$ & $\star,$ \\
$0.397804$ & $\frac{1}{16}$ & $\text{1(2)}$ & $\text{1(4)}$ & $\text{4, 8}$ & $\text{}$ \\
$0.398367$ & $\frac{1}{16}$ & $\text{1(2)}$ & $\text{1(4)}$ & $\text{4, 8}$ & $\text{}$ \\
$0.399186$ & $\frac{1}{16}$ & $\text{1(2)}$ & $\text{1(4)}$ & $\text{6, 8}$ & $\text{}$ \\
$0.402385$ & $\frac{1}{16}$ & $\text{1(2)}$ & $\text{1(4)}$ & $\text{2, 8}$ & $\star,$ \\
$0.403082$ & $\frac{1}{16}$ & $\text{1(2)}$ & $\text{1(4)}$ & $\text{4, 8}$ & $\text{}$ \\
$0.403357$ & $\frac{1}{16}$ & $\text{1(2)}$ & $\text{1(4)}$ & $\text{4, 8}$ & $\star,$ \\
$0.403859$ & $\frac{1}{16}$ & $\text{1(2)}$ & $\text{1(4)}$ & $\text{4, 8}$ & $\text{}$ \\
$0.403921$ & $\frac{1}{16}$ & $\text{1(2)}$ & $\text{1(4)}$ & $\text{2, 8}$ & $\star,$ \\
$0.405881$ & $\frac{1}{16}$ & $\text{1(2)}$ & $\text{1(4)}$ & $\text{2, 8}$ & $\star,$ \\
$0.407596$ & $\frac{1}{16}$ & $\text{1(2)}$ & $\text{1(4)}$ & $\text{2, 8}$ & $\star,$ \\
$0.414038$ & $\frac{1}{16}$ & $\text{1(2)}$ & $\text{1(4)}$ & $\text{4, 8}$ & $\text{}$ \\
$0.414147$ & $\frac{1}{16}$ & $\text{1(2)}$ & $\text{1(4)}$ & $\text{2, 8}$ & $\star,$ \\
$0.41774$ & $\frac{1}{16}$ & $\text{1(2)}$ & $\text{1(4)}$ & $\text{4, 8}$ & $\text{}$ \\
$0.418276$ & $\frac{1}{16}$ & $\text{1(2)}$ & $\text{1(4)}$ & $\text{2, 8}$ & $\star,$ \\
$0.419079$ & $\frac{1}{16}$ & $\text{1(2)}$ & $\text{1(4)}$ & $\text{4, 8}$ & $\star,$ \\
$0.419573$ & $\frac{1}{16}$ & $\text{1(2)}$ & $\text{1(4)}$ & $\text{2, 8}$ & $\star,$ \\
$0.420037$ & $\frac{1}{16}$ & $\text{1(2)}$ & $\text{1(4)}$ & $\text{2, 8}$ & $\star,$ \\
$0.420698$ & $\frac{1}{16}$ & $\text{1(2)}$ & $\text{1(4)}$ & $\text{4, 8}$ & $\star,$ \\
$0.421879$ & $\frac{1}{16}$ & $\text{1(2)}$ & $\text{1(4)}$ & $\text{2, 8}$ & $\star,$ \\
$0.422539$ & $\frac{1}{16}$ & $\text{1(2)}$ & $\text{1(4)}$ & $\text{2, 8}$ & $\star,$ \\
$0.422564$ & $\frac{1}{16}$ & $\text{1(2)}$ & $\text{1(4)}$ & $\text{2, 8}$ & $\star,$ \\
$0.423244$ & $\frac{1}{16}$ & $\text{1(2)}$ & $\text{1(4)}$ & $\text{4, 8}$ & $\text{}$ \\
$0.423484$ & $\frac{1}{16}$ & $\text{1(2)}$ & $\text{1(4)}$ & $\text{4, 8}$ & $\text{}$ \\
$0.423728$ & $\frac{1}{16}$ & $\text{1(2)}$ & $\text{1(4)}$ & $\text{1, 9}$ & $\star,$ \\
$0.423912$ & $\frac{1}{16}$ & $\text{1(2)}$ & $\text{1(4)}$ & $\text{4, 8}$ & $\text{}$ \\
$0.424515$ & $\frac{1}{16}$ & $\text{1(2)}$ & $\text{1(4)}$ & $\text{4, 8}$ & $\text{}$ \\
$0.424672$ & $\frac{1}{16}$ & $\text{1(2)}$ & $\text{1(4)}$ & $\text{4, 8}$ & $\text{}$ \\
$0.425003$ & $\frac{1}{16}$ & $\text{1(2)}$ & $\text{1(4)}$ & $\text{1, 8}$ & $\star,$ \\
$0.42551$ & $\frac{1}{16}$ & $\text{1(2)}$ & $\text{1(4)}$ & $\text{1, 8}$ & $\star,$ \\
$0.426232$ & $\frac{1}{16}$ & $\text{1(2)}$ & $\text{1(4)}$ & $\text{1, 8}$ & $\star,$ \\
$0.426519$ & $\frac{1}{16}$ & $\text{1(2)}$ & $\text{1(4)}$ & $\text{4, 8}$ & $\text{}$ \\
$0.427174$ & $\frac{1}{16}$ & $\text{1(2)}$ & $\text{1(4)}$ & $\text{4, 8}$ & $\text{}$ \\
$0.427916$ & $\frac{1}{16}$ & $\text{1(2)}$ & $\text{1(4)}$ & $\text{6, 8}$ & $\text{}$ \\
$0.429833$ & $\frac{1}{16}$ & $\text{1(2)}$ & $\text{1(4)}$ & $\text{4, 8}$ & $\text{}$ \\
$0.430072$ & $\frac{1}{16}$ & $\text{1(2)}$ & $\text{1(4)}$ & $\text{2, 8}$ & $\text{}$ \\
$0.430612$ & $\frac{1}{16}$ & $\text{1(2)}$ & $\text{1(4)}$ & $\text{3, 8}$ & $\text{}$ \\
$0.432475$ & $\frac{1}{16}$ & $\text{1(2)}$ & $\text{1(4)}$ & $\text{2, 9}$ & $\text{}$ \\
$0.432602$ & $\frac{1}{16}$ & $\text{1(2)}$ & $\text{1(4)}$ & $\text{5, 8}$ & $\text{}$ \\
$0.432763$ & $\frac{1}{16}$ & $\text{1(2)}$ & $\text{1(4)}$ & $\text{5, 8}$ & $\text{}$ \\
$0.432779$ & $\frac{1}{16}$ & $\text{1(2)}$ & $\text{1(4)}$ & $\text{5, 8}$ & $\text{}$ \\
$0.433177$ & $\frac{1}{16}$ & $\text{1(2)}$ & $\text{1(4)}$ & $\text{3, 8}$ & $\text{}$ \\
$0.433336$ & $\frac{1}{16}$ & $\text{1(2)}$ & $\text{1(4)}$ & $\text{3, 8}$ & $\text{}$ \\
$0.433375$ & $\frac{1}{16}$ & $\text{1(2)}$ & $\text{1(4)}$ & $\text{3, 8}$ & $\text{}$ \\
$0.433745$ & $\frac{1}{16}$ & $\text{1(2)}$ & $\text{1(4)}$ & $\text{3, 8}$ & $\text{}$ \\
$0.433977$ & $\frac{1}{16}$ & $\text{1(2)}$ & $\text{1(4)}$ & $\text{3, 8}$ & $\text{}$ \\
$0.434047$ & $\frac{1}{16}$ & $\text{1(2)}$ & $\text{1(4)}$ & $\text{5, 8}$ & $\text{}$ \\
$0.434078$ & $\frac{1}{16}$ & $\text{1(2)}$ & $\text{1(4)}$ & $\text{3, 8}$ & $\text{}$ \\
$0.434171$ & $\frac{1}{16}$ & $\text{1(2)}$ & $\text{1(4)}$ & $\text{3, 8}$ & $\text{}$ \\
$0.434717$ & $\frac{1}{16}$ & $\text{1(2)}$ & $\text{1(4)}$ & $\text{3, 8}$ & $\text{}$ \\
$0.434806$ & $\frac{1}{16}$ & $\text{1(2)}$ & $\text{1(4)}$ & $\text{3, 8}$ & $\text{}$ \\
$0.435204$ & $\frac{1}{16}$ & $\text{1(2)}$ & $\text{1(4)}$ & $\text{3, 8}$ & $\text{}$ \\
$0.435966$ & $\frac{1}{16}$ & $\text{1(2)}$ & $\text{1(4)}$ & $\text{5, 8}$ & $\text{}$ \\
$0.436023$ & $\frac{1}{16}$ & $\text{1(2)}$ & $\text{1(4)}$ & $\text{5, 8}$ & $\text{}$ \\
$0.4362$ & $\frac{1}{16}$ & $\text{1(2)}$ & $\text{1(4)}$ & $\text{3, 8}$ & $\text{}$ \\
$0.436251$ & $\frac{1}{16}$ & $\text{1(2)}$ & $\text{1(4)}$ & $\text{3, 8}$ & $\text{}$ \\
$0.436473$ & $\frac{1}{16}$ & $\text{1(2)}$ & $\text{1(4)}$ & $\text{3, 8}$ & $\text{}$ \\
$0.43654$ & $\frac{1}{16}$ & $\text{1(2)}$ & $\text{1(4)}$ & $\text{3, 8}$ & $\text{}$ \\
$0.436762$ & $\frac{1}{16}$ & $\text{1(2)}$ & $\text{1(4)}$ & $\text{3, 8}$ & $\text{}$ \\
$0.437164$ & $\frac{1}{16}$ & $\text{1(2)}$ & $\text{1(4)}$ & $\text{3, 8}$ & $\text{}$ \\
$0.437292$ & $\frac{1}{16}$ & $\text{1(2)}$ & $\text{1(4)}$ & $\text{5, 8}$ & $\text{}$ \\
$0.437349$ & $\frac{1}{16}$ & $\text{1(2)}$ & $\text{1(4)}$ & $\text{5, 8}$ & $\text{}$ \\
$0.437493$ & $\frac{1}{16}$ & $\text{1(2)}$ & $\text{1(4)}$ & $\text{3, 8}$ & $\text{}$ \\
$0.269845$ & $\frac{28410}{45292}$ & $\text{2(1,1)}$ & $\text{4(1,1,1,1)}$ & $\text{7, 7}$ & $\star,\,\star$ \\
$0.274768$ & $\frac{28410}{45292}$ & $\text{2(1,1)}$ & $\text{4(1,1,1,1)}$ & $\text{7, 7}$ & $\star,\,\star$ \\
$0.278967$ & $\frac{28410}{45292}$ & $\text{2(1,1)}$ & $\text{4(1,1,1,1)}$ & $\text{8, 7}$ & $\star,\,\star$ \\
$0.318579$ & $\frac{28410}{45292}$ & $\text{2(1,1)}$ & $\text{4(1,1,1,1)}$ & $\text{6, 7}$ & $\star,\,\star$ \\
$0.325431$ & $\frac{28410}{45292}$ & $\text{2(1,1)}$ & $\text{4(1,1,1,1)}$ & $\text{6, 7}$ & $\star,\,\star$ \\
$0.397534$ & $\frac{28410}{45292}$ & $\text{2(1,1)}$ & $\text{4(1,1,1,1)}$ & $\text{7, 7}$ & $\star,\,\star$ \\
$0.420142$ & $\frac{28410}{45292}$ & $\text{2(1,1)}$ & $\text{4(1,1,1,1)}$ & $\text{5, 7}$ & $\star,\,\star$ \\
$0.435306$ & $\frac{28410}{45292}$ & $\text{2(1,1)}$ & $\text{4(1,1,1,1)}$ & $\text{6, 7}$ & $\star,\,\star$ \\
$0.135737$ & $\frac{12068}{191535}$ & $\text{2(1,1)}$ & $\text{4(1,1,1,1)}$ & $\text{10, 7}$ & $\star,\,\star$ \\
$0.272473$ & $\frac{12068}{191535}$ & $\text{2(1,1)}$ & $\text{4(1,1,1,1)}$ & $\text{7, 7}$ & $\star,\,\star$ \\
$0.366059$ & $\frac{12068}{191535}$ & $\text{2(1,1)}$ & $\text{4(1,1,1,1)}$ & $\text{5, 7}$ & $\star,\,\star$ \\
$0.377273$ & $\frac{12068}{191535}$ & $\text{2(1,1)}$ & $\text{4(1,1,1,1)}$ & $\text{6, 7}$ & $\star,\,\star$ \\
$0.381222$ & $\frac{12068}{191535}$ & $\text{2(1,1)}$ & $\text{4(1,1,1,1)}$ & $\text{7, 7}$ & $\star,\,\star$ \\
$0.432955$ & $\frac{12068}{191535}$ & $\text{2(1,1)}$ & $\text{4(1,1,1,1)}$ & $\text{7, 7}$ & $\star,\,\star$ \\
$0.278466$ & $\frac{1429}{22201}$ & $\text{2(1,1)}$ & $\text{2(2,2)}$ & $\text{9, 7}$ & $\star,\,\star$ \\
$0.399211$ & $\frac{1429}{22201}$ & $\text{2(1,1)}$ & $\text{2(2,2)}$ & $\text{8, 7}$ & $\star,\,\star$ \\
$0.399211$ & $\frac{1429}{22201}$ & $\text{2(1,1)}$ & $\text{2(2,2)}$ & $\text{8, 6}$ & $\star,\,\star$ \\
$0.40947$ & $\frac{1429}{22201}$ & $\text{2(1,1)}$ & $\text{2(2,2)}$ & $\text{6, 7}$ & $\star,\,\star$ \\
$0.168508$ & $\frac{14307}{180625}$ & $\text{2(1,1)}$ & $\text{4(1,1,1,1)}$ & $\text{8, 7}$ & $\star,\,\star$ \\
$0.296289$ & $\frac{14307}{180625}$ & $\text{2(1,1)}$ & $\text{4(1,1,1,1)}$ & $\text{6, 7}$ & $\star,\,\star$ \\
$0.312871$ & $\frac{14307}{180625}$ & $\text{2(1,1)}$ & $\text{4(1,1,1,1)}$ & $\text{7, 7}$ & $\star,\,\star$ \\
$0.316749$ & $\frac{14307}{180625}$ & $\text{2(1,1)}$ & $\text{4(1,1,1,1)}$ & $\text{6, 7}$ & $\star,\,\star$ \\
$0.382194$ & $\frac{14307}{180625}$ & $\text{2(1,1)}$ & $\text{4(1,1,1,1)}$ & $\text{5, 7}$ & $\star,\,\star$ \\
$0.404715$ & $\frac{14307}{180625}$ & $\text{2(1,1)}$ & $\text{4(1,1,1,1)}$ & $\text{3, 7}$ & $\star,\,\star$ \\
$0.420064$ & $\frac{14307}{180625}$ & $\text{2(1,1)}$ & $\text{4(1,1,1,1)}$ & $\text{4, 7}$ & $\star,\,\star$ \\
$0.42027$ & $\frac{14307}{180625}$ & $\text{2(1,1)}$ & $\text{4(1,1,1,1)}$ & $\text{5, 7}$ & $\star,\,\star$ \\
$0.476291$ & $\frac{14307}{180625}$ & $\text{2(1,1)}$ & $\text{4(1,1,1,1)}$ & $\text{6, 7}$ & $\star,\,\star$ \\
$0.483238$ & $\frac{14307}{180625}$ & $\text{2(1,1)}$ & $\text{4(1,1,1,1)}$ & $\text{5, 7}$ & $\star,\,\star$ \\
$0.440034$ & $\frac{1}{5}$ & $\text{1(2)}$ & $\text{2(2,2)}$ & $\text{5, 7}$ & $\text{}$ \\
$0.579461$ & $\frac{1}{5}$ & $\text{1(2)}$ & $\text{2(2,2)}$ & $\text{0, 7}$ & $\text{}$ \\
$0.598295$ & $\frac{1}{5}$ & $\text{1(2)}$ & $\text{2(2,2)}$ & $\text{3, 7}$ & $\star,$ \\
$0.598295$ & $\frac{1}{5}$ & $\text{1(2)}$ & $\text{2(2,2)}$ & $\text{3, 7}$ & $\star,$ \\
$0.613931$ & $\frac{1}{5}$ & $\text{1(2)}$ & $\text{2(2,2)}$ & $\text{4, 7}$ & $\star,$ \\
$0.655038$ & $\frac{1}{5}$ & $\text{1(2)}$ & $\text{2(2,2)}$ & $\text{2, 7}$ & $\star,$ \\
$0.679402$ & $\frac{1}{5}$ & $\text{1(2)}$ & $\text{2(2,2)}$ & $\text{1, 7}$ & $\star,$ \\
$0.845888$ & $\frac{1}{5}$ & $\text{1(2)}$ & $\text{2(2,2)}$ & $\text{3, 7}$ & $\text{}$ \\
$0.849994$ & $\frac{1}{5}$ & $\text{1(2)}$ & $\text{2(2,2)}$ & $\text{2, 7}$ & $\text{}$ \\
$0.435233$ & $\frac{157}{784}$ & $\text{2(1,1)}$ & $\text{2(2,2)}$ & $\text{7, 7}$ & $\star,\,\star$ \\
$0.435233$ & $\frac{157}{784}$ & $\text{2(1,1)}$ & $\text{2(2,2)}$ & $\text{7, 6}$ & $\star,\,\star$ \\
$0.561397$ & $\frac{157}{784}$ & $\text{2(1,1)}$ & $\text{2(2,2)}$ & $\text{5, 7}$ & $\star,\,\star$ \\
$0.561397$ & $\frac{157}{784}$ & $\text{2(1,1)}$ & $\text{2(2,2)}$ & $\text{5, 6}$ & $\star,\,\star$ \\
$0.582485$ & $\frac{157}{784}$ & $\text{2(1,1)}$ & $\text{2(2,2)}$ & $\text{2, 6}$ & $\star,\,\star$ \\
$0.614584$ & $\frac{157}{784}$ & $\text{2(1,1)}$ & $\text{2(2,2)}$ & $\text{6, 7}$ & $\star,\,\star$ \\
$0.614584$ & $\frac{157}{784}$ & $\text{2(1,1)}$ & $\text{2(2,2)}$ & $\text{6, 6}$ & $\star,\,\star$ \\
$0.615509$ & $\frac{157}{784}$ & $\text{2(1,1)}$ & $\text{2(2,2)}$ & $\text{4, 7}$ & $\star,\,\star$ \\
$0.615509$ & $\frac{157}{784}$ & $\text{2(1,1)}$ & $\text{2(2,2)}$ & $\text{4, 6}$ & $\star,\,\star$ \\
$0.632019$ & $\frac{157}{784}$ & $\text{2(1,1)}$ & $\text{2(2,2)}$ & $\text{5, 7}$ & $\star,\,\star$ \\
$0.632019$ & $\frac{157}{784}$ & $\text{2(1,1)}$ & $\text{2(2,2)}$ & $\text{5, 6}$ & $\star,\,\star$ \\
$0.672395$ & $\frac{157}{784}$ & $\text{2(1,1)}$ & $\text{2(2,2)}$ & $\text{3, 7}$ & $\star,\,\star$ \\
$0.672395$ & $\frac{157}{784}$ & $\text{2(1,1)}$ & $\text{2(2,2)}$ & $\text{3, 7}$ & $\star,\,\star$ \\
$0.672395$ & $\frac{157}{784}$ & $\text{2(1,1)}$ & $\text{2(2,2)}$ & $\text{3, 6}$ & $\star,\,\star$ \\
$0.700306$ & $\frac{157}{784}$ & $\text{2(1,1)}$ & $\text{2(2,2)}$ & $\text{4, 7}$ & $\star,\,\star$ \\
$0.700306$ & $\frac{157}{784}$ & $\text{2(1,1)}$ & $\text{2(2,2)}$ & $\text{4, 6}$ & $\star,\,\star$ \\
$0.845775$ & $\frac{157}{784}$ & $\text{2(1,1)}$ & $\text{2(2,2)}$ & $\text{5, 7}$ & $\star,\,\star$ \\
$0.845775$ & $\frac{157}{784}$ & $\text{2(1,1)}$ & $\text{2(2,2)}$ & $\text{5, 6}$ & $\star,\,\star$ \\
$0.84714$ & $\frac{157}{784}$ & $\text{2(1,1)}$ & $\text{2(2,2)}$ & $\text{4, 7}$ & $\star,\,\star$ \\
$0.84714$ & $\frac{157}{784}$ & $\text{2(1,1)}$ & $\text{2(2,2)}$ & $\text{4, 6}$ & $\star,\,\star$ \\
$0.440354$ & $\frac{112141}{559504}$ & $\text{2(1,1)}$ & $\text{3(2,1,1)}$ & $\text{6, 7}$ & $\star,\,\star$ \\
$0.580188$ & $\frac{112141}{559504}$ & $\text{2(1,1)}$ & $\text{3(2,1,1)}$ & $\text{1, 7}$ & $\star,\,\star$ \\
$0.599872$ & $\frac{112141}{559504}$ & $\text{2(1,1)}$ & $\text{3(2,1,1)}$ & $\text{4, 7}$ & $\star,\,\star$ \\
$0.607946$ & $\frac{112141}{559504}$ & $\text{2(1,1)}$ & $\text{3(2,1,1)}$ & $\text{5, 7}$ & $\star,\,\star$ \\
$0.657683$ & $\frac{112141}{559504}$ & $\text{2(1,1)}$ & $\text{3(2,1,1)}$ & $\text{3, 7}$ & $\star,\,\star$ \\
$0.671595$ & $\frac{112141}{559504}$ & $\text{2(1,1)}$ & $\text{3(2,1,1)}$ & $\text{2, 7}$ & $\star,\,\star$ \\
$0.847494$ & $\frac{112141}{559504}$ & $\text{2(1,1)}$ & $\text{3(2,1,1)}$ & $\text{4, 7}$ & $\star,\,\star$ \\
$0.85121$ & $\frac{112141}{559504}$ & $\text{2(1,1)}$ & $\text{3(2,1,1)}$ & $\text{3, 7}$ & $\star,\,\star$ \\
$\frac{3}{5}$ & $\frac{1}{4}$ & $\text{1(2)}$ & $\text{1(4)}$ & $\text{5, 5}$ & $\text{}$ \\
$\frac{3}{5}$ & $\frac{1}{4}$ & $\text{1(2)}$ & $\text{1(4)}$ & $\text{0, 5}$ & $\text{}$ \\
$\frac{3}{4}$ & $\frac{1}{4}$ & $\text{1(2)}$ & $\text{1(4)}$ & $\text{3, 6}$ & $\star,$ \\
$\frac{3}{4}$ & $\frac{1}{4}$ & $\text{1(2)}$ & $\text{1(4)}$ & $\text{1, 6}$ & $\star,$ \\
\end{longtable}
\end{center}

\begin{center}
\begin{longtable}{|c c c c c c|}\caption{Fixed points found for $N_s=3$ and $N_f=1$.} \\
\hline
$S$ & $Y$ & \makecell{$\#$ different \\[-3pt] $\gamma_\phi$(degeneracies)} & \makecell{$\#$ different \\[-3pt] $\gamma_\psi$(degeneracies)} & $\#\,\kappa<0$, =0 & $a_2,\,b_1\neq0$ \\ [0.5ex] 
 \hline
\endfirsthead

\multicolumn{6}{c}%
{{\bfseries \tablename\ \thetable{} -- continued from previous page}} \\
\hline
$S$ & $Y$ & \makecell{$\#$ different \\[-3pt] $\gamma_\phi$(degeneracies)} & \makecell{$\#$ different \\[-3pt] $\gamma_\psi$(degeneracies)} & $\#\,\kappa<0$, =0 & $a_2,\,b_1\neq0$  \\ [0.5ex] 
 \hline
\endhead

\hline \multicolumn{6}{|r|}{{Continued on next page}} \\ \hline
\endfoot

\hline
\endlastfoot \label{tab:31}
$\!\!0.247862$ & $\frac{1}{49}$ & $\text{2(1,2)}$ & $\text{1(1)}$ & $\text{11, 3}$ & $\star,\,\star$ \\
$0.248061$ & $\frac{1}{49}$ & $\text{2(1,2)}$ & $\text{1(1)}$ & $\text{11, 3}$ & $\star,\,\star$ \\
$0.269318$ & $\frac{1}{49}$ & $\text{2(1,2)}$ & $\text{1(1)}$ & $\text{9, 3}$ & $\star,\,\star$ \\
$0.286823$ & $\frac{1}{49}$ & $\text{2(1,2)}$ & $\text{1(1)}$ & $\text{9, 2}$ & $\star,\,\star$ \\
$0.346219$ & $\frac{1}{49}$ & $\text{2(1,2)}$ & $\text{1(1)}$ & $\text{6, 3}$ & $\star,\,\star$ \\
$0.40898$ & $\frac{1}{49}$ & $\text{2(1,2)}$ & $\text{1(1)}$ & $\text{4, 3}$ & $\star,\,\star$ \\
$0.423905$ & $\frac{1}{49}$ & $\text{2(1,2)}$ & $\text{1(1)}$ & $\text{4, 2}$ & $\star,\,\star$ \\
\end{longtable}
\end{center}

\end{appendices}

\bibliography{main}

\end{document}